%% file: main.tex
\documentclass[a4paper,UKenglish, autoref, thm-restate]{lmcs} 
\pdfoutput=1

\usepackage[utf8]{inputenc}
\usepackage{amssymb}
\usepackage[super]{nth}
\usepackage{amsmath}
\usepackage{graphics}
\usepackage{graphicx}
\usepackage{amsthm}
\usepackage{subcaption}
\usepackage{float}
\usepackage{afterpage}

\usepackage{mathtools}
\usepackage{hyperref}
\usepackage{virginialake}
\vlnostructuressyntax
\usepackage{tikz}
\usepackage{tikz-cd} 
\usetikzlibrary{arrows}
\usetikzlibrary{decorations.markings}
\usetikzlibrary{decorations.pathmorphing}
\usetikzlibrary{shapes.geometric}
\usetikzlibrary{positioning}
\usetikzlibrary{tikzmark}

\usepackage{mathdots}

\usepackage{cmll} 
\usepackage{cleveref}
\usepackage{url}
\usepackage[textsize=footnotesize, 
]{todonotes}
\usepackage{thmtools,thm-restate}
\usepackage{subfloat}

\usepackage{adjustbox}
\usepackage{xspace}
\usepackage[notxmath]{proofzilla}

\usepackage{afterpage}

\input{macros}

\begin{document}

\title[Non-wellfounded parsimonious proofs  and  non-uniform complexity]{Non-wellfounded parsimonious proofs \\ and  non-uniform complexity}

\author[M. Acclavio]{Matteo Acclavio}
 \address{University of Sussex}
 \email{macclavio@gmail.com}

\author[G. Curzi]{Gianluca Curzi}
\address{University of Gothenburg}
\email{gianluca.curzi@gu.se}

\author[G. Guerrieri]{Giulio Guerrieri}
\address{University of Sussex}
\email{g.guerrieri@sussex.ac.uk}

\begin{abstract}
 
 In this paper we investigate the complexity-theoretical aspects of cyclic and non-wellfounded proofs in the context of parsimonious logic, a variant of linear logic where the exponential modality ! is interpreted as a constructor for streams over finite data. We present non-wellfounded parsimonious proof systems capturing the classes $\fptime$ and $\fppoly$. Soundness is established via a polynomial modulus of continuity for continuous cut-elimination. Completeness relies on an encoding of polynomial Turing machines~with~advice within a type assignment system based on parsimonious logic. 
 
 As a byproduct of our proof methods, we establish a series of characterisation results for various finitary proof systems.
\end{abstract}

\maketitle

\section{Introduction}

In its modern guise, \emph{non-wellfounded proof theory} emerged for the first time in the context of  the modal $\mu$-calculus~\cite{niwinski1996games,dax2006proof}. 
Since then,  this area of  proof theory has provided  a promising theoretical framework for studying  least and greatest fixed points,  hence for reasoning about induction and coinduction. What is more, its applications  have spanned, over the years,   a number of rather diverse topics, such as     predicate logic~\cite{brotherston2011sequent, BerardiT19},  algebras~\cite{das2017cut,DP18,das:de:LICS24,das:de:IJCAR24}, arithmetic~\cite{Simpson17, BerardiT17, das2018logical},  proofs-as-programs  interpretations~\cite{BaeldeDS16,DeS19,Das2021-preprint,Kuperberg-Pous21, Das2021}, modal logics and $\mu$-calculi~\cite{niwinski1996games, dax2006proof, SprengerD03-bis, SprengerD03, DekkerKMV23, acc:mon:per:PDL,PDL:interpolation,conversePDL:interpolation}, and continuous cut elimination~\cite{mints1978finite, fortier2013cuts, Saurin23}.
	
Non-wellfounded proof theory studies proofs whose underlying tree structure is possibly infinite  (but finitely branching), where logical  consistency is guaranteed  by appropriate  global proof-theoretic conditions, called \emph{progressing} 
\emph{criteria}.  Within this framework, the so-called \emph{regular} proofs represent a major focus of interest. These are special non-wellfounded proofs having only finitely many distinct subproofs, and admit a finite description  in terms of cyclic directed graphs. Because of their graph-theoretic representation, these proofs are also called  \emph{circular} or \emph{cyclic}.

In~\cite{Das2021-preprint,Kuperberg-Pous21, Das2021} non-wellfounded proof-theory has been investigated from the perspective of the \emph{Curry-Howard correspondence}, where proofs are interpreted as (functional) programs, and program execution is given in terms of cut elimination.  Non-wellfounded proofs can be understood as programs defined by a possibly infinite list of instructions, where the progressing criterion ensures \emph{totality}, i.e., that those programs  are always well-defined on all arguments. 
On the other hand, the regularity condition on proof trees has a natural counterpart in the notion of \emph{computational uniformity}, meaning that programs defined by regular proofs can always be described  by a \emph{finite} sets of machine instructions.

	In \cite{CurziDas} the second author and Das extended the computational reading of non-wellfounded proofs  to the realm of \emph{computational complexity}, introducing circular proof systems capturing  the class of functions computable in polynomial time ($\fptime$) and the elementary functions ($\felementary$). These proof systems  are  based on Bellantoni and Cook's algebra of functions for safe recursion~\cite{BellantoniCook}, and are defined by identifying global conditions on circular progressing proofs motivated by ideas from \emph{Implicit Computational Complexity} (ICC), which studies machine-free characterisations of complexity classes that do not rely on explicit bounds on resources. This paper substantially launched a new topic in ICC called \emph{Cyclic Implicit Complexity} (CIC).

The results in~\cite{CurziDas} have been generalized by the same authors in \cite{Curzi023} 
 to capture  the class of functions computable in polynomial time by Turing machines with access to \emph{polynomial advice} ($\fppoly$) or, equivalently, computable by non-uniform families of polynomial-size circuits (see, e.g., \cite{arora_barak_2009}).  
Specifically, \emph{non-uniform complexity} is modeled by more permissive non-wellfounded proof systems (compared to circular proof systems), obtained by weakening the regularity condition, hence relaxing finite presentability of proofs\footnote{Note that {$\fppoly$} includes \emph{undecidable problems}, and so cannot be characterised by purely circular proof systems, which typically represent only computable functions.}.

In this paper, we take an alternative route to  CIC   based on \emph{linear logic}~\cite{LLL1}. 
Linear logic ($\mathsf{LL}$) is a refinement of both classical and intuitionistic logic that allows a better control over computational resources thanks to the so-called  \emph{exponential modalities} (denoted by $\oc$ and $\wn$), which mark the distinction between those assumptions that can be used linearly (that is, exactly once), and the ones that are reusable at will. According to  the Curry-Howard reading of linear logic, these modalities introduce non-linearity in functional programs: a proof of the linear implication $\oc A \limp B$ is interpreted as a program returning an output of type $B$ using an arbitrary number of times (including zero times) an input of type $A$.

Linear logic has inspired a variety of methods for taming complexity. The central idea  is to  weaken  the exponential rules for inducing a bound on cut elimination, which reduces  the computational strength of the system. These restricted systems of linear logic are  called ``light logics''.  Examples are \emph{soft linear logic}~\cite{LLL4} or \emph{light linear logic}~\cite{LLL1} for $\fptime$, and \emph{elementary linear logic}~\cite{LLL2, LLL3} for $\felementary$. 
Light logics  are typically  endowed with second-order quantifiers, which allow for a direct encoding of (resource-bounded) Turing machines, the crucial step for proving completeness w.r.t.~a complexity class.

Continuing this tradition, in a series of papers~\cite{Mazza14, Mazza15,MazzaT15}  Mazza introduced \emph{parsimonious logic} ($\mathsf{PL}$),  a 
type system based on a
variant of linear logic (defined in a type-theoretic fashion) where the exponential modality $\oc$ satisfies Milner's law (i.e., $\oc A \multimapboth A \otimes \oc A$) and  invalidates   the implications $\oc A \multimap \oc \oc A$ (\emph{digging}) and $\oc A \multimap \oc A \otimes \oc A$ (\emph{contraction}).  
In parsimonious logic, a proof of  $\oc A$ can be interpreted  as a  \emph{stream} over  proofs of  $A$, i.e., as a greatest fixed point.   
The linear implications $A \otimes \oc A \multimap \oc A$ (\emph{co-absorption}) and $\oc A \multimap A \otimes \oc A$ (\emph{absorption}), which form the two directions of Milner's law,  can be  read computationally as the \emph{push} and \emph{pop} operations on streams. 
In particular, in~\cite{MazzaT15} Mazza and Terui presented \emph{non-uniform parsimonious logic} ($\mathsf{nuPL}$), an extension  of $\mathsf{PL}$  equipped with  an {infinitely branching rule} {$\nuprule$ (see \Cref{fig:exponantial-rules-functorial-non-uniform-conditional})} that constructs a stream $\nwpstream{\der_{f(0)}}{\der_{f(1)}}{\der_{f(n)}}$ of type  $\oc A$  from 	a  finite set of proofs $\der_1, \ldots, \der_n$ of $A$ and  a (possibly non-recursive) function $f \colon \mathbb{N}\to \set{1, \ldots, n}$.
	
The fundamental result of~\cite{MazzaT15} is that, when endowed with restricted second-order quantifiers, the logic $\mathsf{PL}$ (resp.~$\mathsf{nuPL}$) characterises the class $\ptime$ of problems decidable in polynomial time  (resp.~the class $\ppoly$ of problems decidable by polynomial size families of circuits)\footnote{
	Despite not  explicitly stated in~\cite{MazzaT15}, the characterisation of $\ptime$  is a direct byproduct.
}. 
On the one hand, the infinitely branching rule $\nuprule$ can be used to  encode   streams, hence to model Turing machines querying an \emph{advice}~\cite{arora_barak_2009};  on the other hand, the absence of digging and contraction induces a polynomial bound on normalisation.

The analysis of parsimonious logic conducted in~\cite{Mazza14, Mazza15,MazzaT15} reveals that fixed point-based definitions of the exponentials are better behaving when  digging and  contraction are discarded. However, these results rely on the co-absorption rule ({$\ocbrule$} in \Cref{fig:exponantial-rules-functorial-non-uniform-conditional}) which is not admissible in $\mathsf{LL}$. 
Proper subsystems of $\mathsf{LL}$ (free of the co-absorption rule) admitting a stream-based interpretation of the exponentials have been provided in our previous work~\cite{CSL}, where we have defined \emph{parsimonious linear logic} ($\mathsf{PLL}$) and its non-uniform version,  called  \emph{non-uniform parsimonious linear logic} ($\mathsf{nuPLL}$).
Furthermore, in this work we also recast $\mathsf{PLL}$ and $\mathsf{nuPLL}$ in a non-wellfounded framework $\nwpllprop$ by identifying appropriate global  conditions that duly reflect the proof-theoretic features of these systems.
As a result, we introduced \emph{regular parsimonious linear logic} ($\mathsf{rPLL}^\infty$), defined in terms of regular non-wellfounded proofs, and \emph{weakly regular parsimonious linear logic} ($\mathsf{wrPLL}^\infty$),  where regularity is relaxed to model non-uniform computation. 
The main contribution of~\cite{CSL} is a \emph{continuous cut elimination theorem} for $\mathsf{rPLL}^\infty$ and $\mathsf{wrPLL}^\infty$, obtained by applying a novel cut elimination technique in the \mbox{non-wellfounded setting}.

\paragraph{Contributions}
In this paper, we define 
second-order extensions (noted with the subscript $2\ell$) of the proof systems presented in~\cite{CSL}, and we establish the {characterisations} below:
\begin{itemize}
\item $\nupll$ (and $\dpll$) characterise $\fppoly$;
\item  $\cpll$ (and $\pll$)  characterise $\fptime$.
\end{itemize}

The interconnections between our  results are summarised in~\Cref{fig:grand-tour-diagram}. 
The key result for this characterisations is the polynomial modulus of continuity on cut elimination for $\nupll$ and $\cpll$,  (\Cref{prop:polynomial-moduli}), 
from which we infer that $\nupll$ is \emph{sound} for $\fppoly$, and that $\cpll$ is sound for $\fptime$~(\Cref{thm:soundness}.\ref{enum:soundness1}-\ref{enum:soundness2}).  
\emph{Completeness} (\Cref{thm:completeness-non-wellfounded} requires a series of intermediate steps. 
We first introduce the type system $\typestream$, which implements a form of stream-based computation. This system represents an alternative approach to the type systems for parsimonious logic  introduced by Mazza et al. in~\cite{Mazza14,Mazza15,MazzaT15}. 
Then, we describe an encoding of polynomial time Turing machines with (polynomial) advice within  $\typestream$, which allows us to prove  that $\typestream$ is complete for $\fppoly$  (\Cref{thm:completeness}.\ref{enum:completeness2}). 
This is done by adapting standard methods from~\cite{MairsonTerui, Ronchi-Gaboardi} to the setting of non-uniform  computation.   Thirdly, we define  a  translation from $\typestream$ to $\dpll$   (\Cref{thm:embedding}.\ref{enum:embedding2}).  
Finally, we show that computation over strings in $\dpll$ can be simulated within   $\nupll$  (\Cref{thm:simulation}.\ref{enum:simulation2}). 
A  similar completeness argument can be restated  for $\cpll$ (and $\pll$) by considering the type system $\pta$ 
(\Cref{thm:completeness}.\ref{enum:completeness1}, \Cref{thm:embedding}.\ref{enum:embedding1}, and \Cref{thm:simulation}.\ref{enum:simulation1}).

On a technical side, the present paper contributes to the previous literature on the topic in several ways:
\begin{itemize}
	\item \emph{Cyclic Implicit Complexity (CIC).} 
	The contribution of this paper advances the development of CIC and explores it for the first time  in the setting of linear logic.
	Compared to previous work on CIC, this paper adopts rather different techniques for achieving soundness and completeness results. 
	In particular, the proof of soundness introduces novel ideas to estimate moduli of continuity.

	\item \emph{Parsimonious logic.} This paper lifts to a non-wellfounded proof-theoretic setting the characterisation results for parsimonious logic presented by Mazza at al. in~\cite{Mazza14, Mazza15, MazzaT15}, improving certain aspects of the latter. 
	First, as already anticipated, our characterisation theorems do not require the presence of the \emph{co-absorption rule}, which is expressed by the formula $A \otimes \oc A\limp \oc A$. Computationally, this means that it is possible to encode polytime Turing machines (with polynomial advice) in our systems \emph{without}  the need of a ``push'' operation on streams. 	Secondly, we generalise the characterisation of classes of  problems given in~\cite{MazzaT15}  (i.e., $\ppoly$ and $\ptime$) to classes of functions (i.e., $\fppoly$ and $\fptime$).	Last, non-wellfounded proofs do not rely on  infinitely branching rules such as  $\nuprule$,  avoiding the introduction of the  constant growth-rate function  $f \colon \mathbb{N}\to \set{1, \ldots, n}$ and   making our  characterisations more ``implicit'', thus  closer to ICC.
\end{itemize}

\paragraph{Outline of the paper}
\Cref{sec:prelim} introduces some preliminary notions and results on linear logic, non-wellfounded proofs and  (non-uniform) complexity classes. 
\Cref{sec:parsimoniousLogic,sec:coder}  recall parsimonious linear logic and extend to a second-order setting the proof systems introduced in~\cite{CSL}. Specifically, in~\Cref{sec:parsimoniousLogic} we define the proof system $\pll$ and $\dpll$. In~\Cref{sec:coder} we consider their  non-wellfounded proof-theoretic versions  $\cpll$ and $\nupll$, and study cut elimination properties and simulation results relating the systems. 
\Cref{subsec:soundness,subsec:completeness} represent the main contribution of this paper. In~\Cref{subsec:soundness} we prove the soundness theorem, and in~\Cref{subsec:completeness} we define the type systems $\typestream$ and $\pta$ based on parsimonious logic and use them to establish  completeness.

\begin{figure}
	\begin{tikzcd}[column sep=1.2cm]
		&
		\cpll 
		\arrow[r, "\subseteq"]   
		\arrow[dl, bend right, swap, "\text{\Cref{thm:soundness}.\ref{enum:soundness1}}"]  
		& 
		\nupll  
		\arrow[dr, bend left, "\text{\Cref{thm:soundness}.\ref{enum:soundness2}}"] 
		& 
		\\ 
		\fptime 
		\arrow[dr, bend right, swap, "\text{\Cref{thm:completeness}.\ref{enum:completeness1}}"] 
		& 
		\pll 
		\arrow[r, "\subseteq"]  
		\arrow[u, "\text{\Cref{thm:simulation}.\ref{enum:simulation1}}"]   
		& 
		\dpll  
		\arrow[u, "\text{\Cref{thm:simulation}.\ref{enum:simulation2}}"]    
		&  
		\fppoly   
		\arrow[dl, bend left, "\text{\Cref{thm:completeness}.\ref{enum:completeness2}}"] 
		\\
		&	
		\pta 
		\arrow[r, "\subseteq"]\arrow[u, "\text{\Cref{thm:embedding}.\ref{enum:embedding1}}"]    
		& 
		\typestream 
		\arrow[u, "\text{\Cref{thm:embedding}.\ref{enum:embedding2}}"]     
		& 
	\end{tikzcd}
	\caption{Diagram {of the main results}.}
	\label{fig:grand-tour-diagram}
\end{figure}

\section{Preliminary notions}\label{sec:prelim}

In this section we recall some basic notions and notation from (non-wellfounded) proof theory and computational complexity.

\subsection{Derivations and coderivations}\label{subsec:coderivations}

We assume that the  reader is familiar with the syntax of sequent calculus, e.g. \cite{troelstra_schwichtenberg_2000}. 
Here we specify some conventions adopted to simplify the content of  this paper.

We consider (\defin{sequent}) \defin{rules} of the form
$\vlinf{\rrule}{}{\Gamma}{}$
or
$\vlinf{\rrule}{}{\Gamma}{\Gamma_1}$
or
$\vliinf{\rrule}{}{\Gamma}{\Gamma_1}{\Gamma_2}$,
and we refer 
to the sequents $\Gamma_1$ and $\Gamma_2$ as the \defin{premises}, 
and
to the sequent $\Gamma$ as the \defin{conclusion} of the rule $\rrule$.
{Following~\cite{BaeldeDS16},  we define \defin{sequents} as finite {sets of  pairwise disjoint formula occurrences}\footnote{{This can be done by associating with any formula occurrence an address, as done in~\cite{BaeldeDS16}. The benefit of this sequent calculus presentation  is that we can track formula occurrences along a branch of the prooftree while avoiding some technicalities involved in  the sequents-as-lists presentation.}}. In particular, when we refer to a formula in a sequent we always consider a \emph{specific occurrence} of it.
 }
\def\tree{\mathcal T}
\def\branch{\mathcal B}

\begin{definition}
	\label{def:coderivation}
	
	A (binary, possibly infinite) \defin{tree} $\tree$ is a subset of words in $\set{1,2}^*$
	that~contains the empty word $\epsilon$ (the \defin{root} of $\tree$)
	and is \emph{ordered-prefix-closed}
	(i.e., if $n \in \set{1,2}$ and $vn \in \mathcal{T}$, then $v \in \mathcal{T}$,
	and
	if moreover $v2 \in \mathcal{T}$, then $v1 \in \mathcal{T}$).
	The elements of $\tree$ are called \defin{nodes} and their \defin{height} is the length of the word.
	A \defin{child} of $v\in\tree$ is any $vn\in\tree$ with $n\in\set{1,2}$.
	The \defin{prefix order} is a partial order $\leq_\tree$ on $\tree$ defined by: for any $v, v' \in \tree$, $v \leq_\tree v'$ if $v' = vw$ for some $w \in \{1,2\}^*$. 
	A maximal element of $\leq_\tree$ is a \defin{leaf} of $\tree$. 
	A \defin{branch} of $\tree$ is a set $\branch \subseteq \tree$ such that $\epsilon \in \branch$ and if $w \in \branch$ is not a leaf of $\tree$ then $w$ has exactly one~child~in~$\branch$.
	
	A \defin{coderivation} over a set of rules $\Sys$
	is a labeling $\der$ of a tree $\tree$ by sequents such that
	if $v$ is a node of $\tree$ with children $v_1,\ldots, v_n$ (with $n\in\set{0,1,2}$), 
	then there is an occurrence of a rule $\rrule$ in $\Sys$ with conclusion the sequent $\der(v)$ and premises the sequents $\der(v_1),\ldots, \der(v_n)$.
	The \defin{height} of $\rrule$ in $\der$ is the height of $v \in \tree$ such that $\der(v)$ is the conclusion~of~$\rrule$.
	The \defin{conclusion} of $\der$ is the sequent $\der(\epsilon)$.
	If $v$ is a node of the tree, the \defin{sub-coderivation} of $\der$ rooted at $v$ is the coderivation $\der_v$ defined by $\der_v(w)=\der(vw)$.
	It is \defin{regular} if it has finitely many distinct sub-coderivations; it is \defin{non-wellfounded} if it labels an infinite tree, and 
	it is a \defin{derivation} (with \defin{size} $\size{\der} \in \Nset$) if it labels a finite tree (with $\size{\der}$ nodes).

\end{definition}

Regular coderivations ({aka} circular or cyclic) can be represented  as \emph{finite} directed (possibly cyclic) graphs: a cycle is created by linking the roots of two identical~subcoderivations.

\begin{definition} \label{defn:bars}
	Let $\der$ be a coderivation labeling a tree $\tree$.
	A \defin{bar} (resp.~\defin{prebar}) of $\der$ is a set $\Nodes \subseteq \tree$ where:
\begin{itemize}
	\item any  branch (resp.~infinite branch) of the tree $\tree$ underlying $\dD$ contains a node in $\Nodes$;
	\item any pair of nodes in $\Nodes$ are mutually incomparable with respect to the prefix order $\leq_\tree$.
\end{itemize}
\end{definition}

\subsection{(Non-uniform) complexity classes}\label{subsec:non-uniform-complexity-classes}

We recall that  $\fptime$ is the class of functions computable in polynomial time by a Turing machine, and that $\fppoly$ is the class of functions computable in polynomial time  by a Turing machine with access to a ``polynomial amount of advice'' (determined only by the length of the input). Formally:

\begin{definition}
	$\fppoly$ is the class of functions 
	$f(\vec x)$ for which, for all $n \in \mathbb{N}$, there is a string (called the \defin{advice}) $\alpha_{n}$ of length polynomial in $n$ and $f'(y,\vec x)\in \fptime$ such that
	\mbox{$f(\vec x) = f'(\alpha_{|\vec x|}, \vec x)$}.
\end{definition}

$\fppoly$ extends $\fptime$ and contains some incomputable functions, for instance the characteristic function of undecidable unary languages \cite[Example 6.4]{arora_barak_2009}.
The class $\fppoly$ can be also defined in terms of non-uniform families of circuits.

\begin{theorem}[\cite{arora_barak_2009}, Thm.~6.11]\label{prop:fppoly-circuits}
	A function $f$ is in $\fppoly$ iff there is polynomial-size familiy of circuits computing $f$.
\end{theorem}

Following~\cite{Curzi023}, we adopt a different presentation of $\fppoly$ that {eases} the proof of completeness. 

\begin{definition} 
    Let $\RR \coloneqq \{0,1\}^\mathbb{N} = \set{ r \colon \mathbb{N} \to \set{0,1}}$ and $\fptime(\RR)$ be the set of functions computable in polynomial time by a Turing machine with access to an oracle  $r \in \RR$.
	\footnote{{Note that the elements of $\RR$ may be identified with Boolean streams.}}
\end{definition}

\begin{restatable}[See, e.g.,~\cite{Curzi023}]{proposition}{PTIMER}
	\label{prop:fppoly=fptime(RR)}
	$\fppoly = \fptime(\RR)$.
\end{restatable}
\begin{proof}[Proof sketch]
	For the left-right inclusion, let {$p(n)$ be a polynomial and} $\mathbf C = (C_n)_{n<\omega}$ be a circuit family with each $C_n$ taking $n$ Boolean inputs and having size $<p(n)$. We need to show that the language computed by $\mathbf C$ is also computed in $\fptime(\RR)$.
	Let $c \in \RR$ be the function that, on inputs $x,y$ returns the $|y|$\textsuperscript{th} bit of $C_{|x|}$.
	Using this oracle we can compute $C_{|x|}$ by polynomially queries to $c$, and this may be evaluated as usual using a polynomial-time evaluator in $\fptime$. For the right-left inclusion, notice that a polynomial-time machine can only make polynomially many calls to oracles with inputs of only polynomial size. Thus, if $f \in \fptime (\RR)$ then there is some $p_f$ with $f\in \fptime (\RR^{<p_f})$, where $\RR^{<p_f}$ is the restriction of each $r \in \RR$ to only its first $p_f(|\vec x|)$ many bits.
	Now, since $f$ can only call a fixed number of oracles from $\RR$, we can collect these finitely many polynomial-length prefixes into a single advice string for computation in $\fppoly$.
\end{proof}

\section{Second-Order Parsimonious Linear Logic}\label{sec:parsimoniousLogic}

{In~\cite{acclavio2023infinitary}, we introduced \emph{parsimonious linear logic} ($\mathsf{PLL}$), a subsystem of linear logic inspired by parsimonious logic~\cite{Mazza14,Mazza15,MazzaT15}.  In $\mathsf{PLL}$ the usual {promotion} rule is replaced by \emph{functorial promotion}  $\fprule$, and the usual contraction and  dereliction  rules  by the {\emph{absorption}} rule  $\wnbrule$ (see \Cref{fig:exponantial-rules-functorial-non-uniform-conditional}).
As a consequence, the exponential modalities of $\mathsf{PLL}$ are weaker than the usual linear logical modalities. In particular,  the  \emph{digging} formula $\oc A \limp \oc \oc A$
and the \emph{contraction} formula $\oc A \limp \oc A \otimes \oc A$
are not provable in $\mathsf{PLL}$. 

As already remarked in~\cite{Mazza14,Mazza15,MazzaT15}, systems based on parsimonious logic like $\mathsf{PLL}$ admit a straightforward  computational reading based on  streams: on the one hand, functorial promotion can be interpreted as a constructor  $\oc (\_) : A \limp \oc A$, which takes  a program $M$ of type $A$  and returns the stream $\oc( M)\dfn (M, M, \ldots, M, \ldots)$; on the other hand,  the absorption rule can be seen as a \emph{pop} operation on streams $\mathsf{pop}: \oc A \limp A \otimes \oc A$, which extracts the first element off the stream, i.e., 
$\mathsf{pop} (\oc( M)) = M \otimes \oc ( M)$. This step of computation is reflected by  the  cut elimination rule governing the interaction between $\fprule$ and $\wnbrule$ (see~\Cref{fig:exponential-functorial-ce}).

The computational interpretation outlined above can be pushed further by considering \emph{non-uniform parsimonious linear logic}  $\mathsf{nuPLL}$, which replaces $\fprule$ with {its} infinitely branching version {$\nuprule$ in \Cref{fig:exponantial-rules-functorial-non-uniform-conditional}}, called \emph{non-uniform promotion}~\cite{acclavio2023infinitary,MazzaT15, Mazza15}, 
which allows the construction of  general streams of the form $(M_1, M_2, \ldots, M_n, \ldots)$, where the side condition ensures that streams are over {finitely many} data.

Motivated by our complexity-theoretic goals, we present $\pll$ and $\dpll$,  second-order versions of  $\mathsf{PLL}$ and $\mathsf{nuPLL}$  where instantiation of the existential quantifier $\exists$ is restricted to \emph{linear formulas}, i.e., ($\oc,\wn$)-free formulas. On the one hand, second-order quantifiers are essential to implement iteration of a polynomial time Turing machine transition function (whose advice will be encoded as a stream via the rule $\nuprule$). On the other hand,  the weaker exponential rules of (non-uniform) parsimonious linear logic and the linearity restriction on second-order instantiation together will guarantee a polynomial bound on cut elimination\footnote{{If we dropped the linearity restriction on $\exists$, cut elimination for $\pll$ and $\dpll$  would become superexponential. See~\Cref{rem:exponential-blow-up} for an example formulated within a type-theoretic version of these logics.}}.
}

\subsection{The proof systems $\pll$ and $\dpll$}

The proof systems considered in this paper are formulated in the sequent calculus style presentation, with  \emph{formulas}  from second-order multiplicative-exponential linear logic with units ($\mell$).
These are generated by a countable set of propositional variables $\atoms=\set{X,Y,\ldots}$ using the following grammar:
\[
A\Coloneqq X \mid \cneg X   \mid A\ltens  A\mid A\lpar A \mid \oc A \mid \wn A \mid \lone \mid \lbot \mid \forall X. A \mid \exists X. A 
\]

A \emph{$\oc$-formula}  (resp.~\emph{$\wn$-formula}) is a formula of the form $\oc A$ (resp.~$\wn A$).  An \emph{exponential formula} is either a $\oc$-formula or a $\wn$-formula. We denote by $\mathsf{FV}(A)$ the set of propositional variables occurring free in $A$, and by $A[B/X]$ the standard meta-level capture-avoiding substitution of $B$ for the free occurrences of the propositional variable $X$ in $A$. 
\emph{Linear negation} $\cneg{(\cdot)}$ is defined by De Morgan's~laws
$\cneg{({\cneg A})}		= A$, 
$\cneg{(A \ltens B)} 	= \cneg A \lpar \cneg B$, 
$\cneg {(\oc A)} 		= \wn \cneg A$,  
$\cneg {(\lone)} 	= \lbot $,    and    $\cneg{(\forall X. A)}= \exists X. \cneg{A}$,
while \emph{linear implication} is \mbox{$A \limp B \coloneqq \cneg A \parr B$}.

\begin{figure*}[!t]
	\def\myskip{\hskip3em}
	\centering
	\[
		\vlinf{\axr}{}{A, \cneg A}{}
		\myskip
		\vliinf{\cutr}{}{ \Gamma, \Delta}{\Gamma, A}{\cneg A\!,\Delta}
	\]

	\[
		\vliinf{\ltens}{}{\Gamma,\Delta,A \ltens B}{\Gamma, A}{B, \Delta}
		\myskip
		\vlinf{\lpar}{}{\Gamma, A \lpar B}{\Gamma, A , B}
		\myskip
		\vlinf{\lone}{}{\lone}{}
		\myskip
		\vlinf{\lbot}{}{ {\Gamma}, \lbot}{ {\Gamma}}
	\]
			
\medskip
	\[	
		\vlinf{\forall}{\mbox{\scriptsize \ with $X \not \in \mathsf{FV}(\Gamma)$}}{\Gamma, \forall X. A}{\Gamma, A} 
		\myskip
		\vlinf{\exists}{\mbox{\scriptsize \ where $B$  is ($\oc$,$\wn$)-free }}{\Gamma, \exists X. A}{\Gamma, A[B/X]} 
	\]
	\caption{
	Identity (first line), multiplicative (second line), and second-order (third line) sequent calculus rules. 
	}
	\label{fig:sequent-system-pll}
\end{figure*}

\begin{figure}
	\centering
	\[
	\vlinf{\wnwrule}{}{ {\Gamma}, \wn A}{ {\Gamma}}
	\qquad\qquad
	\vlinf{\wnbrule}{}{ {\Gamma}, {\wn A}}{ {\Gamma}, A, {\wn A}}
	\qquad\qquad
    \vlinf{\fprule}{}{  \wn \Gamma,\oc A}{  \Gamma, A} 
	\qquad\qquad
	\vliinf{\cprule}{}{\wn \Gamma, \oc A}{\Gamma, A}{\wn \Gamma, \oc A}
	\]
	\vspace{0.5cm}
	\[
	\vlderivation{
		\vliiiin{\nuprule}{
			\scriptsize{\text{$\set{\der_i\mid i\in\Nset}$ is finite}}
		}{\wn \Gamma, \oc A}{\vldr{\der_{0}}{\Gamma,A}}{\vlhy{\cdots}}{\vldr{\der_{n}}{\Gamma,A}}{\vlhy{\cdots}}
	}
	\qquad
	\vlinf{\ocwrule}{}{\oc A}{}
	\qquad\qquad
	\vliinf{\ocbrule}{}{\Gamma, \Delta, \oc A}{\Gamma, A}{\Delta, \oc A}
	\]
	\caption{Exponential sequent calculus rules.}
	\label{fig:exponantial-rules-functorial-non-uniform-conditional}
\end{figure}

\begin{definition}[Sequent calculus rules]
		We consider the sequent calculus rules 
		\emph{axiom} ($\axr$), 
		\emph{cut} ($\cutr$), 
		\emph{tensor} ($\ltens$), 
		\emph{par} ($\lpar$),
		\emph{one} ($\lone$),
		\emph{bottom} ($\lbot$),
		\emph{weakening} ($\wnwrule$),  
		\emph{absorption} ($\wnbrule$),
        \emph{functorial promotion} ($\fprule$),  
        \emph{conditional promotion} ($\cprule$), 
        \emph{non-uniform functorial promotion} ($\nuprule$),  
        \emph{co-weakening} ($\ocwrule$),
        \emph{co-absorption} ($\ocbrule$)
		\emph{universal quantifier} ($\forall$), 
		\emph{existential quantifier} ($\exists$) from 
		\Cref{fig:sequent-system-pll,fig:exponantial-rules-functorial-non-uniform-conditional}.
        Rules
        $\ltens$, $\lpar$, $\lone$ and $\lbot$
		are  \emph{multiplicative},  rules $\fprule$, $\cprule$, $\nuprule$, $\wnwrule$ and $\wnbrule$  are  \emph{exponential}, rules $\forall$ and $\exists$ are \emph{second-order}. 
        The formulas $A$ and $\cneg{A}$ in the rule $\cutr$ in \Cref{fig:sequent-system-pll} are the \emph{cut-formulas} of the rule.
        A $\cutr$ rule is:
        \begin{itemize}
            \item of shape \emph{$\cutstep{\rrule_1}{\rrule_2}$} if its premises are conclusions of rules $\rrule_1$ and $\rrule_2$;
            \item \emph{multiplicative} if it is of shape $\cutstep{\rrule_1}{\rrule_2}$ where either at least one among $\rrule_1$ and $\rrule_2$ is $\axr$, or $\rrule_1$ and $\rrule_2$ are multiplicative and introduce the cut-formulas;
            \item \emph{exponential} (resp.~\emph{second-order}) if it is of shape $\cutstep{\rrule_1}{\rrule_2}$ where $\rrule_1$ and $\rrule_2$ are \emph{exponential} (resp.~\emph{second-order}) and introduce the cut-formulas.
        \end{itemize}
\end{definition}

Every proof system we will consider is a subset of the rules in \Cref{fig:sequent-system-pll,fig:exponantial-rules-functorial-non-uniform-conditional}.
        
\begin{definition}[Proof systems $\pll$ and $\dpll$]
	\defin{Second-order parsimonious linear logic}, noted $\pll$, is defined by the set of rules 
	$ \set{\axr,\cutr,\ltens,\lpar,\lone,\lbot,\wnbrule,\wnwrule,\fprule, \forall, \exists}$.
	The set of derivations over the rules in $\pll$ is also denoted by $\pll$.
	The \emph{propositional} fragment of $\pll$ (both the set 	of rules and the set of its derivations) is denoted~by~$\pllprop$. 
	
	\defin{Second-order non-uniform parsimonious linear logic}, noted $\dpll$,  is defined by the set of rules $ \set{\axr,\cutr,\ltens,\lpar,\lone,\lbot,\wnbrule,\wnwrule,\nuprule, \forall, \exists}$, i.e., by replacing $\fprule$ with $\nuprule$ in $\pll$. 
	The set of derivations over the rules in $\dpll$ is also noted~$\dpll$.%
	\footnote{This requires a slight change in \Cref{def:coderivation}: the tree labelled by a derivation in $\dpll$ must be over $\Nset^\omega$ instead of $\set{1,2}^*$, to deal with infinitely branching derivations.}
\end{definition}

Note that $\dpll$ subsumes $\pll$. Indeed, rule $\fprule$ is simulated by rule $\nuprule$ when the derivations $\der_0, \der_1, \ldots \der_n, \ldots$ in its premises are the same.

\begin{remark}
	The proof systems  $\dpll$ and $\pll$ are inspired by Mazza's (type) systems \emph{non-uniform parsimonious logic} $\mathsf{nuPL}_{\forall\ell}$  and \emph{parsimonious logic} $\mathsf{PL}_{\forall\ell}$. 
	The main difference lies in the fact that  our proof systems are free of the \emph{co-absorption rule} ($\ocbrule$) and the \emph{co-weakening rule} ($\ocwrule$) in \Cref{fig:exponantial-rules-functorial-non-uniform-conditional}, which are not admissible in linear logic.
	Because of the absence of these rules,  $\dpll$ and $\pll$  can be properly seen as subsystems of linear logic. 
\end{remark}

\begin{figure*}[!t]
	\adjustbox{max width = \textwidth}{$\begin{array}{c}
		\begin{array}{c|c|c|c}
			\trueder
		&
			\falseder
		&
			\der_\mathsf{abs}
		&
			\der_\mathsf{der}
		\\\cline{1-4}
			\vlderivation{
				\vlin{\forall}{}{
					\vlin{\lpar}{}{( {X^\bot_1}  \parr  {X^\bot_2}) \lpar ({X_3} \otimes {X_4})}
					{\forall X.({X^\bot}  \parr { X^\bot}) \lpar ({X} \otimes X) }
				}{
					\vlin{\parr}{}{ ({ X_1^\bot}  \parr { X^\bot_2}) , ({X_3} \otimes {X_4})}
					{\vliin{\otimes}{}{ { X^\bot_1}  ,  { X^\bot_2} ,  {X_3} \otimes {X_4}}
						{\vlin{\axr}{}{{ X^\bot_1}, {X_3}}{\vlhy{}}}
						{\vlin{\axr}{}{{ X^\bot_2}, {X_4}}{\vlhy{}}}
					}
				}
			} 
		&
			\vlderivation{
				\vlin{\forall}{}{
					\vlin{\lpar}{}{( {X^\bot_1}  \parr  {X^\bot_2}) \lpar ({X_3} \otimes {X_4})}
					{\forall X.({X^\bot}  \parr { X^\bot}) \lpar ({X} \otimes X) }
				}{
					\vlin{\parr}{}{ ({ X_1^\bot}  \parr { X^\bot_2}) , ({X_3} \otimes {X_4})}
					{\vliin{\otimes}{}{ { X^\bot_1}  ,  { X^\bot_2} ,  {X_3} \otimes {X_4}}
						{\vlin{\axr}{}{{ X^\bot_1}, {X_4}}{\vlhy{}}}
						{\vlin{\axr}{}{{ X^\bot_2}, {X_3}}{\vlhy{}}}
					}
				}
			}  
		&
			\vlderivation{
				\vlin{\lpar}{}{\wn \cneg{A} \lpar (A \ltens \oc A)}
				{
					\vlin{\wnbrule}{}{\wn \cneg{A}, A \ltens \oc A}
					{
						\vliin{\ltens}{}{\cneg{A}, \wn \cneg{A}, A \ltens \oc A}
						{
							\vlin{\axr}{}{\cneg{A}, A}{\vlhy{}}
						}
						{
							\vlin{\axr}{}{\wn \cneg{A}, \oc A}{\vlhy{}}
						}
					}
				}
			}
		&
			\vlderivation{
				\vlin{\lpar}{}{\wn \cneg{A} \lpar A}
				{
					\vlin{\wnbrule}{}{\wn \cneg{A}, A}
					{
						\vlin{\wnwrule}{}{\cneg{A}, \wn \cneg{A}, A}
						{
							\vlin{\axr}{}{\cneg{A}, A}{\vlhy{}}
						}
					}
				}
			}
		\end{array}
		\end{array}$}
	\caption{
		Examples of derivations in $\pll$ ($\trueder,
		\falseder,
		\der_\mathsf{abs},
		\der_\mathsf{der}$). 
	}
	\label{fig:der:examples-of-derivations}
\end{figure*}

\begin{example}\label{ex:derivations}
	\Cref{fig:der:examples-of-derivations} gives some examples of derivation in $\pll$.
	The (distinct) derivations $\falseder$ and $\trueder$ prove the same formula $\Bool = \forall X.\allowbreak ({\cneg X}  \parr {\cneg X}) \parr ({X} \otimes X)$, where 
	{we distinguish the occurrences of the variable $X$ by writing them as $X_1$, $X_2$, $X_3$, and $ X_4$ to highlight the differences between the two derivations. 
	}
	Derivations $\dD_{\mathsf{abs}}$ and  $\dD_\mathsf{der}$
	respectively prove  the 
	\emph{absorption law} $\oc A \multimap A \ltens \oc A$
	and
	\mbox{the \emph{dereliction law} $\oc A \multimap A$}.
\end{example}

\begin{definition}[Cut elimination steps in $\pll$ and $\dpll$]
The \defin{cut elimination} relation $\cutelim$ in $\pll$ (resp., $\dpll$)
is the union of
multiplicative, second-order and exponential cut elimination steps 
in  	\Cref{fig:cut-elim-finitary,fig:exponential-functorial-ce} (resp., 	\Cref{fig:cut-elim-finitary,fig:damiano-cut-elimination}). The non-commutative steps are called \emph{principal}.
The reflexive-transitive closure of $\cutelim$ is noted $\cutelim^*$.
\end{definition}

Termination of cut elimination in $\pllprop$ has been proved in~\cite{CSL}, relying on strong normalization of (second-order) linear logic \cite{pag:tor:StrongNorm}. Hence, the same argument extends straightforwardly to $\pll$.
\begin{theorem}
	\label{thm:cut-elimination-pll} 
	
	For every $\der \in \pll$,
	there is a cut-free $\der'\in\pll$ such that $\der\cutelims\der'$.
\end{theorem}


\begin{figure*}[t]
\[
\vlderivation{
	\vliin{\cutr}{}{\Gamma, A}{\vlin{\axr}{}{A, \cneg A}{\vlhy{}}}{\vlhy{\Gamma, A}}
}
\quad 
\cutelim
\quad 
\vlderivation{\vlhy{\Gamma, A}}
\qquad \qquad 
\vlderivation{
	\vliin{\cutr}{}{\Gamma}
	{
		\vlin{\lbot}{
		}{\Gamma, \lbot}{\vlhy{\Gamma}}
	}{
		\vlin{\lone}{}{\lone}{}
	}
}
\quad 
\cutelim
\quad 
\vlderivation{
	\vlhy{\Gamma}
}	
\]
\smallskip
\[
\vlderivation{
	\vliin{\cutr}{}{\Gamma, \Delta,\Sigma}{
		\vlin{\lpar}{}{\Gamma, A\lpar B}{\vlhy{\Gamma, A, B}}
	}{
		\vliin{\ltens}{}{ \Delta,\cneg A \ltens \cneg{B}\!, \Sigma}{\vlhy{ \Delta,\cneg A}}{\vlhy{\cneg{B}\!,\Sigma}}
	}
}\quad 
\cutelim
\quad 
\vlderivation{
	\vliin{\cutr}{}{\Gamma, \Delta, \Sigma}{
		\vliin{\cutr}{}{\Gamma, \Delta, B}{\vlhy{\Gamma, B,A}}{\vlhy{\cneg{A}\!, \Delta}}
	}{
		\vlhy{\cneg{B}\!,\Sigma}
	}
}
\]
	\vspace{0.5cm}
	\[
			\vlderivation{
		\vliin{\cutr}{}{\Gamma, \Delta}
		{
			\vlin{\forall}{
			}{\Gamma, \forall X. A}{\vlhy{\Gamma, A}}
		}{
			\vlin{\exists}{}{\Delta, \exists X. \cneg{A}}{\vlhy{\Delta, \cneg{A}[B/X]}}
		}
	}
\quad 
	\cutelim
	\quad 
	\vlderivation{
		\vliin{\cutr}{}{\Gamma, \Delta}
		{
			\vlhy{\Gamma, A[B/X]}
		}{
			\vlhy{\Delta, \cneg{A}[B/X]}
		}
	}
\]
	\vspace{0.5cm}
\[
	\begin{array}{cl}
		\qquad 
			\vlderivation{
			\vliin{\cutr}{}{\Gamma, \Delta}
			{\vlin{\rrule_1}{}{\Gamma, A}{\vlhy{  \Gamma_1, A}}}
			{\vlhy{\cneg{A}, \Delta}}
		}
		\quad \cutelim \quad 
		\vlderivation{
			\vlin{\rrule_1}{}{\Gamma, \Delta} 
			{\vliin{\cutr}{}{\Gamma_1, \Delta}
				{\vlhy{   \Gamma_1, A }}
				{\vlhy{ \cneg{A}, \Delta}}
			}
		}\\[1cm]
	\qquad 	\vlderivation{
	\vliin{\cutr}{}{\Gamma, \Delta}
	{\vliin{\rrule_2}{}{\Gamma, A}{\vlhy{\Gamma_1,A}}{\vlhy{\Gamma_2}}}
	{\vlhy{\Delta, \cneg A}}
}
\quad \cutelim \quad 
\vlderivation{
	\vliin{\rrule_2}{}{\Gamma,\Delta} 
	{\vliin{\cutr}{}{\Gamma_1,\Delta}
		{\vlhy{\Gamma_1, A}}
		{\vlhy{\cneg{A},\Delta}}
	}
	{\vlhy{\Gamma_2}}
}
&\qquad\mbox{with $\rrule_2 \neq \cutr$.}
	\end{array}
\]

	\caption{Multiplicative {(first two lines)}, second-order {(third line)}, and commutative {(last two lines)} cut elimination steps.}
	\label{fig:cut-elim-finitary}
\end{figure*}

\begin{figure}
	\[
	\vlderivation{
		\vliin{\cutr}{}{\wn \Gamma, \wn \Delta, \oc B}{
			\vlin{\fprule}{}{\wn \Gamma, \oc A}{
				\vlhy{\Gamma, A}
			}
		}{
			\vlin{\fprule}{}{\wn \cneg {A}\!, \wn \Delta, \oc B}{
				\vlhy{\cneg {A}\!, \Delta, B}
			}
		}
	}
	\quad 
	\cutelim
	\quad 
	\vlderivation{
		\vlin{\fprule}{}{\wn \Gamma, \wn \Delta, \oc B}{
			\vliin{\cutr}{}{\Gamma, \Delta, B}{
				\vlhy{\Gamma, A}
			}{
				\vlhy{\cneg {A}\!, \Delta, B}
			}
		}
	}
	\]
	\smallskip
	\[
	\vlderivation{
		\vliin{\cutr}{}{\wn \Gamma, \Delta}{
			\vlin{\fprule}{}{\wn \Gamma, \oc A}{
				\vlhy{\Gamma, A}
			}
		}{
			\vlin{\wnwrule}{}{\Delta,\wn \cneg A}{\vlhy{\Delta}}
		}
	}
	\quad 
	\cutelim
	\quad
	\vlderivation{
		\vliq{\wnwrule}{}{\wn\Gamma,\Delta}{\vlhy{\Delta}}
	}		
	\]
	\smallskip
	\[
	\vlderivation{
		\vliin{\cutr}{}{\wn \Gamma, \Delta}{
			\vlin{\fprule}{}{\wn \Gamma, \oc A}{
				\vlhy{\Gamma, A}
			}
		}{
			\vlin{\wnbrule}{}{\Delta, \wn \cneg A}{
				\vlhy{\Delta,\cneg {A}\!, \wn\cneg A}
			}
		}
	}
	\quad 
	\cutelim
	\quad 
	\vlderivation{
		\vliq{\wnbrule}{}{\wn\Gamma, \Delta}{
			\vliin{\cutr}{}{\Gamma, \wn \Gamma,\Delta }{
				\vlhy{\Gamma, A}
			}{
				\vliin{\cutr}{}{\wn\Gamma,  \Delta, \cneg A}{
					\vlin{\fprule}{}{\wn \Gamma,\oc A}{\vlhy{\Gamma, A}}
				}{
					\vlhy{\Delta, \cneg {A}\!, \wn \cneg A}
				}
			}
		}
	}
	\]
	\caption{Exponential cut elimination steps {in $\pll$ (with $\fprule$)}.}
	\label{fig:exponential-functorial-ce}
\end{figure}

A byproduct of our grand tour diagram in~\Cref{fig:grand-tour-diagram} is that $\pll$ {represents} exactly the class of functions in $\fptime$. 
To see this,  we  introduce a rather permissive notion of representability for $\pll$, along the lines of~\cite{Ronchi-Gaboardi,Uniform-enc}.
This notion smoothly adapts to other proof systems we shall study in this paper.

\begin{definition}[Representability]\label{defn:representability-bis} 
	A set $T$ is \defin{represented in $\pll$} by a formula $\mathbf{T}$  if 
	there is an injection $(\cod{\,\cdot\,})$ from $T$ to the set of $\cutr$-free derivations in $\pll$ with conclusion $\mathbf{T}$.
	
	A derivation $\der$ in $\pll$ \defin{represents} a (total) function  $f \colon T_1 \times \ldots \times T_n \to T$  if it proves {$ \mathbf{T_1}\limp \ldots\limp \mathbf{T_n}\limp \mathbf{T}$} where $\mathbf{T_1}, \ldots, \mathbf{T_n}, \mathbf{T}$ represent $T_1, \ldots, T_n, T$ respectively, and for all $x_1 \in T_1, \ldots, x_n \in T_n$, the reduction  in~\Cref{fig:representability} holds.
	A (total) function  $f \colon T_1 \times \ldots \times T_n \to T$ is \defin{representable in $\pll$}   if there is a derivation in $\pll$ representing $f$.
	We denote by $\cod{f}$ a derivation representing $f$. 
\end{definition}

\begin{figure*}[t]
	\centering
\[
		\vlderivation{
			\vliin{\er \limp}{}{\mathbf{T}}
			{
				\vliin{\er \limp}{}{\mathbf{T_n}\limp \mathbf{T}}
				{
					\vliin{\er \limp}{}{\vdots}
					{	
						\vldr{\der}{ \mathbf{T_1}\limp \ldots\limp \mathbf{T_n}\limp \mathbf{T}}
					}
					{
						\vltr{\cod{x_1}}{\mathbf{T_1}}{\vlhy{\ }}{\vlhy{ \ }}{\vlhy{\ }}
					}
				}
				{
					\vltr{\cod{x_{n-1}}}{\mathbf{T_{n-1}}}{\vlhy{\ }}{\vlhy{ \ \ \ \ }}{\vlhy{\ }}
				}
			}
			{
				\vltr{\cod{x_{n}}}{\mathbf{T_{n}}}{\vlhy{\ }}{\vlhy{ \ }}{\vlhy{\ }}
			}
		}
		\quad 
		\cutelim^*  
		\vlderivation{
			\vltrf{\cod{f(x_1, \ldots, x_n)}}{\mathbf{T}}{\vlhy{}}{\vlhy{\qquad\qquad\qquad}}{\vlhy{}}{0.5}
		}
		\]
		\[		
		\mbox{where}\quad
		\vliinf{\er \limp}{}{\Gamma, \Delta, B}{
			\Gamma, A \limp B
		}{
			\Delta, A 
		}
		\dfn
		\vlderivation{
			\vliin{\cutr}{}{\Gamma, \Delta, B}{
				\vlhy{\Gamma, A \limp B}
			}{
				\vliin{\otimes}{}{\Delta, \cneg{(A \limp B)}, B}{
					\vlhy{\Delta, A}
				}{
					\vlin{\axr}{}{\cneg{B}, B}{\vlhy{}}
				}
			}
		}
		\]
	\caption{Representability of a function $f: T_1 \times \ldots \times T_n \to T$.}
	\label{fig:representability}
\end{figure*}

\begin{example}\label{ex:representation-bool}
	The set of Booleans $\Boolset = \{\false, \true\}$ is represented in $\pll$ by the formula $\Bool$ in~\Cref{ex:derivations} thanks to the derivations $\falseder$ and $\trueder$  in~\Cref{fig:der:examples-of-derivations}.
	The set $\{\false,\true\}^*$ of Boolean strings is represented by the formula $	\String \dfn\forall X. \oc (\Bool \limp X \limp X)\limp X\limp X$.
	We will actually mainly work with a parametric version of $\String$, i.e., the instantiation	$\String[A]\dfn\oc (\Bool \limp A \limp A)\limp A\limp A$  for any formula $A$. 
	We  write $\String[]$ to denote $\String[A]$ for some $A$.	
	Each string $b_1 \cdots b_n \in \{\false,\true\}^*$ with $n \geq 0$ is then encoded in $\pll$ by the derivation $\cod{b_1 \cdots b_n}$ of $\String[A]$ shown in \Cref{eq:string-encoding} below:
\begin{equation}\label{eq:string-encoding}
\vlderivation{
	\vliq{2\times\parr}{}{S[A]}{
		\vliq{n\times \wnbrule}{}{
			\wn ({\Bool} \otimes {A}\otimes  \cneg{A}),  \cneg{A} , A
		}{
			\vlin{\wnwrule}{}{
				({\Bool} \otimes {A}\otimes  \cneg{A}), 
				\overset{n }{\ldots }, 
				({\Bool} \otimes {A}\otimes  \cneg{A}), 
				\wn ({\Bool} \otimes {A}\otimes  \cneg{A}),
				\cneg{A} , A
			}{
				\vliin{\ltens}{}{ 
					({\Bool} \otimes {A}\otimes  \cneg{A}), \overset{n }{\ldots }, ({\Bool} \otimes {A}\otimes  \cneg{A}), \cneg{A} , A
				}{
						\vldr{\cod{b_1}}{\Bool}
				}{
					\vliin{\otimes}{}{
						({\Bool} \otimes {A}\otimes  \cneg{A}), \overset{n -1}{\ldots }, ({\Bool} \otimes {A}\otimes  \cneg{A}), {A}\otimes  \cneg{A},  \cneg{A} , A
					}{
						\vlin{}{}{
							({\Bool} \otimes {A}\otimes  \cneg{A}), \overset{n-1 }{\ldots }, ({\Bool} \otimes {A}\otimes  \cneg{A}), \cneg{A} , A
						}{
							\vlin{}{}{\vdots}{
								\vliin{\otimes }{}{
									({\Bool} \otimes {A}\otimes  \cneg{A}),  \cneg{A} , A
								}{
									\vldr{\cod{b_n}}{\Bool}
								}{
									\vliin{\otimes}{}{ 
										{A}\otimes  \cneg{A},  \cneg{A} , A
									}{ 
										\vlin{\axr}{}{\cneg{A}, A}{\vlhy{}}
									}{
										\vlin{\axr}{}{\cneg{A}, A}{\vlhy{}}
									}
								}
							}
						}
					}{
						\vlin{\axr}{}{\cneg{A}, A}{\vlhy{}}
					}
				}
			}
		}
	}
}
\end{equation}

	%
\end{example}

\begin{figure*}[t]
\[
			\vlderivation{
				\vliin{\cutr}{}{\wn \Gamma, \wn \Delta, \oc B}{
					\vlin{\nuprule}{}{\wn \Gamma, \oc A}{
						\vlhy{\left\{ \vlderivation{ \vldr{\der_i}{\Gamma, A}}\right\}_{i \in \Nset} }
					}
				}{
					\vlin{\nuprule}{}{\wn \cneg{A}\!, \wn \Delta, \oc B}{
						\vlhy{\left\{ \vlderivation{ \vldr{\der'_i}{\cneg{A}\!, \Delta, B}}\right\}_{i \in \Nset} }
					}
				}
			}
		\quad 
			\cutelim\quad 
			\vlderivation{
				\vlin{\nuprule}{}{\wn \Gamma, \wn \Delta, \oc B}{
					\vlhy{
						\left\{  \vlderivation{ 
							\vliin{\cutr}{}{\Gamma, \Delta, B}
							{
								\vldr{\der_i}{\Gamma, A}
							}{
								\vldr{\der'_i}{\cneg{A}\!, \Delta, B}
							}
						} 	\right\}_{i \in \Nset}
					}
				}
			}
			\]
			\medskip
			\[
			\vlderivation{
				\vliin{\cutr}{}{\wn \Gamma, \Delta}{
					\vlin{\nuprule}{}{\wn \Gamma, \oc A}{
						\vlhy{\left\{ \vlderivation{ \vldr{\der_i}{\Gamma, A}}\right\}_{i \in \Nset} }
					}
				}{
					\vlin{\wnwrule}{}{\Delta,\wn \cneg A}{\vlhy{\Delta}}
				}
			}
			\quad \cutelim \quad 
			\vlderivation{
				\vliq{\size{\Gamma}\times\wnwrule}{}{\wn\Gamma,\Delta}{\vlhy{\Delta}}
			}
			\]
			\[
			\vlderivation{
				\vliin{\cutr}{}{\wn \Gamma, \Delta}{
					\vlin{\nuprule}{}{\wn \Gamma, \oc A}{
						\vlhy{\left\{ \vlderivation{ \vldr{\der_i}{\Gamma, A}}\right\}_{i \in \Nset} }
					}
				}{
					\vlin{\wnbrule}{}{\Delta, \wn \cneg A}{
						\vlhy{\Delta,\cneg{A}\!, \wn\cneg A}
					}
				}
			}
			\quad \cutelim\quad 
			\vlderivation{
				\vliq{\size{\Gamma}\times\wnbrule}{}{\wn\Gamma, \Delta}{
					\vliin{\cutr}{}{\Gamma, \wn \Gamma,\Delta }{
						\vldr{\der_0}{\Gamma, A}
					}{
						\vliin{\cutr}{}{\wn\Gamma,  \Delta, \cneg A}{
							\vlin{\nuprule}{}{\wn \Gamma,\oc A}
							{
								\vlhy{\left\{ \vlderivation{ \vldr{\der_{i+1}}{\Gamma, A}}\right\}_{i \in \Nset} }
							}
						}{
							\vlhy{\Delta, \cneg{A}\!, \wn \cneg A}
						}
					}
				}
			}
\]
	\caption{Exponential cut elimination steps {in $\dpll$ (with $\nuprule$)}.}
	\label{fig:damiano-cut-elimination}
\end{figure*}

\section{Non-wellfounded Second-Order Parsimonious Linear Logic }\label{sec:coder}

\subsection{The non-wellfounded proof system $\nwpll$}
\label{subsect:from-infitely-branching}

{In~{\cite{CSL}} we introduced $\nwpllprop$, a non-wellfounded (finitely branching) version of propositional parsimonious linear logic $\pllprop$, by exploiting the notion of coderivation, as opposed to derivation (see~\Cref{def:coderivation}). 
	We now introduce $\nwpll$, a second-order extension of $\nwpllprop$.}

\begin{definition}[The proof system $\nwpll$] 
\defin{Non-wellfounded second-order parsimonious linear logic}, noted $\nwpll$, is defined as the set of \emph{coderivations} over the set of rules $ \set{\axr,\ltens,\lpar, \allowbreak\lone, \lbot, \allowbreak\cutr,\wnbrule,\wnwrule,\cprule, \forall, \exists}$, i.e., obtained from $\pll$ (resp., $\dpll$) by replacing  $\fprule$ (resp., $\nuprule$) with {the \defin{conditional promotion} rule $\cprule$ (see \Cref{fig:exponantial-rules-functorial-non-uniform-conditional}).} 
\end{definition}

{Non-wellfounded second-order parsimonious linear logic $\nwpll$ subsumes both $\pll$ and $\dpll$. 
	Indeed, both $\fprule$ and $\nuprule$ can be simulated in $\nwpll$ by an infinite coderivation called \emph{non-wellfounded box} (see~\Cref{eq:box} and the definition below) obtained by iterating  $\cprule$ (to simulate $\fprule$, the coderivations $\der_0, \der_1, \ldots \der_n, \ldots$ in \Cref{eq:box} have to be the same).}

\begin{figure}[t]
	\centering
{
$\begin{array}{c|c|c}
	\zeroder
	&
	\wnder 
	&
	\derstream{\der_0,\ldots, \der_n} 
	\\\cline{1-3}
	\vlderivation{
		\vliin{\cutr}{}
		{\Gamma, A}
		{
			\vlin{\axr}{}{\cneg A, A}{\vlhy{}}
		}
		{
			\vliin{\cutr}{}
			{\Gamma,  A}
			{
				\vlin{\axr}{}{\cneg A, A}{\vlhy{}}
			}
			{\vlin{\cutr}{}{\Gamma, A}{\vlhy{\vdots}}}
		}
	}
	\quad 
	&
	\quad 
	\vlderivation{
		\vlin{\wnbrule}{}{\wn A}{
			\vlin{\wnbrule}{}{A, \wn A}{
				\vlin{\wnbrule}{}{A, A, \wn A}{\vlhy{\vdots}}
			}
		}
	}
	\quad
	&
	\quad
	\vlderivation{
		\vliin{\cprule}{}{\wn \Gamma, \oc A}{\vldr{\der_0}{\Gamma, A}}{
			\vliin{\cprule}{}{{\wn \Gamma, \oc A}}{\vldr{\der_1}{\Gamma, A}}{				 			
				\vliin{\cprule}{}{\reflectbox{$\ddots$}}{\vldr{\der_n}{\Gamma, A}}{
					\vlin{\cprule}{}{{\wn \Gamma, \oc A}}{\vlhy{\vdots}}
				}
			}
		}
	}
\end{array}$
}	
	\caption{{Examples of coderivations in $\nwpll$ ($\zeroder,\wnder$) and the non-wellfounded box $\derstream{\der_0,\ldots, \der_n}$ in $\nwpll$.}}
	\label{eq:box}
\end{figure}

\begin{definition}\label{defn:boxes} 
	A   \defin{non-wellfounded box} 
	(\nwbox for short)
	is a coderivation of  $  \nwpll$ of the form $\derstream{\der_0,\der_1,\ldots, \der_n}$  as in~\Cref{eq:box}, for {any} formula $A${, sequent $\Gamma$} and coderivations $\der_0, \der_1, \ldots \der_n, \ldots \in \nwpll$.  Its \defin{principal formula} is the formula $\oc A$; its \defin{main branch} is the  infinite branch $\set{\epsilon, 2,22,\dots}$, and its $i$-th  \defin{call} is the coderivation $\der_i$.  
	
	Non-wellfounded boxes will range over $\nwpromotion$, where $\nwpromotion(i)$ denotes its $i$-th call, and  $\supportof\nwpromotion=\set{\nwpromotion(i)\mid i\in\Nset}$ denotes the set of its calls.

	We say that $\nwpromotion$ has \defin{finite support} ({resp.~}is \defin{periodic} with \defin{period} $k$)
	if $\supportof\nwpromotion$ is finite ({resp.~}if there is a minimal $k\in\Nset$ such that $\nwpromotion(i)= \nwpromotion(k+i)$ for any $i\in\Nset$).
	A coderivation $\der$ 
	has  \defin{finite support} ({resp.~}is \defin{periodic}) 
	if any $\nwbox$ in $\der$ has finite support ({resp.~}is periodic).

\end{definition}

\begin{example}\label{ex:nonReg}
Streams of booleans  can be encoded in $\nwpll$ by $\nwbox$s 
$\nwpromotion=\derstream{\der_0,\ldots, \der_n}$ as in \Cref{eq:box} with  $A \dfn \Bool$ and  $\der_i\in  \set{\falseder,\trueder}$ for each $i \in \Nset$. Then, $\nwpromotion$ has finite support, as its only calls can be $\falseder$ or $\trueder$, and it is periodic if and only if so is the stream  $\seq{\der_0,\ldots, \der_n,\ldots} \in  \set{\falseder,\trueder}^\omega$.
\end{example}

\begin{definition}[Cut elimination and representability in $\nwpll$]
	The \defin{cut elimination} relation $\cutelim$ in $\nwpll$ is the union of multiplicative, second-order, commutative and exponential cut elimination steps in~\Cref{fig:cut-elim-finitary,fig:cut-elim-pll}.
	The reflexive-transitive closure of $\cutelim$ is noted $\cutelim^*$.
	The notion of \defin{representability} for $\nwpll$  can be obtained by  adapting  \Cref{defn:representability-bis} to coderivations in $\nwpll$ {in the obvious way}. \Cref{ex:nonReg} shows that streams of booleans are representable in $\nwpll$.
\end{definition}

Similarly to $\pll$ and $\dpll$, the non-wellfounded proof-system $\nwpll$ admits a computational interpretation based on streams.   A stream of data $\der_0, \der_1, \ldots, \der_n, \ldots$ of type $A$ can be encoded by a \nwbox of the form $\derstream{\der_0,\ldots, \der_n}$ as in~\Cref{eq:box} (see, e.g., ~\Cref{ex:nonReg}).
The cut elimination step $\cutstep\cprule\wnbrule$ {then pops}  the head of the stream, the step $\cutstep\cprule\wnwrule$ erases a stream and,  finally,   $\cutstep\cprule\cprule$ allows us to stepwise ``zip'' two streams, that is, to create a new \nwbox whose $i$-th call is obtained by cutting the $i$-th calls of the two input \nwboxes.
 
\begin{remark}\label{eqn:non-finite-support}
	Unlike  the streams encoded by the rule $\nuprule$ of $\dpll$,  \nwboxes can encode streams with infinitely many distinct entries. As an example, the stream whose $i^\textup{th}$ entry is the  boolean stream $\seq{\trueder, \overset{i}{\ldots}, \trueder, \falseder, \dots}$ is represented by the  $\nwbox$	 $\nwpromotion=\derstream{\der_0,\ldots, \der_n}$ as in \Cref{eq:box} with $A \dfn \oc  \Bool$ and 
	  $	\der_i\dfn \derstream{\trueder, \overset{i}{\ldots}, \trueder, \falseder}$, whose calls are (cut-free and) pairwise~distinct.
\end{remark}

In~\Cref{subsec:reg-weak-reg} we will introduce conditions on coderivations of $\nwpll$ enforcing finite support of the stream encoded by a \nwbox by requiring that it has \emph{finitely} many distinct calls.

\begin{figure*}[t]
	\[
			\vlderivation{
				\vliin{\cutr}{}{\wn \Gamma, \wn \Delta, \oc B}{
					\vliin{\cprule}{}{\wn \Gamma, \oc A}{
						\vlhy{\Gamma, A}
					}{
						\vlhy{\wn \Gamma, \oc A}
					}
				}{
					\vliin{\cprule}{}{\wn \cneg A\!, \wn \Delta, \oc B}{
						\vlhy{\cneg{A}\!, \Delta, B}
					}{
						\vlhy{\wn \cneg{A}\!, \wn \Delta, \oc B}
					}
				}
			}
		\quad 
			\cutelim
			\quad 
			\vlderivation{
				\vliin{\cprule}{}{\wn \Gamma, \wn \Delta, \oc B}{
					\vliin{\cutr}{}{\Gamma, \Delta, B}{
						\vlhy{\Gamma, A}
					}{
						\vlhy{\cneg A, \Delta, B}
					}
				}{
					\vliin{\cutr}{}{\wn\Gamma, \wn\Delta,\oc B}{
						\vlhy{\wn\Gamma, \oc A}
					}{
						\vlhy{\wn \cneg{A}\!, \wn \Delta, \oc B}
					}
				}
			}
			\]
			\medskip
			\[
			\vlderivation{
				\vliin{\cutr}{}{\wn \Gamma, \Delta}{
					\vliin{\cprule}{}{\wn \Gamma, \oc A}{
						\vlhy{\Gamma, A}
					}{
						\vlhy{\wn \Gamma, \oc A}
					}
				}{
					\vlin{\wnwrule}{}{\Delta,\wn \cneg A}{\vlhy{\Delta}}
				}
			}
		\quad 
			\cutelim
			\quad 
			\vlderivation{
				\vliq{\size{\Gamma}\times\wnwrule}{}{\wn\Gamma,\Delta}{\vlhy{\Delta}}
			}
			%
			\]
			\medskip
			\[
			\vlderivation{
				\vliin{\cutr}{}{\wn \Gamma, \Delta}{
					\vliin{\cprule}{}{\wn \Gamma, \oc A}{
						\vlhy{\Gamma, A}
					}{
						\vlhy{\wn \Gamma, \oc A}
					}
				}{
					\vlin{\wnbrule}{}{\Delta, \wn \cneg A}{
						\vlhy{\Delta,\cneg{A}\!, \wn\cneg A}
					}
				}
			}
		\quad 
			\cutelim
			\quad 
			\vldownsmash{\vlderivation{
					\vliq{\size{\Gamma}\times\wnbrule}{}{\wn\Gamma, \Delta}{
						\vliin{\cutr}{}{\Gamma, \wn \Gamma,\Delta }{
							\vlhy{\wn\Gamma, \oc A}
						}{
							\vliin{\cutr}{}{\Gamma,  \Delta, \wn \cneg A}{
								\vlhy{ \Gamma, A}
							}{
								\vlhy{\Delta, \cneg{A}\!, \wn \cneg A}
							}
						}
					}
			}}
\]
	\caption{Exponential cut elimination steps {in $\nwpll$ (with $\cprule$)}.}
	\label{fig:cut-elim-pll}
\end{figure*}

\subsection{Totality via  the progressivity criterion}

The non-wellfounded proof system $\nwpll$ is logically inconsistent, as the  coderivation $\zeroder$ in {\Cref{eq:box}} shows that any non-empty sequent is provable in $\nwpll$.
	In particular, this coderivation is not $\cutr$-free and can  only reduce to itself by a cut elimination step, so that  cut elimination  fails for $\nwpll$. 	
	From a computational viewpoint, this means that  $\nwpll$ can represent non-total functions.

In non-wellfounded proof theory, the typical way to recover  logical consistency and (computationally) \emph{totality} of representable functions, is to introduce a global soundness condition on coderivations called \emph{\prog criterion} \cite{brotherston2011sequent,Kuperberg-Pous21,Das2021,Das2021-preprint}.
In $\nwpll$, this criterion relies on tracking occurrences of \mbox{$\oc$-formulas in coderivations} \cite{CSL}.

\begin{definition}\label{defn:octhread}
	Let $\dD$ be a coderivation in $\nwpll$.
	An occurrence of a formula in a premise of a rule $\rrule$ is the \defin{parent} of an occurrence of a formula in the conclusion if they are connected according to the edges depicted in \Cref{fig:threads}.
	A \defin{\octhread} in $\der$ is a maximal sequence $(A_i)_{i \in I}$ of $\oc$-formulas for some downward-closed $I \subseteq \Nset$ such that $A_{i+1}$ is the parent of $A_i$ for all $i \in I$.
	A \octhread $(A_i)_{i \in I}$ is \defin{\prog} if $A_j$ is in the conclusion of a $\cprule$ for infinitely many $j\in I$.
	$\der$ is \defin{\prog} if every infinite branch contains a \prog \octhread. 
\end{definition}

Note that every derivation in $\nwpll$ is (vacuously) progressing.

\begin{figure*}[!t]
\[
			\vliinf{\cutr}{}{
				\vF1_1, \ldots,\vF2_n,\vG1_1,\ldots,\vG2_m
			}{
				\vF3_1,\dots\vF4_n, A
			}{\cneg{A},\vG3_1,\ldots,\vG4_m}
			\Tedges{F1.center/F3.center,F2.center/F4.center,G1.center/G3.center,G2.center/G4.center}
			\quad 
			\vlinf{\lpar}{}{
				\vF1_1, \ldots,\vF2_n,\vA2\lpar \vB2
			}{
				\vF3_1,\dots\vF4_n, \vA1\;,\;\vB1
			}
			\Tedges{F1.center/F3.center,F2.center/F4.center,A1.center/A2.center,B1.center/B2.center}
			\quad 
			\vliinf{\ltens}{}{
				\vF1_1, \ldots,\vF2_n,\vA2\ltens \vB2,\vG1_1,\ldots,\vG2_m
			}{
				\vF3_1,\dots\vF4_n, \vA1
			}{\vB1,\vG3_1,\ldots,\vG4_m}
			\Tedges{F1.center/F3.center,F2.center/F4.center,G1.center/G3.center,G2.center/G4.center,A1.center/A2.center,B1.center/B2.center}
			\]
			\medskip
\[
			\vlinf{\lbot}{}{\vF1_1, \ldots,\vF2_n, \lbot}{ {\vF3_1,\dots,\vF4_n \quad}}
			\Tedges{F1.center/F3.center,F2.center/F4.center}
\qquad 
			\vliinf{\cprule}{}{
				\wn \vF1_1, \ldots,\wn \vF2_n,\oc \vA1
			}{F_1,\ldots, F_n, A\quad}{
				\wn \vF3_1, \ldots,\wn\vF4_n,\oc \vA2
			}
			\Tedges{F1.center/F3.center,F2.center/F4.center,A1.center/A2.center}
			\quad 
			\vlinf{\wnwrule}{}{\vF1_1, \ldots,\vF2_n, \wn A}{ {\vF3_1,\dots,\vF4_n \quad}}
			\Tedges{F1.center/F3.center,F2.center/F4.center}
		\]
		\medskip
		\[
			\vlinf{\wnbrule}{}{\vF1_1, \ldots,\vF2_n, \wn \vA2}{ {\vF3_1,\dots,\vF4_n,A, \wn \vA1 }}
			\Tedges{F1.center/F3.center,F2.center/F4.center,A1.center/A2.center}
			\qquad
				\vlinf{\forall}{X \not \in FV(\Gamma)}{
				\vF1_1, \ldots,\vF2_n,\forall X. \vA2
			}{
				\vF3_1,\dots\vF4_n, \vA1
			}
			\Tedges{F1.center/F3.center,F2.center/F4.center,A1.center/A2.center}
			\qquad 
				\vlinf{\exists}{}{
				\vF1_1, \ldots,\vF2_n,\exists X.\vA2
			}{
				\vF3_1,\dots\vF4_n, \vA1[B/X]
			}
			\Tedges{F1.center/F3.center,F2.center/F4.center,A1.center/A2.center,B1.center/B2.center}			
\]
	\caption{$\nwpll$ rules: edges connect a formula in the conclusion with its parent(s) in a premise.}
	\label{fig:threads}
\end{figure*}


\newemptyvertex{ghost}{}
\begin{example}\label{ex:prog}
	The coderivations  $\zeroder$ and  $\wnder$ in {\Cref{eq:box}} are not 
	\prog:
	the rightmost branch of $\zeroder$, i.e., the branch  $\set{\epsilon,2,22,\ldots}$,
	and the unique branch of $\wnder$
	are infinite and contain no $\cprule$-rules.
	By contrast, the \nwbox $\derstream{\cod{i_0},\ldots,\cod{i_n}}$ discussed in  \Cref{ex:nonReg} is \prog  since the only infinite branch is its main branch,  which contains a $\oc$-thread of formulas $\oc A$, each one principal for a $\cprule$ rule. Finally, the regular coderivation below is not \prog:
	the branch $\set{\epsilon, 2, 21, 212, 2121, \dots}$ is infinite but has no progressing \octhread (where $X_1, X_2, X_3$ are distinct occurrences of the propositional variable $X$).
	
	\[
			\vlderivation{
				\vliin{\cprule}{}{\wn\vnX1 \;, \;\;\;\oc{\vX1 _1}}{
					\vlin{\axr}{}{X, \ \cneg{X}}{\vlhy{}}}{
					\vliin{\cutr}{}{\wn \vnX2 \;\;,  \;\;\;\; \oc\vX2 _1}{
						\vliin{\cprule}{}{\wn\vnX3, \ {\oc\vX3 _2}}{
							\vlin{\axr}{}{X, \ \cneg{X}}{\vlhy{}}}{
							\vliin{\cutr}{}{\wn \vnX4, \ \oc\vX4_2}{
								\vlin{\cprule}{}{\wn\vnX5, {\oc \vX5 _3}}{
									\vlhy{  
										\vghost2{~}\;\; \vdots \quad\vghost1{~}
									}
								}
							}{
								\vlin{\axr}{}{{\wn \cneg X},\oc \vX7 _2}{\vlhy{}}
							}
						}
					}{
						\vlin{\axr}{}{{\wn \cneg X},\oc \vX6 _1}{\vlhy{}}
					}
				}
			}
			\Tedges{
				X1.225/X1.135,
				X1.135/X2.225,
				X2.225/X2.135,
				X2.135/X6.225,
				X6.225/X6.135,
				X3.225/X3.135,
				X3.135/X4.225,
				X4.225/X4.135,
				X4.135/X7.225,
				X7.225/X7.135,
				X5.225/ghost1.south%
			}
			\dTedges{ghost1.south/ghost1.north}
		\]
\end{example}

\begin{remark}\label{rem:uniqueness}
	Any infinite branch in a \prog coderivation $\der \in \nwpll$ contains exactly one \prog $\oc$-thread.
	This follows from maximality of $\oc$-threads and the fact that conclusions of $\cprule$-rules contain at most one $\oc$-formula. 
	As a consequence, 	any infinite $\oc$-thread  in a branch of  $\der$ must be \prog. 
\end{remark}

In~\cite{CSL}, we proved a cut elimination result for the \emph{propositional} progressing coderivations of {$\nwpllprop$} (i.e., without second-order),   called \emph{continuous cut elimination theorem}. Its proof relies on defining particular infinitary rewriting strategies, showing that the infinite branches of the limit cut-free coderivation constructed are well-defined and contain progressing $\oc$-threads. 
The proof smoothly extends to the whole $\nwpll$ thanks to ($\oc,\wn$)-freeness of  the formulas instantiated by the rule $\exists$. Indeed, by virtue of that condition, a $\oc$-thread of $\nwpll$  never starts at the active formula of $\exists$: as in the propositional case, it can only start at a formula in the conclusion of the coderivation or at a cut-formula. Therefore,  our restricted second-order quantifiers, and the corresponding cut elimination step $\cutstep{\exists}{\forall}$, do not change the \emph{geometry} of  coderivations. As a consequence, the $\cutr$ admissibility result below holds.

\begin{theorem}
	[Cut elimination for progressing $\nwpll$]
For every progressing coderivation of $\nwpll$ there is a  cut-free progressing coderivation with the same conclusion.
\end{theorem}

\subsection{Approximating coderivations}\label{subsec:cut-elim-approx}

Rewriting a coderivation to a cut-free one may require infinitely many steps of  cut elimination. In this subsection we  introduce a notion of approximation for coderivations and show that for \emph{finite} approximations  there is a bound on the number of cut elimination steps.

\begin{definition}
	We define the set of rules $\opll \dfn \nwpll \cup \set{\zero}$, where $\zero\dfn \vlupsmash{\vlinf{\zero}{}{\Gamma}{}}$  for any sequent $\Gamma$.
	We will also refer to $\opll$  as the set of coderivations over $\opll$, which we call \defin{open coderivations}. 
	An  \defin{open derivation} is a derivation in $\opll$.  	Previously introduced notions and definitions on coderivations extend to open coderivations in the obvious way,  e.g., the global condition in~\Cref{defn:octhread}  as well as cut elimination $\cutelim$.
\end{definition}
		
Note that there are open coderivations containing $\cutr$ rules that cannot be further reduced by the cut elimination steps in~\Cref{fig:cut-elim-finitary} and~\Cref{fig:cut-elim-pll}, since no  cut elimination step is defined for $\zero$.  Henceforth, we will call such open coderivations \defin{normal}.

\newcommand{\weightcp}[1]{\mathsf{C}(#1)}
		
\begin{definition}
	Let  
	$\der$ be an open coderivation and 
	$\Nodes=\{v_1, \ldots, v_n\} \allowbreak\subseteq \set{1,2}^*$ be a finite set of mutually incomparable  nodes of $\der$ (w.r.t. the prefix  order). If $\set{\der'_{i}}_{ 1 \leq i \leq n}$ is a set of open coderivations such that $\der'_{i}$ has the same conclusion  as the sub-coderivation $\der_{v_i}$ of $\der$,  denote by 
	$
	\der(\der'_{1}/v_1, \ldots, \der'_{n}/v_n)$, 
	the open coderivation obtained by replacing each $\der_{v_i}$ with 
	$\der'_i$ in $\der$.  The \defin{pruning} of $\der$ over $\Nodes$
	is the open coderivation {$\prun{\der}{\Nodes}=\der(\zero/v_1, \ldots, \zero/v_n)$}.

	If $\der$ and $\der'$ are open coderivations, 
	$\der$ is an \defin{approximation} of $\der'$
	(noted $\der \preceq \der'$) 
	iff $\der= \prun{\der'}{\Nodes}$ for some $\Nodes\subseteq\set{1,2}^*$.
	An approximation is \defin{finite} if it is an open derivation.
\end{definition}
		
		%
		%

Cut elimination steps do  not increase the size of open derivations:
		
\begin{proposition}[Cubic bound]
	\label{thm:cut-elimApp} 
	Let $\der$ be an open derivation and let $\mathsf{S}(\der)$ be the maximum number of $\wn$-formulas in the conclusion of a $\cprule$ rule of $\der$. 
	If 
	$
	\der= \der_0 \cutelim \cdots \cutelim\der_n
	$
	then:
	\begin{enumerate}
		\item \label{enum:cubic1} $n\in\mathcal{O}(\mathsf{S}(\der)^3\cdot\vert \der \vert^3)$
		\item  \label{enum:cubic2} $\size{\der_i}\in \mathcal{O}(\mathsf{S}(\der)\cdot\vert \der \vert)$ 	for any $i\in\intset0n$.
	\end{enumerate}
\end{proposition}
\begin{proof}
	For $\der$ an open derivation, let $\weightcp{\der}$ be the number of $\cprule$ in $\der$ and $\heightcut{\der}$ be the sum of the sizes of all subderivations
	of $\der$ whose root is the conclusion of a $\cutr$ rule.
	If~$\der \cutelim \der'$~via:
	\begin{itemize}
		\item a commutative cut elimination step (\Cref{fig:cut-elim-finitary}), then $\weightcp{\der} =  \weightcp{\der'}$, $\size{\der} = \size{\der'}$ and $\heightcut{\der} > \heightcut{\der'}$;
		\item a multiplicative or second-order cut elimination step (\Cref{fig:cut-elim-finitary}),  then $\weightcp{\der} = \weightcp{\der'}$ and $\size{\der} > \size{\der'}$;
		\item an exponential cut elimination step (\Cref{fig:cut-elim-pll}),  then 
		and  $\weightcp{\der} > \weightcp{\der'}$.
	\end{itemize}
	Since the lexicographic order over  the tuple
	$\seq{\weightcp{\der},\size{\der},\heightcut{\der}}$
	is wellfounded,  we conclude that there is no infinite sequence $\seq{\der_i}_{i \in \Nset}$ such that $\der_0 = \der$ and $\der_{i} \cutelim \der_{i+1}$.
	
	Now, let $\der= \der_0 \cutelim \cdots \cutelim\der_n$. 
	First, we show that  the number $n_{p}$ of its principal cut elimination steps  is bounded by $\weight{\der}\dfn   \mathsf{S}(\der)\cdot \weightcp{\der}+ \mathsf{M}(\der)$, 
	where  $\mathsf{M}(\der)$ is the number of inference rules different from $\cprule$ in $\der$. 
	This boils down to showing that $\der' \cutelim \der''$ implies $ \weight{\der''}<\weight{\der'}$.  Indeed:
	\begin{itemize}
		\item every cut elimination step cannot increase  $\mathsf{S}(\der)$
		\item every multiplicative  cut elimination step  decreases $\mathsf{M}(\der)$  and cannot increase $\weightcp{\der}$
		\item the exponential  steps $\cutstep{\cprule}{\wnwrule}$ and $\cutstep{\cprule}{\cprule}$  decrease $\weightcp{\der}$ and cannot increase $\mathsf{M}(\der)$ 
		\item if $\der' \cutelim \der''$ is obtained by applying  a $\cutstep{\cprule}{\wnbrule}$ step then 
		\[
		\def\arraystretch{1.2}
		\arraycolsep=2pt
		\begin{array}{rcl}
			\weight{\der''} &\dfn& \mathsf{S}(\der'') \cdot \weightcp{\der''}+ \mathsf{M}(\der'') \\
			&\leq &\mathsf{S}(\der') \cdot \weightcp{\der''}+ (\mathsf{M}(\der')-1 + \mathsf{S}(\der'))\\
			&=& \mathsf{S}(\der') \cdot (\weightcp{\der'}-1)+ \mathsf{M}(\der')-1 + \mathsf{S}(\der')\\
			&=& \mathsf{S}(\der') \cdot \weightcp{\der'}+ \mathsf{M}(\der')-1 < \weight{\der'} \\
		\end{array}
		\]
	\end{itemize}
	
	At the same time, 
	the number 
	$n^i_c$ of commutative steps performed after the $i$-th principal 
	is bounded by the square of the maximum size of the proof during  rewriting, which can be bounded by  $ \weight{\der}$.
	Hence, we have:
	$$
	\def\arraystretch{1.2}
	\arraycolsep=2pt
	\begin{array}{rcl}
		n  
		&= & n_{p}+\sum_{i=1}^{n_{p}}n^i_c 
		\leq 
		n_{p}+n_{p}\max_i \set{n^i_c }
		\\
		&\leq & 
		n_{p}\left(\max_i \set{n^i_c}+1\right)
		\leq 
		\weight{\der}\cdot (\weight{\der}^2+1) \\
		&\leq&
		2 \weight{\der}^3  
	\end{array}
	$$
	We conclude as $ \weight{\der}\in \mathcal{O}(\mathsf{S}(\der)\cdot\vert \der \vert)$ and $\size{\der_i}\leq \weight{\der_i}\leq \weight{\der}$.
\end{proof}

\begin{corollary}\label{lem:normalisation-confluence}
	$\cutelim$ over open derivations is strongly normalizing and confluent.
\end{corollary}
\begin{proof}
	Strong normalisation is a consequence of~\Cref{thm:cut-elimApp}.		Moreover, since cut elimination $\cutelim$ is strongly normalizing over open derivations 
	and it is locally confluent by  inspection of critical pairs, by Newman's lemma it is also confluent. 
\end{proof}

\subsection{The proof systems $\nupll$ and $\cpll$}\label{subsec:reg-weak-reg}
Starting from~\Cref{eqn:non-finite-support}, in~\cite{acclavio2023infinitary,CSL} we introduced two  proof systems, $\mathsf{wrPLL}^\infty$ and $\mathsf{rPLL}^\infty$ representing the non-wellfounded proof-theoretic counterparts of $\mathsf{PLL}$ and $\mathsf{nuPLL}$ respectively. Specifically, $\mathsf{wrPLL}^\infty$ and $\mathsf{rPLL}^\infty$ are obtained from $\nwpllprop$ by identifying additional global conditions called \emph{weak regularity} and \emph{finite expandability}. Roughly, weak regularity corresponds to a relaxation of the regularity property (see~\Cref{subsec:coderivations}) that allows us to discard those \nwboxes  with infinitely many distinct calls, so that only streams with finite support can be encoded. On the other hand, finite expandability discards those infinite branches whose sequents have an unbounded  number of $\oc$- and $\wn$-formulas. Indeed, $\cutr$ and $\wnbrule$ are the only rules that can increase that number (recall that  $\exists$ can only instantiate $(\wn, \oc)$-free formulas).

In this subsection, we define the second-order versions of $\mathsf{wrPLL}^\infty$ and $\mathsf{rPLL}^\infty$.

\begin{definition}\label{def:wR}[Proof systems $\nupll$ and $\cpll$]
	\label{defn:weak-regulairity}   	A coderivation is 
	\defin{weakly regular} if it has only finitely many distinct sub-coderivations 
	whose conclusions are left premises of $\cprule$-rules;
	it is
	\defin{\FE} if any branch contains finitely many $\cutr$ and $\wnbrule$ rules.
\defin{Weakly regular second-order parsimonious logic},  noted $\nupll$, is the set of progressing, \FE, and weakly regular coderivations of $\nwpll$. 
\defin{Regular second-order parsimonious logic},  noted $\cpll$, is the set of progressing, \FE, and  regular coderivations of $\nwpll$. 
\end{definition}

\begin{remark}\label{prop:weak-regular-finite-support}
	Regularity implies weak regularity and the converse fails (see 
	\Cref{ex:weak-regular} below), so $\cpll \subsetneq \nupll$.	
	A progressing and \FE $\der \in \nwpll$  is  regular (resp. weakly regular)
	if and only if any $\nwbox$ in $\der$ is periodic (resp.~has finite support).
\end{remark}

\begin{example}\label{ex:weak-regular}
	$\zeroder$ and $\wnder$ in {\Cref{eq:box}}  are weakly regular (they have no $\cprule$ rules) but not \FE (their only infinite branch has infinitely many $\cutr$ or  $\wnbrule$).
	The coderivation in \Cref{eqn:non-finite-support} is not weakly regular (it has infinitely many distinct~calls).
	
	An example of a weakly regular but not regular coderivation is  the  $\nwbox$ 
	$\derstream{\cod{i_0},\ldots,\cod{i_n}}$ in \Cref{ex:nonReg} when the infinite sequence 
	$(i_j)_{j \in \Nset} \in \set{\false,\true}^\omega$ is not periodic (see~\Cref{prop:weak-regular-finite-support}).
\end{example}

By inspecting   \Cref{fig:cut-elim-finitary,fig:cut-elim-pll} for $\nwpll$, we prove the following. 

\begin{restatable}{proposition}{preserves}\label{prop:cut-elim-preserves-finexp-reg-weakreg}
	Cut elimination preserves progressivity, weak-regularity, regularity and \FElity.
	Therefore, if $\dD\in\sysX$ with $\sysX\in\set{\cpll,\nupll}$ and
	$\der \cutelim \der'$, then also $\der' \in \sysX$.
\end{restatable}

	Akin to linear logic, we can recover a notion of \emph{depth} for coderivations, defined as the maximal number of ``nested" $\nwbox$s.

\begin{definition}[Nesting and depth]\label{def:depth}
	Let {$\der \in \nwpll$}.
	The \defin{nesting level of a rule} $\rrule$ in $\der$ is the number of nodes below it that are roots of a call of a $\nwbox$. The  \defin{nesting level of a $\nwbox$} is the nesting level of its bottommost rule. Finally, the  \defin{nesting level of a formula occurrence} in $\der$ is the nesting level of the rule whose conclusion contains it. 
	
	We say that a  rule (resp., formula occurrence, $\nwbox$) is \defin{shallow} if it has nesting level $0$. 
	The \defin{depth of $\der$}, written $\depth \der$, is the supremum of the nesting level of its rules.
\end{definition}

Notice that, if  $\nwpromotion$  is a $\nwbox$ with nesting level $n$, then the nesting level of its calls are $n+1$, and the nesting level of the formula occurrences in the main branch of $\nwpromotion$ is $n$. 
Note also that, although the depth of a coderivation can be infinite in general, weakly regular coderivations always have finite depth.

\begin{proposition}[\cite{CSL}]\label{lem:finite-depth}\label{lem:depth}
	If $\der$ is weakly regular then $\depth{\der} \in \Nset$. Moreover,  $\der \cutelim \der'$ implies $\depth{\der} \geq \depth{\der'}$.
\end{proposition}

\subsection{NL-decidability of $\cpll$}\label{app:properties}

We call a coderivation $\dD$  in $\nwpll$  \defin{weakly progressing} if every infinite branch contains infinitely many right premises of $\cprule$-rules. As already shown in~\cite{CSL} for the propositional setting, \prog and weak \prog coincide for finitely expandable coderivations.
\begin{restatable}{lemma}{infinitebrances}
	\label{lem:infinite-branches}\label{cor:simple-structure}
	Let $\der \in \nwpll$ be  \FE.  If $\der$ is weakly progressing then any infinite branch contains the main branch of a $\nwbox$.	
	Moreover, $\der \in \nwpll$  is progressing  if and only if it is weakly progressing.
\end{restatable}
\begin{proof}
	Clearly, a progressing coderivation is also weakly progressing. Now, 
	let 	 $\der\in \nwpll$ be \FE and weakly progressing, and let $\mathcal{B}$ be an infinite branch in $\der$. By \FElity there is  $h\in\Nset$ such that 
	$\mathcal{B}$ contains no conclusion of a $\cutr$ or $\wnbrule$ with height greater than $h$. 
	Moreover, by weak progressing condition  
	there is an infinite sequence $h\leq h_0 < h_1 <  \ldots <h_n <\ldots$ such that the sequent of $\mathcal{B}$ at height $h_i$ has shape $\wn \Gamma_i, \oc A_i$. 	By  inspecting the rules in~\Cref{fig:sequent-system-pll},  each such $\wn \Gamma_i, \oc A_i$ can  be the conclusion of either a $\wnwrule$ or a $\cprule$ (with right premise $\wn \Gamma_i, \oc A_i$). So, there is a $k$ large enough such that, for any $i \geq k$,  only the latter case applies (and, in particular, $\Gamma_i=\Gamma$ and $A_i=A$ for some $\Gamma, A$). Therefore, $h_k$ is the root of a  \nwbox. This also shows that $\der$ is progressing.
\end{proof}

Moreover, by an argument similar to~\cite[Corollary 32]{CurziDas} we have

\begin{restatable}{proposition}{NL}
	It is $\nl$-decidable if a regular coderivation is in $\cpll$.
\end{restatable}
\begin{proof}
	A regular coderivation is represented by a finite cyclic graph. By~\Cref{cor:simple-structure} checking progressivity comes down to checking that no branch has infinitely many occurrences of a particular rule, which in turn   reduces to checking acyclicity for this graph (see \cite{CurziDas}).
	We conclude since checking acyclicity is a well-known $\conl$ problem, and $\conl=\nl$\cite{arora_barak_2009}.
\end{proof}

Of course a similar decidability result cannot hold for $\nupll$,  this proof system containing continuously many coderivations, as hinted by the $\nwprule$ depicted in~\Cref{ex:nonReg}.

\begin{figure*}[t]
	\adjustbox{max width=\textwidth}{$\begin{array}{ccc}
			%
			\fproj{
				\vlderivation{
					\vlin{\fprule}{}
					{\wn \Gamma, \oc A}
					{\vldr{\der}{\Gamma, A}}
				}
			}
			=
			\vlderivation{
				\vliin{\cprule}{}
				{\wn \Gamma, \oc A}
				{\vldr{\fprojp{\der}}{ \Gamma, A}}
				{
					\vliin{\cprule}
					{}
					{\wn \Gamma, \oc A}
					{\vldr{\fprojp{\der}}{ \Gamma, A}}
					{\vlin{\cprule}{}{\wn \Gamma, \oc A}{\vlhy{\vdots}}}
				}
			}
			&
			\qquad
			&
			\cproj{
				\vlderivation{
					\vliiiin{\nuprule}{}{\wn \Gamma, \oc A}{
						\vldr{\der_{0}}{\Gamma,A}
					}{
						\vlhy{\cdots}
					}{
						\vldr{\der_{n}}{\Gamma,A}}{\vlhy{\cdots}
					}
				}
			}
			=
			\vlderivation{
				\vliin{\cprule}{}{\wn \Gamma, \oc A}{
					\vldr{\der_0^\bullet}{\Gamma, A}
				}{
					\vliin{\cprule}{}{
						\reflectbox{$\ddots$}
					}{
						\vldr{\der_n^\bullet}{\Gamma, A}
					}{
						\vlin{\cprule}{}
						{{\wn \Gamma, \oc A}}
						{\vlhy{\vdots}}
					}
				}
			}
		\end{array}$}
	\vspace{-10pt}
	\caption{
		Translations  $\fprojp{(\cdot)}\colon\pll\to\cpll$, 
		and
		$\cprojp{(\cdot)}\colon\dpll\to\nupll$.
	}
\label{fig:translations-pll}
\end{figure*}

\begin{figure}
    \centering
    \[
        \tiny
    \begin{array}{ccc}
      \vlderivation{
		\vliin{\cutr}{}{\wn \Gamma, \wn \Delta, \oc B}{
			\vlin{\fprule}{}{\wn \Gamma, \oc A}{
				\vldr{\der}{\Gamma, A}
			}
		}{
			\vlin{\fprule}{}{\wn \cneg {A}\!, \wn \Delta, \oc B}{
				\vldr{\der'}{\cneg {A}\!, \Delta, B}
			}
		}
	}
&\overset{\fprojp{{(\_)}}}{\longrightarrow}&
   \vlderivation{
   \vliin{\cutr}{}{\wn \Gamma, \wn \Delta, \oc B}
   {
				\vliin{\cprule}{}
				{\wn \Gamma, \oc A}
				{\vldr{\fprojp{\der}}{ \Gamma, A}}
				{
					\vliin{\cprule}
					{}
					{\wn \Gamma, \oc A}
					{\vldr{\fprojp{\der}}{ \Gamma, A}}
					{\vlin{\cprule}{}{\wn \Gamma, \oc A}{\vlhy{\vdots}}}
				}
                }
                {
				\vliin{\cprule}{}
				{\wn \cneg A, \wn \Delta, \oc B}
				{\vldr{\fprojp{\der'}}{\cneg A, \Delta, B}}
				{
					\vliin{\cprule}
					{}
					{\wn \cneg A, \wn \Delta , \oc B}
					{\vldr{\fprojp{\der'}}{\cneg A,  \Delta, B}}
					{\vlin{\cprule}{}{\wn \cneg A,\wn \Delta, \oc B}{\vlhy{\vdots}}}
				}
                }
			}  
 \\  \\ \rotatebox[origin=c]{270}{$\cutelim$} &&  \rotatebox[origin=c]{270}{$\cutelim^{\omega}$} \\ \\
 \vlderivation{
		\vlin{\fprule}{}{\wn \Gamma, \wn \Delta, \oc B}{
			\vliin{\cutr}{}{\Gamma, \Delta, B}{
				\vldr{\der}{\Gamma, A}
			}{
				\vldr{\der'}{\cneg {A}\!, \Delta, B}
			}
		}
	}  

    &\overset{\fprojp{{(\_)}}}{\longrightarrow}&
            \vlderivation{
            	\vliin{\cprule}{}
				{\wn \cneg{A},\wn \Gamma, \oc B}
				{\vliin{\cutr}{}{\Gamma, \Delta, B}{
				\vldr{\fprojp{\der}}{\Gamma, A}
			}{
				\vldr{\fprojp{\der'}}{\cneg {A}\!, \Delta, B}
			}}
				{
					\vliin{\cprule}
					{}
					{\wn \cneg{A},\wn \Gamma, \oc B}
					{\vliin{\cutr}{}{\Gamma, \Delta, B}{
				\vldr{\fprojp{\der}}{\Gamma, A}
			}{
				\vldr{\fprojp{\der'}}{\cneg {A}\!, \Delta, B}
			}}
					{\vlin{\cprule}{}{\wn \cneg{A},\wn \Gamma, \oc B}{\vlhy{\vdots}}}
				}
                }
    \end{array}
    \]
    \caption{Simulation of $\cutstep{\fprule}{\fprule}$ in $\cpll$ via $\fprojp{{(\_)}} $.}
    \label{fig:simulation-cp-vs-cp}
\end{figure}

\subsection{Simulation results}
We conclude this section by showing that all functions representable in  $\pll$ and  $\dpll$ are also representable in $\cpll$ and  $\nupll$ respectively (\Cref{thm:simulation}). To this end, we first prove that a cut elimination sequence in $\pll$ or $\dpll$ can be simulated by a $\omega$-long cut elimination sequence in their non-wellfounded counterparts (\Cref{lem:construction-of-omega-reduction-sequence}). \Cref{thm:simulation} is then proved by observing that, when derivations in $\pll$ and $\dpll$ have $\oc$-free conclusion, cut elimination sequences that fully eliminate the cut rule can be, in fact, simulated by \emph{finite} cut elimination sequences on coderivations (\Cref{lem:construction-of-finitary-reduction-sequence}).

We begin with some useful structural properties. 

\begin{lemma}\label{prop:downward-oc}
	Let $\rrule \in \pll \cup \dpll \cup \nwpll$ be a rule such that either $\rrule\neq  \cutr$ or  $\rrule\in\{\cutstep{\fprule}{\fprule}, \cutstep\nuprule\nuprule, \cutstep{\cprule}{\cprule}\}$. 
	If a $\oc$ occurs in $\rrule$, then a $\oc$ occurs in its conclusion.
\end{lemma}
\begin{proof}
	If $\rrule = \cutr$ then it is in  $\{\cutstep{\fprule}{\fprule}, \cutstep\nuprule\nuprule, \cutstep{\cprule}{\cprule}\}$ by hypothesis, and in all such  cases $\rrule$ contains a $\oc$-formula in the conclusion. Otherwise,  $\rrule$ is not a $\cutr$, and the property follows by inspecting the other rules of $\pll \cup \dpll \cup \nwpll$, recalling that instantiation in the $\exists$ rule requires $\oc$-freeness. 
\end{proof}

\begin{lemma}\label{prop:cut-free-oc-free-derivation}
Any cut-free progressing $\der \in \nwpll$ with a $\oc$-free conclusion is a derivation.
\end{lemma}
\begin{proof}
    By progressivity, every infinite branch of $\der$ would contain a sequent with an occurrence of $\oc$. Since $\der$ is cut-free, by repeatedly applying \Cref{prop:downward-oc} we  have that the conclusion of $\der$ must contain an occurrence of $\oc$, contradicting the hypothesis.
\end{proof}

Simulation of cut elimination relies on two translations for $\pll$ and $\dpll$ into their non-wellfounded formulations $\cpll$ and $\nupll$, respectively.

\begin{definition}[Translation]
	We define   two (conclusion-preserving) translations  
	$\fprojp{(\cdot)}\colon\pll\to\cpll$ 
	and
	$\cprojp{(\cdot)}\colon\dpll\to\nupll$, 
	which  expand bottom-up the promotion rules $\fprule$ and $\nuprule$
	into \nwboxes as in~\Cref{fig:translations-pll}, leaving the other rules unchanged.
\end{definition}

Note that the images of the translations $\fprojp{(\cdot)}$ and $\cprojp{(\cdot)}$ are in $\cpll$ and $\nupll$, respectively, by~\Cref{prop:weak-regular-finite-support}.

	Observe that if $\der_1 \cutelim \der_2$ is a cut elimination step of the form $\cutstep\fprule\fprule$ in $\pll$ then $\fprojp{\der_2}$ can only be obtained from
 $\fprojp{\der_1}$ by applying {infinitely} many cut elimination steps in $\cpll$, as shown in~\Cref{fig:simulation-cp-vs-cp}, and similarly for the cut elimination step  $\cutstep\nuprule\nuprule$ in $\dpll$. 	
Nonetheless, we can show that coderivations of $\pll$ and $\dpll$  with $\oc$-free conclusion  can be turned into cut-free coderivations  using only \emph{finitely many} cut elimination steps. As a straightforward consequence, we can infer that any function over binary strings representable in $\pll$ (resp.~$\dpll$) is also representable in  $\cpll$ (resp.~$\nupll$).

To this end, we introduce some definitions inspired by~\cite{Saurin23}. In what follows, coderivations in $\nwpll$ will be equipped with a \emph{distance} $\delta \colon \nwpll \times \nwpll \to \Nset$  given by 
$$
\delta(\der, \der') =
\begin{cases}
	0 &\text{if $\der = \der'$;}
	\\
	\min\{2^{-h}\mid \der \text{ and } \der' \text{ coincide in all nodes up to height }h \} &\text{otherwise.}
\end{cases}
$$
Note that this is well defined even when $\der$ or $\der'$ is a derivation.
As well known, this distance forms a complete (ultra)metric space over any set of (binary, possibly infinite) labeled trees, inducing the so-called \emph{tree~topology}.
So, sequences of coderivations in $\nwpll$ may have \emph{limits}.

In the next definition, we identify $\Nset$ with the least limit ordinal $\omega$.

\begin{definition}[$\lambda$-reduction sequence] 
Let $\lambda \in \Nset \cup \{\omega\}$.  A \defin{$\lambda$-reduction sequence} is a $\lambda$-indexed sequence $\sigma\dfn (\der_i)_{i \in \lambda}$ for any $i$ such that $i+1 \in \lambda$. 
The height of the $\cutr$ rule reduced at the cut elimination step $\der_i \cutelim \der_{i+1}$ in $\sigma$ is denoted by~$h_\sigma(i)$. 
We say that $\sigma $ is \defin{height-increasing} if $\lim_{i \in \lambda} h_\sigma(i)=\infty$. 
It is \defin{weakly converging} if $\lim_{i \in \lambda} \der_i$ exists; we then write  $\sigma \dfn \der_0 \cutelim^\lambda \der$ to mean that $\lim_{i \in \lambda} \der_i= \der$. 
    Finally, $\sigma$ is \defin{strongly converging} if it is weakly converging and height-increasing. 
\end{definition}

If $\lambda \in \Nset$, any $\lambda$-reduction sequence is weakly converging and not height-increasing.

\begin{definition}[Splitting function]
Given a height-increasing $\omega$-reduction sequence $\sigma$, a strictly monotone function 
    $\ell \colon \Nset \to \Nset$ is a \defin{splitting function (for $\sigma$)} if, for all $j \in \Nset$:
    \begin{itemize}
        \item $h_\sigma(i)\leq j$ for all $ i \leq\ell(j)$;
        \item $h_\sigma(i) > j$ for all $i > \ell(j)$. 
    \end{itemize}
\end{definition}

Note that every height-increasing $\omega$-reduction sequence $\sigma$ has a splitting function.
Indeed, since $\lim_{i \in \omega} h_\sigma(i) = \infty$, for every $j \in \Nset$ there is $n_j \in \Nset$ such that, for all $i \in \omega$, if $i > n_j$ then $h_\sigma(i) > j$, and $h_\sigma(i) \leq j$ otherwise; 
we can shift the $n_j$'s so that the sequence $(n_j)_{j\in \Nset}$ is strictly increasing;
a splitting function for $\sigma$ is then $\ell \colon \Nset \to \Nset$ defined by $\ell(j) = n_j$.

\begin{lemma}[$\omega$-compression]\label{lem:compression}
  Let $\sigma\dfn \der_0 \cutelim^\omega \der_\omega$ be a strongly converging $\omega$-reduction sequence in $\nwpll$, and let $\ell$ be a splitting function for $\sigma$. Then:
\begin{enumerate}
    \item \label{enum:compression-1} If $\der_\omega \cutelim \der_{\omega+1}$ reduces a cut with height $k$ {there is a $\lambda$-reduction sequence $\sigma'\dfn \der_0 \cutelim^{\lambda} \der_{\omega+1}$ such that, if $\lambda=\omega$:
    \begin{itemize}
        \item $\sigma'$ is  strongly converging 
        \item there is  a splitting function  $\ell'$ for $\sigma'$ such that $\sigma(j)=\sigma'(j)$ for all $j \leq \ell(k)\leq \ell'(k)$.
    \end{itemize}}
    \item \label{enum:compression-2} If $\tau \dfn \der_\omega \cutelim^{\omega} \der_{\omega\cdot2}$ is a strongly converging $\omega$-reduction sequence such that $(h_\tau(\omega+i))_{i \in \omega}$ is strictly increasing, {then  there is a $\lambda$-reduction sequence $\sigma^* \dfn \der_0 \cutelim^{\lambda} \der_{\omega\cdot 2}$ such that, if $\lambda=\omega$ then $\sigma^* $ is strongly converging.}
\end{enumerate}
\end{lemma}
\begin{proof}~
    \begin{enumerate}
    	\item Let $\ell$ be a splitting function for $\sigma$ and let $k$ be the height of  the cut reduced by the cut elimination step $\der_\omega \cutelim \der_{\omega+1}$. {We have two cases. If the cut elimination step is not of the form $\cutstep{\cprule}{\wnwrule}$ then it commutes with any step $\der_i \cutelim\der_{i+1}$ with $i > \ell(k)$, and so we can construct a strongly converging $\omega$-reduction sequence $\sigma'\dfn \der \cutelim^\omega \der_{\omega+1}$. Notice that the function $\ell'$ defined by $\ell'(j)\dfn \ell(j)+1$ for all $j \geq k$ and $\ell'(j)\dfn \ell(j)$ otherwise is splitting for $\sigma'$. It is easy to see that $\sigma(j)=\sigma'(j)$ for all $j \leq \ell(k)\leq \ell'(k)$. Otherwise, the cut elimination step $\der_\omega \cutelim \der_{\omega+1}$ is of the form $\cutstep{\cprule}{\wnwrule}$. In this case  it only \emph{weakly} commutes with any step $\der_i \cutelim\der_{i+1}$ with $i > \ell(k)$. Indeed, a cut elimination step $\cutstep{\cprule}{\wnwrule}$ erases possibly infinite inference rules, and so anticipating it in the sequence $\sigma$ might prevent infinitely many reduction steps being applied. Therefore, we can construct a $\lambda$-reduction sequence $\sigma'\dfn \der_0 \cutelim^{\lambda} \der_{\omega+1}$  satisfying the required conditions.
        }
    	
    	\item {We have two cases. If by  repeatedly applying~\Cref{enum:compression-1} we obtain a $\lambda$-reduction sequence $\sigma': \der_0 \cutelim^\lambda \der_{\omega+n}$ where $\lambda\in \omega$, then we simply define $\sigma^*$ by concatenating $\sigma'$ with $\tau(\omega+n)\cutelim^{\omega}\der_{\omega \cdot 2}$. Clearly, $\sigma^*$ is a strongly convergent $\omega$-reduction sequence, and so we are done. Otherwise, 
        }
         we can construct a family $(\sigma_n \dfn \der_0 \cutelim^\omega \der_{\omega+n})_{n \in \omega}$ of strongly convergent $\omega$-reduction sequences and a family $(\ell_n)_{n \in \omega}$ such that $\ell_n$ is a splitting function for $\sigma_n$. For the base case, we  set $\sigma_0 \dfn \sigma$ and $\ell_0$ any splitting function for $\sigma$. Concerning the inductive case,  $\sigma_{n+1}$ and $\ell_{n+1}$ are obtained by applying~\Cref{enum:compression-1} to $\sigma_n$, $\ell_n$, and $\der_{\omega+n}\cutelim \der_{\omega+n+1}$. Note that $\sigma_{n}(j)=\sigma_{n+1}(j)$ for all $j \leq \ell_n(h_{\tau}(n))\leq \ell_{n+1}(h_{\tau}(n))$ by construction of $\sigma_{n+1}$. 
    	
    	We  now consider the sequence $\sigma^*$ of length $\omega$ defined as $\sigma^*(i)\dfn \lim_n \sigma_{n}(i)$ for all $i \in \omega$. Notice that the limit defining $\sigma^*(i)$ exists. Indeed, by construction $(\sigma_n(i))_{n\in \omega}$ is eventually constant. 
        Indeed,  since $(h_\tau(n))_{n \in \omega}$ is strictly increasing and each $\ell_{n}$ is strictly monotone, for every $i \in \omega$ there is  $n_0 \in \omega$ such that $i \leq \ell_{n_0}(h_\tau(n_0))$. Moreover, by construction $\ell_{n_0}(h_\tau(n_0))\leq \ell_{m}(h_\tau(n_0))$ for every $m \geq n_0$ and so   $\sigma_{m}(i)=\sigma_{n_0}(i)$. Hence, $\sigma^*(i)= \lim_n \sigma_{n}(i)= \sigma_{n_0}(i)$. This also shows that $\sigma^*$ is an $\omega$-reduction sequence. 
    	
        Moreover, the sequence is weakly converging. 
        Indeed, we have that  $\lim_i \sigma^*(i)=\lim_i \lim_n \sigma_{n}(i)=\lim_n \lim_i \sigma_{n}(i)= \lim_n \der_{\omega+n}= \der_{\omega \cdot 2}$. To show that it is also strongly converging we need to prove that it is height-increasing. Notice that by construction, for every $n \in \omega$,  $\sigma^*$ contains a cut elimination step reducing a cut of height $h_\tau(n)$. Since by hypothesis  $(h_\tau(n))_{n \in \omega}$ is strictly increasing, it must be that $\lim_{i \in \omega} h_{\sigma^*}(i)= \infty$.
    	\qedhere
    \end{enumerate}
    
\end{proof}

\begin{lemma}[$\omega$-simulation]\label{lem:construction-of-omega-reduction-sequence}
     Let $\der$ be a derivation of $\pll$ (resp., $\dpll$). If $\der \cutelims\der'$ then there is a $\lambda$-reduction sequence $\sigma\dfn \fprojp{\der} \cutelim^{\lambda}\fprojp{\der'}$ (resp., $\sigma\dfn \cprojp{\der} \cutelim^{\lambda}\cprojp{\der'}$) with $\lambda\leq \omega$. Moreover, if $\lambda=\omega$ then the sequence is strongly converging. 
\end{lemma}
\begin{proof}
    We prove the statement by induction on the length of $\der \cutelims\der'$. We only consider the case where $\der$ is a derivation of $\pll$, as the case for $\dpll$ can be treated similarly. If $\der= \der'$ then we are done. Otherwise, we have $\der \cutelims\der'' \cutelim\der'$. By induction hypothesis, there is a $\lambda$-reduction sequence $\sigma''\dfn \fprojp{\der} \cutelim^{\lambda}\fprojp{\der''}$ satisfying the conditions. We do case analysis on $\der'' \cutelim \der'$. If the cut elimination step reduces a cut $\rrule \neq \cutstep{\cprule}{\cprule}$ then it is easy to check that $\fprojp{\der''}\cutelim \fprojp{\der'}$. So, if  $\lambda= \omega$ then $\sigma''$ is strongly converging, and we can apply~\Cref{lem:compression}.\ref{enum:compression-1}. Otherwise, $\rrule = \cutstep{\cprule}{\cprule}$ and we simulate $\der'' \cutelim \der'$ by an $\omega$-reduction sequence $\sigma'\dfn \fprojp{\der''}\cutelim^\omega \fprojp{\der'}$  as in~\Cref{fig:simulation-cp-vs-cp}. It is easy to see that  $\sigma'$ can be constructed as a strongly converging  $\omega$-reduction sequence and,  in particular, such that $(h_{\sigma'}(i))_{i \in \omega}$ is strictly increasing. So, if $\lambda= \omega$ we conclude by applying~\Cref{lem:compression}.\ref{enum:compression-2}.
\end{proof}

\begin{lemma}[Finite simulation]\label{lem:construction-of-finitary-reduction-sequence} Let $\der$ be a derivation of $\pll$ (resp., $\dpll$) with $\oc$-free conclusion. If $\der \cutelims \widehat{\der}$ with $\widehat{\der}$ cut-free,  then $\fprojp{\der}\cutelims \fprojp{\widehat{\der}}$ (resp., $\cprojp{\der}\cutelims \cprojp{\widehat{\der}}$).    
\end{lemma}
\begin{proof}
We only consider the case where $\der$ is a derivation of $\pll$, as the case for $\dpll$ can be treated similarly. By~\Cref{lem:construction-of-omega-reduction-sequence} we obtain a $\lambda$-reduction sequence  $\sigma\dfn \fprojp{\der_0}\cutelim^{\lambda}\fprojp{\hat{\der}}$ where $\lambda\leq \omega$ such that, if $\lambda= \omega$ then $\sigma$ is strongly converging. Since $\widehat{\der}$ is cut-free and has $\oc$-free conclusion, then it is a derivation by~\Cref{prop:cut-free-oc-free-derivation}.  This implies that  $\lambda<\omega$. Indeed, if $\lambda= \omega$ then  $\sigma$ would be strongly converging, and so height-increasing. But then  $\lim_i h_\sigma(i)= \infty$, which contradicts finiteness of $\widehat{\der}$.
\end{proof}

\begin{theorem}[Simulation] \label{thm:simulation}  
	Let $f \colon (\{\false,\true\}^*)^n \to \{\false,\true\}^*$.
	\begin{enumerate}
		\item \label{enum:simulation1} 
		If $f$ is representable in $\dpll$, 
		then so it is  in $\nupll$.
		\item  \label{enum:simulation2} 
		If $f$ is representable in $\pll$,  then so it is in $\cpll$.
	\end{enumerate}
\end{theorem}
\begin{proof}
	We  only prove~\Cref{enum:simulation1}, as~\Cref{enum:simulation2} is proven similarly. Let   $\der: \String[] \limp \overset{n\geq 0}{\ldots} \limp \String[] \limp \String$ represent $f \colon (\{\false,\true\}^*)^n \to \{\false,\true\}^*$ in $\dpll$ and let  $x_1, \ldots, x_n \in \{\false,\true\}^*$. This means that the reduction in~\Cref{fig:representability} holds for $\mathbf{T_1}=\ldots= \mathbf{T}_n= \String[]$ and $\mathbf{T}=\String$.	Let $\sigma\dfn   \der_0 \cutelim \der_1 \cutelim \ldots \cutelim \der_n=\cod{f(s_1, \ldots, s_n)} $ be such a reduction  sequence, and notice that  $\cod{f(s_1, \ldots, s_n)}$ is cut-free by definition.  
    We conclude by applying~\Cref{lem:construction-of-finitary-reduction-sequence}, observing that $\cod{s} ^\bullet=\cod{s} $ for any binary string $s \in \{\false, \true\}^*$.
\end{proof}

\section{Soundness}\label{subsec:soundness}

In this section we prove the soundness theorem, which states that every function $f$ over binary strings  representable by coderivations of $\nupll$ (resp.~$\cpll$) is  in  $ \fppoly$ (resp.~$\fptime$). Soundness is a straightforward consequence of a polynomial modulus of continuity result for cut elimination (\Cref{prop:polynomial-moduli}), according to which  we can extract a family of polynomial size circuits $C_f$ computing $f$. We conclude by observing that  $C_f$  is, in fact,  P-uniform whenever the coderivation representing $f$ is regular.

\newcommand{\derround}[1]{\der^{#1}}
\newcommand{\derphase}[1]{\der_{\mathsf{e}}^{#1}}	
\newcommand{\mphase}[2]{{#1}\cutelim^{*\mathsf{m}} {#2}}
\newcommand{\ephase}[2]{{#1}\cutelim^{*\mathsf{e}} {#2}}
\newcommand{\round}[2]{{#1}\cutelim^{*\mathsf{r}} {#2}}

\subsection{Shallow cut elimination strategy}\label{subsec:shallow}

Eliminating cuts in coderivations of   $\nupll$ and $\cpll$ typically requires infinitary rewriting (see, e.g.,~\cite{CSL}). However, if we focus on coderivations with $\oc$-free  conclusion we can define a cut elimination strategy, we call it \emph{shallow}, that always halt after a  \emph{finite} number of steps. This restricted form of cut elimination is enough for establishing soundness, in that computation over binary strings can  be duly simulated by reducing cuts in such coderivations.

The so-called \emph{shallow cut elimination strategy} for $\oc$-free coderivations of $\nupll$ and $\cpll$ is divided into rounds, each one divided into two phases. Phase 1 reduces all shallow cuts (i.e., those cuts with nesting level $0$) that do not involve $\nwboxes$. The shallow cuts affecting $\nwboxes$,  called \emph{steady cuts}, are dealt with in Phase 2. This phase reduces hereditarily all steady cuts produced by Phase 1 except those of the form $\cutstep{\cprule}{\cprule}$ (as reducing the latter would ``merge two boxes" and can be  avoided thanks to~\Cref{prop:invariance}.\ref{enum:invariance2}).

\begin{definition}[Steady cuts and $\nwboxes$] 
	Let {$\der \in \nwpll$}. A \defin{steady cut} is a  shallow cut with  an active formula that is $\oc$-principal for a (shallow) $\nwbox$ $\nwpromotion$. We call  $\nwpromotion$ a  \defin{steady $\nwbox$}.
\end{definition}

\begin{proposition}\label{prop:advanced-properties-exp-flow}
	Let $\der \in \nupll$ be a coderivation  of a   $\oc$-free sequent  whose shallow cuts are all steady. Then:
	\begin{enumerate}
		\item  \label{enum:steady1} All shallow $\nwboxes$ are all steady.
		\item   \label{enum:steady2} There is a cut $\rrule\neq \cutstep{\cprule}{\cprule}$.
	\end{enumerate}
\end{proposition}
\begin{proof}
	The proof of~\Cref{enum:steady1}  is analogous to~\Cref{prop:cut-free-oc-free-derivation}. If there is a shallow $\nwbox$ $\nwpromotion$ 
that	 is not steady and all shallow cuts are steady then, by  \Cref{prop:downward-oc},  the conclusion of $\der$ contains a $\oc$, which contradicts the hypothesis.
	
	Concerning~\Cref{enum:steady2}, by~\Cref{enum:steady1} there are exactly $n$ steady cuts and $n$ steady $\nwboxes$. By the tree structure of $\der$ there must be a shallow cut that with an active formula that is not in the conclusion of a shallow $\nwbox$.
\end{proof}

\begin{definition}[Shallow rewriting strategy] 
	Let $\der\in \nupll$ be a coderivation of a $\oc$-free sequent. The \defin{shallow cut elimination strategy} iterates $\depth{\der}+1$ times the two phases below\footnote{To make the strategy deterministic, we can give priority   to the rightmost reducible $\cutr$ with smallest height. This would ensure that the strategy eventually applies a cut elimination step to every reducible cut.}:
	\begin{itemize}
		\item  \textbf{Phase 1.} Reduce all shallow cuts that are not  steady.
		%
		%
		\item  \textbf{Phase 2.} Fully reduce all steady cuts $\rrule \neq \cutstep{\cprule}{\cprule}$ except those that become shallow during this phase. 
		
		
		
	\end{itemize}
	We call each iteration of the two phases above a \defin{round}.  	We set $\derround 0\dfn \der$ and, for all $1 \leq d \leq \depth{\der}+1$,  $\derround d$ to be the coderivation obtained after the $d$-th round. For all $1 \leq d \leq \depth{\der}+1$ we set  $\derphase d$ as the coderivation obtained from $\derround{d-1}$  by applying \textbf{Phase 1}.  We write 
	$\mphase{\derround d}{\derphase {d+1}}$ to denote that $\derphase {d+1}$ has been obtained in a finite number of steps from 
	$\derround d$ by applying  \textbf{Phase 1}; $\ephase{\derphase d}{\derround d}$ to denote that
	$\derround d$ has been obtained in a finite number of steps from $\derphase d$
	by applying \textbf{Phase 2}; finally, we set $\round{\derround d}{\derround {d+1}}\dfn  \ephase{\mphase{\derround d}{\derphase {d+1}}}{\derround {d+1}}$.  
\end{definition}

In the next subsections we will introduce the technical definitions and results to prove that the shallow cut elimination strategy applied to coderivations with $\oc$-free conclusion always terminates. The idea behind termination is that after each round all steady $\nwboxes$ are eventually erased, so that the depth of the coderivation decreases by 1. This happens because during each round all steady cuts will eventually be reduced to a $\cutstep{\cprule}{\wnwrule}$ cut.

\subsection{Decomposition prebar and truncations}

As already noticed in~\cite{acclavio2023infinitary,CSL} in a propositional setting, thanks to finite expandability and (weak) regularity,  coderivations of $\nupll$ and $\cpll$ can be ``decomposed" into a finite tree together with a finite number of   \nwboxes. Such a decomposition property will allow an inductive description of   coderivations for our non-wellfounded  proof systems.

Recalling~\Cref{defn:bars}, the following is a straightforward consequence of~\Cref{cor:simple-structure}:
\begin{restatable}{proposition}{BAR}
	\label{prop:canonicity} 
	Let $\der \in \nupll$. There  is a prebar $\Nodes\subseteq \set{1,2}^*$ of  $\der$
	such that each $v\in\Nodes$ is the root of a \nwbox.
\end{restatable}

\begin{definition}
	\label{def:decomp}
	Let $\der\in  \nupll$. 
	The \defin{decomposition prebar}
	of $\der$ is the
	minimal prebar $\Nodes$ of $\der$ such that, for all $v \in \Nodes$,  
	$\der_{v}$ is a $\nwbox$. We denote with $\bord \der$ such a prebar and we set     
	$\base{\der}\coloneqq\prun{\der}{\bord\der}$.
\end{definition}

\begin{remark}\label{rem:finite-derivation}
If  $\der\in  \nupll$   then,  by weak K\"{o}nig's lemma, $\bord \der$  is finite and $\base{\der}$ is a finite approximation of $\der$. 
\end{remark}

{The \emph{$n$-truncation} of a coderivation $\der$ is a particular finite approximation of $\der$ that  only considers the first $n$ calls of each $\nwprule$}. Similarly, an \emph{$n$-hypertruncation} only considers the first $n$-calls of the shallow \nwboxes of $\der$.

\begin{figure}
	\[
	\vlderivation{
		\vliin{\cprule}{}{\wn \Gamma, \oc A}{\vltr{\der_0}{\Gamma, A}{\vlhy{ \  \ }}{\vlhy{\  \ }}{\vlhy{\ \   }}}{
			\vliin{\cprule}{}{\wn \Gamma, \oc A}{\vltr{\der_1}{\Gamma, A}{\vlhy{\  \ }}{\vlhy{\  \ }}{\vlhy{\ \  }}}{		\vliin{\cprule}{}{\reflectbox{$\ddots$}}{\vltr{\der_{n-1}}{\Gamma, A}{\vlhy{\ \  }}{\vlhy{\  \ }}{\vlhy{\ \  }}}{
					\vliin{\cprule}{}{\wn \Gamma, \oc A}{\vlin{\zero}{}{\Gamma, A}{\vlhy{}}}{\vlin{\zero}{}{\wn \Gamma, \oc A}{\vlhy{}}}
				}
			}
		}
	}
	\]
	\caption{A finite non-wellfounded promotion.}
	\label{fig:finNWP}
\end{figure}

\begin{definition}[Finite $\nwboxes$ and truncations]\label{defn:trunc}
	A \defin{finite non-wellfounded promotion} is defined as a coderivation $\finitenwpromotion= \derfinstream{\der_0}{\der_1}{\der_{n-1}}$   in \Cref{fig:finNWP}.
	We  write $\finitenwpromotion(i)$ to denote $\der_i$.

	Let $\der\in \nupll$ with $\bord \der= \set{v_1, \ldots, v_k}$ (so  $\nwpromotion_i \coloneqq \der_{v_i}$ is a $\nwprule$ for all $1 \leq i \leq k$).  
	The \defin{$n$-truncation} $\prun \der n$ and the \defin{$n$-hypertruncation} $\trunk \der n$ of $\der$ are the open derivations defined for all $n>0$ as follows: if $\depth \der=0$, then $\prun \der n=\trunk \der n\dfn \base{\der}= \der$, and  if $\depth \der >0$, then 
	\begin{align*}
		\prun \der n &\dfn \base{\der}(\vec{\finitenwpromotion_i}/\vec{v_i})
		&
		\trunk \der n &\dfn \base{\der}(\vec{\finitenwpromotion'_i}/\vec{v_i})
	\end{align*}	
	where for all $i\in\intset1k$, $\finitenwpromotion_i =  \derfinstreamshort{\prun{\nwpromotion_i(0)}n}{\prun{\nwpromotion_i(1)}n}{\prun{\nwpromotion_i(n-1)}n}$ and $\finitenwpromotion_i' =  \derfinstreamshort{\base{\nwpromotion_i(0)}}{\base{\nwpromotion_i(1)}}{\base{\nwpromotion_i(n-1)}}$.
\end{definition}

Notice  that $\prun \der n$ and $\trunk \der n$ are  finite,  and  $\trunk \der n\preceq \prun \der n$. 

\subsection{Exponential flows} 

We now   introduce the  \emph{exponential graph} of a coderivation $\der$, a  directed graph associated to $\der$ that allows  a static analysis of the cut elimination steps reducing steady cuts.  Directed paths of this graph, called \emph{exponential flows}, can be then used to precompute  the maximum number of calls of  a shallow $\nwprule$ that eventually become shallow as a consequence of a $\cutstep\cprule\wnbrule$ cut elimination step. This number will be called \emph{rank} of $\der$, and plays a crucial role for establishing a polynomial bound on cut elimination.  

\newcommand{\cpp}{\mathsf{p}}
\newcommand{\cpa}{\mathsf{a}}
\newcommand{\expflow}{\mathcal{\phi}}
\newcommand{\measure}[2]{{#1}({#2})}
\newcommand{\res}{\widehat{\expflow}}
	
\begin{definition}[Exponential flow] 
	Let $\der\in \nupll$. The \defin{exponential graph of $\der$}, written $\egraph \der$, is a finite directed forest\footnote{A directed forest is a directed acyclic graph whose underlying undirected graph is a forest.} whose nodes are (labelled by) the shallow exponential formulas in $\base{\der}$ and whose directed edges connect a node $A$ to a node $B$ if:
	\begin{itemize}
		\item $A=\wn \cneg{C}$ and  $B=\oc C$ are the conclusions of an $\axr$ rule;
		\item $A=\wn \cneg{C}$ and $B=\oc C$ are conclusions of a $\cprule$ rule;
		\item $A=\wn C$ is principal for a $\wnbrule$-rule with active formula $B=\wn C$;
		\item $A=\oc {C}$ and $B=\wn \cneg{C}$ are the cut-formulas of a $\cutr$ rule;
		\item  $A=B$, where $A$ is a non-principal  $\wn$-formula in the conclusion of a rule $\rrule \neq \cprule$, and $B$ is the corresponding non-active $\wn$-formula in a premise of $\rrule$.
		\item $A=B$, where $B$ is a non-active  $\oc$-formula in a premise of a rule $\rrule \neq \cprule$, and $A$ is the corresponding non-principal $\oc$-formula in the conclusion of $\rrule$.
	\end{itemize}
	A \defin{$\oc$-node}  (resp.~\defin{$\wn$-node}) is a node labelled by a $\oc$-formula (resp.~$\wn$-formula). 
	A \defin{$\brule$-node} (resp.~\defin{$\wrule$-node}) is a   node labelled by  the  principal formula  for a $\wnbrule$ rule (resp.~for a $\wnwrule$ rule). 

	Directed paths of $\egraph{\der}$ range over $\phi, \phi', \phi'', \ldots$.  Maximal directed graphs of $\egraph{\der}$  will be called  \defin{exponential flows}.  We say that a directed path $\phi$ \defin{crosses} an exponential $\cutr$ rule  when it crosses  (both of its) active formulas. The \defin{rank of a directed path} $\phi$, written  $\rank{}\phi$, is the number of $\brule$-nodes crossed by $\phi$.  Finally, the \defin{rank of $\der$}, written  $\rank{}{\der}$, is the  number $\brule$-nodes of $\egraph{\der}$.
\end{definition}

\begin{figure}[t]
	\[
	\small
	\vlderivation{
		\vliin{\cutr}{}{ \tikzmarknode[circle, draw,  very thick,minimum width=0.8cm]{d1}{\wn \cneg{ A}}\ \ ,\ \  \tikzmarknode[circle, draw,  very thick,minimum width=0.8cm]{d2}{\oc A}}{
			\vliin{\cprule}{}{ \tikzmarknode[circle, draw, very thick, minimum width=0.8cm]{a1}{\wn \cneg{A}}\ \ , \ \ \tikzmarknode[circle, draw,  very thick,minimum width=0.8cm]{b1}{\oc A}}{\vlin{\axr}{}{  \cneg{A},  A}{\vlhy{}}}{
				\vliin{\cprule}{}{ \wn \cneg{A},  \oc A}{\vlin{\axr}{}{  \cneg{A},  A}{\vlhy{}}}{\vlin{\cprule}{}{\wn \cneg{A}, \oc A}{\vlhy{\vdots}}}
			}
		}
		{
			\vlin{\wnbrule}{}{\tikzmarknode[circle, draw,  very thick,minimum width=0.8cm]{c1}{\wn \cneg{ A}}\ \ ,\ \  \tikzmarknode[circle, draw,  very thick,minimum width=0.8cm]{c2}{\oc A}}
			{\vlin{\wnwrule}{}{\cneg{A} ,\ \ \tikzmarknode[circle, draw,  very thick,minimum width=0.8cm]{c3}{\wn \cneg{ A}}\ \ ,\ \  \tikzmarknode[circle, draw,  very thick,minimum width=0.8cm]{c4}{\oc A}}{\vlin{\axr}{}{\tikzmarknode[circle, draw,  very thick,minimum width=0.8cm]{c5}{\wn \cneg{ A}}\ \ ,\ \  \tikzmarknode[circle, draw,  very thick,minimum width=0.8cm]{c6}{\oc A}}{\vlhy{}}}}
		}
	}
	\begin{tikzpicture}[overlay, remember picture]
		\draw[->, >=latex', very thick] (a1) to node [ above]{\scriptsize{$b$}} (b1); 
		\draw[->, >=latex', very thick] (c5) to node [ above]{\scriptsize{$g$}}(c6); 
		\draw[->, >=latex', very thick] (b1) to node [ above]{\scriptsize{$c$}}  (c1); 
		\draw[->, >=latex', very thick] (d1) to[bend left=20] node [ below]{\scriptsize{$a$}}(a1); 
		\draw[->, >=latex', very thick] (c2) to[bend left=20] node [ below]{\scriptsize{$l$}} (d2); 
		\draw[->, , >=latex', very thick] (c1) to node [ left, pos=0.9]{\scriptsize{$e$}} (c3); 
		\draw[->, , >=latex', very thick] (c3) to node [ left, pos=0.2]{\scriptsize{$f$}} (c5); 
		\draw[->, , >=latex', very thick] (c6) to  node [ right, pos=0.1]{\scriptsize{$h$}}(c4); 
		\draw[->, , >=latex', very thick] (c4) to node [ right, pos=0.95]{\scriptsize{$i$}} (c2); 
	\end{tikzpicture}
	\]
	\caption{A coderivation $\der$ and its exponential graph $\egraph{\der}$.}
		\label{fig:example-egraph}
\end{figure}

\begin{example}\label{rem:example-exponential-graph}
	Let $\der$ be the coderivation in~\Cref{fig:example-egraph}. Notice that the sub-coderivation $\der_{1}$ (i.e., the one with conclusion the right premise of the $\cutr$ rule) is a steady $\nwprule$ $\nwpromotion$ whose calls are axioms. There is only one exponential flow in  $\egraph{\der}$, which is   $\phi\dfn abcdefghil$. We have $\rank{}{\phi}=\rank{}{\der}=1$. Finally,  $\der$ contains only one $\cutr$ which is shallow and steady.
\end{example}

\newcommand{\rescut}{\hat{\rrule}}

As we mentioned before,  exponential flows can be seen as  static representations of cut elimination steps applied to steady cuts. We now study their  invariance properties under rewriting. First, we need a notion of residue of cut rules and exponential flows along  a  cut elimination step.

\begin{definition}[Residues]
	Let  $\der \in \nupll$ and  $\rrule$ be a steady cut. The \defin{residue of $\rrule$ (along $\der \cutelim\der'$)} is the cut rule $\rescut$ of $\der'$ defined as follows:
	\begin{itemize}
		\item if $\der \cutelim\der'$  does not reduce $\rrule$ then $\rescut\dfn \rrule$ 
		\item if   $\der \cutelim\der'$    reduces $\rrule$  and $\rrule\neq \cutstep{\cprule}{\wnwrule}$ then $\rescut$ is the \emph{unique   steady} cut that is obtained by reducing $\rrule$ (see~\Cref{fig:cut-elim-pll}).
		\item otherwise, $\rescut$ is undefined. 
	\end{itemize}
	
	We denote by  $\phi_\rrule$ an exponential flow that crosses $\rrule$. We call the \defin{residue of $\phi_\rrule$ (along $\der \cutelim\der'$)} the \emph{unique} exponential flow of $\egraph{\der'}$ defined by  $\widehat{\phi_\rrule}\dfn \phi_{\widehat{\rrule}}$. 
\end{definition}

\begin{proposition}[Invariance]\label{prop:invariance}
	Let  $\der \in \nupll$ with $\oc$-free conclusion where all shallow cuts are steady. Then:
	\begin{enumerate}
		\item \label{enum:invariance1} Every exponential flow $\phi$ of $\egraph{\der}$ ends at a $\wrule$-node.
		\item \label{enum:invariance2} If  $\der \cutelim \der'$ reduces a steady cut $\rrule^*$ then, for any steady cut $\rrule$ in $\der$,  $\rank{}{\widehat{\phi_{\rrule}}}\leq \rank{}{\phi_{\rrule}}$ holds whenever  $\rrule^*\neq \rrule$ or $\rrule^*\neq \cutstep{\cprule}{\wnwrule}$.
	\end{enumerate}
\end{proposition}
\begin{proof}
	Let $\phi$ be an exponential flow of $\egraph{\der}$. We show that:
	\begin{enumerate}[(a)]
\item   every $\wn$-node that has no directed edge to another $\wn$-node is either a $\wrule$-node at the conclusion of a $\cprule$-rule.
\item  every $\oc$-node with no directed edge to another $\oc$-node is the active formula of a $\cutr$ rule.
	\end{enumerate}
The Point (a) follows by inspecting the definition of exponential flow. Concerning Point (b), let $v$ be a $\oc$-node $v$ with no directed edge to another $\oc$-node. Since all cuts are steady and the conclusion of $\der$ is $\oc$-free, by~\Cref{prop:downward-oc}, $v$ must be the $\oc$-active formula of a $\cutr$-rule.

 We now show that~\Cref{enum:invariance1} follows from the two points above. Indeed,  by maximality of exponential flows, if  $\phi$ crosses a  $\wn$-principal formula in the conclusion of a $\cprule$-rule (resp. the $\wn$-active formula of a $\cutr$ rule), then it also crosses its $\oc$-principal conclusion (resp. its $\oc$-active formula).  This means that, since exponential flows are finite,  $\phi$ must end at a $\wrule$-node.

Let us now prove~\Cref{enum:invariance2}.  Let $\expflow$  be an exponential flow $\expflow$ of $\egraph{\der}$. If $\rrule^*$ is not crossed by $\phi_\rrule$ then $\rank{}{\widehat{\phi_{\rrule}}}= \rank{}{\phi_{\rrule}}$. Otherwise, $\rrule^*$ is  crossed by $\phi_\rrule$ , and we proceed by case analysis. The cases where $\rrule$ is a commuting cut or a $\cutstep{\cprule}{\cprule}$ cut are straightforward. Suppose now $\rrule$ is a  cut $\cutstep{\cprule}{\wnwrule}$.  Since by assumption $\rrule^*\neq\rrule$, $\rescut$ is defined, and so is $\widehat{\phi_{\rrule}}$. By inspecting~\Cref{fig:exponential-flow-exponential-cut-elimination} (top) we conclude that $\rank{}{\widehat{\phi_{\rrule}}}\leq \rank{}{\phi_{\rrule}}$.
  Finally, we consider the case where $\rrule^*=\cutstep{\cprule}{\wnbrule}$ as in~\Cref{fig:exponential-flow-exponential-cut-elimination} (bottom). Then, $\phi_\rrule=\phi' x_iz_iab \phi''$ for some $i$ and for some directed paths $\phi'$ and $\phi''$. We notice that the existence of the directed edge labelled with $c$ in the figure is inferred from the fact that $\rrule^*$ is steady by hypothesis. This means that $\widehat{\phi_{\rrule}}=\phi' u_iv_icde\phi''$. But then $\rank{}{\widehat{\phi_{\rrule}}}=\rank{}{\phi'}+ \rank{}{\phi''}+1\leq \rank{}{\phi_{\rrule}}$.
\begin{figure*}
\[
\footnotesize
\vlderivation{
	\vliin{\cutr}{}{\tikzmarknode[rectangle, double, draw, very thick, minimum width=0.6cm, minimum  height=0.6cm]{g1}{\wn \Gamma}\qquad   \Delta}{
		\vliiin{\cprule}{}{\tikzmarknode[circle, double, draw, very thick, minimum width=0.6cm]{g2}{ \wn \Gamma}\qquad    \tikzmarknode[circle, , draw, very thick, minimum width=0.6cm]{a1}{\oc A}}{
			\vldr{\der_1}{\Gamma, A}
		}
		{\vlhy{\ \ }}
		{
			\vldr{\der_2}{ \wn \Gamma, \oc A}
		}
	}{
		\vlin{\wnwrule}{}{ \tikzmarknode[circle, , draw, very thick, minimum width=0.6cm]{na1}{\, \wn \cneg A}\ \  \Delta}{\vldr{\der_3}{\Delta}}
	}
}
\quad 
\cutelim
\quad 
\small
\vlderivation{
	\vliq{\size{\Gamma}\times\wnwrule}{}{ \tikzmarknode[rectangle, double , draw, very thick, minimum width=0.6cm, minimum  height=0.6cm]{g}{\wn\Gamma} \ \  \Delta}{	\vldr{\der_2}{\Delta}}
}
\vspace{0.2cm}
\begin{tikzpicture}[overlay, remember picture]
	\draw[-latex, very thick, double] ([shift=({0,0})]g1.160) to[bend left=0] node[below left, pos=0.4, xshift=1.5]{\scriptsize{$\vec{x}$}} ([shift=({0,0})]g2.300);
	\draw[-latex, very thick, double] ([shift=({0,0})]g2.east) to[bend left=0] node[above, pos=0.4, yshift=-0.5 ]{\scriptsize{$\vec{z}$}}   ([shift=({0,0})]a1.west);
	\draw[-latex, very thick, ] ([shift=({0,0})]a1.east) to[bend left=0] node[below, pos=0.5, yshift=1]{\scriptsize{$a$}}   ([shift=({0,0})]na1.west);
\end{tikzpicture}
\]
\medskip
\[
\footnotesize
\vlderivation{
	\vliin{\cutr}{}{\tikzmarknode[rectangle, double , draw, very thick, minimum width=0.6cm, minimum  height=0.6cm]{g1}{\wn \Gamma}\qquad  \Delta}{
		\vliiin{\cprule}{}{\tikzmarknode[rectangle, double , draw, very thick, minimum width=0.6cm, minimum height=0.6cm]{g2}{\wn \Gamma}\qquad     \tikzmarknode[circle,  , draw, very thick, minimum width=0.6cm]{a1}{\oc A}}
		{
			\vldr{\der_1}{\Gamma, A   }
		}
		{\vlhy{  }}
		{
			\vldr{\der_2}{\wn \Gamma,\oc A}
		}
	}{
		\vlin{\wnbrule}{}{\tikzmarknode[rectangle , draw, very thick, minimum width=0.6cm, minimum height=0.6cm]{na1}{\,  \wn \cneg A}\ \  \Delta\quad }{
			\vldr{\der_3}{\ \ \cneg{A}\ \   \tikzmarknode[rectangle,  , draw, very thick, minimum width=0.6cm, minimum  height=0.6cm]{na2}{\, \wn\cneg A} \ \  \Delta}
		}
	}
}
\quad 
\cutelim
\quad 
\small
\vlderivation{
	\vlin{}{}{\quad \tikzmarknode[rectangle, double , draw, very thick, minimum width=0.6cm, minimum height=0.6cm]{wg1}{\wn\Gamma}\quad  \Delta}{
		\vlin{}{}{\qquad \quad \vdots\ \scalebox{0.7}{$\size{\Gamma}\times\wnbrule$}}{
			\vliin{\cutr}{}{ \tikzmarknode[rectangle, double , draw, very thick, minimum width=0.6cm, minimum height=0.6cm]{wg2}{\wn \Gamma}\quad \Gamma \quad \Delta }{
				\vldr{\der_2}{\tikzmarknode[rectangle, double , draw, very thick, minimum width=0.6cm, minimum  height=0.6cm]{wg3}{\wn\Gamma}\quad \    \tikzmarknode[rectangle,  , draw, very thick, minimum width=0.6cm, minimum  height=0.6cm]{oa1}{\oc A}\ \ }
			}{
				\vliiin{\cutr}{}{ \tikzmarknode[rectangle , draw, very thick, minimum width=0.6cm, minimum height=0.6cm]{wna1}{\, \wn \cneg A}\quad \Gamma\quad   \Delta\quad }{
					\vldr{\der_1}{  \Gamma, A }
				}{\vlhy{}}{
					\vldr{\der_3}{  \cneg{A}\quad \tikzmarknode[rectangle,  , draw, very thick, minimum width=0.6cm, minimum  height=0.6cm]{wna2}{\,  \wn \cneg A} \quad  \Delta }
				}
			}
		}
	}
}
\vspace{0.2cm}
\begin{tikzpicture}[overlay, remember picture]
	\draw[-latex, very thick, double] ([shift=({0,0})]g1.160) to[bend left=0] node[below left, pos=0.4, xshift=1.5]{\scriptsize{$\vec{x}$}} ([shift=({0,0})]g2.300);
	\draw[-latex, very thick, double] ([shift=({0,0})]g2.east) to[bend left=0] node[above, pos=0.4, yshift=-0.5 ]{\scriptsize{$\vec{z}$}}   ([shift=({0,0})]a1.west);
	\draw[-latex, very thick, ] ([shift=({0,0})]a1.east) to[bend left=0] node[below, pos=0.5, yshift=1]{\scriptsize{$a$}}   ([shift=({0,0})]na1.west);
	\draw[-latex, very thick] ([shift=({0,0})]na1.40) to[bend left=0] node[below right, pos=0.85, xshift=-4]{\scriptsize{$b$}} ([shift=({0,0})]na2.south);
	\draw[-latex, very thick, ] ([shift=({0,0})]oa1.east) to[bend left=0] node[below, pos=0.5, yshift=1]{\scriptsize{$d$}}  ([shift=({0,0})]wna1.west);
		\draw[-latex, thick, dashed] ([shift=({0,0})]wg3.east) to[bend left=0] node[below, pos=0.5, yshift=1]{\scriptsize{$c$}}  ([shift=({0,0})]oa1.west);
	\draw[-latex, very thick, double] ([shift=({0,0})]wg1.100) to[bend left=0] node[left, pos=0.4]{\scriptsize{$\vec{u}\, $}}  ([shift=({0,0})]wg2.280);
	\draw[-latex, very thick, ] ([shift=({0,-0.1})]wna1.30) to[bend left=0] node[below right, pos=0.6, xshift=-3, yshift=2]{\scriptsize{$e$}}   ([shift=({0,0})]wna2.220);
	\draw[-latex, very thick, double] ([shift=({0,0})]wg2.160) to[bend left=0] node[below left, pos=0.4]{\scriptsize{$\vec{v} $}}  ([shift=({0,0})]wg3.300);
\end{tikzpicture}
\]
\caption{From top, a cut elimination step $\der \cutelim \der'$  reducing $\cutstep{\cprule}{\wnwrule}$ and $\cutstep{\cprule}{\wnbrule}$, and the corresponding exponential graphs (only the relevant nodes and edges  are displayed).   Double circles (resp., double edges) represent multiple nodes (resp., multiple edges), while squared nodes are nodes shared by  $\egraph{\der}$ and $\egraph{\der'}$.				Edges are labelled by letters, and vectors $\vec{x}=x_1, \ldots, x_n$ represent a list of labels, one for each edge. Finally, the dashed edge labelled by $c$ exists if $\der_2$ is a $\nwbox$, i.e., if the cut is steady.} 
\label{fig:exponential-flow-exponential-cut-elimination}
\end{figure*}
\end{proof}

\newcommand{\ebcut}[1]{\mathsf{bcut}_{\rrule}({#1})}
\newcommand{\ephi}[1]{\phi_{#1}}
\newcommand{\errule}{\rrule}
\newcommand{\bcut}[1]{\mathsf{bcut}({#1})}
\newcommand{\boxes}[1]{\mathsf{box}(#1)}
\newcommand{\bpath}[1]{\wnbrule({#1})}

\subsection{Termination theorem} 

Termination of the shallow cut elimination strategy is an immediate consequence of the Key Lemma (\Cref{lem:termination-key-lemma}).  To this end, we introduce a partial ordering on steady cuts, denoted $ \preceq_\der $, which  relates two  cuts $\rrule$ and $\rrule$ connected by an exponential flow $\phi$ in a coderivation $\der$. Intuitively,  the \emph{distance} of a cut $\rrule$ to the end node of $\phi$, called \emph{weight of $\rrule$}, is a measure of  the number of cut elimination steps required to fully reduce $\rrule$. Reducing a cut $\rrule$ strictly decreases such a distance for $\rrule$, while it might increase it for smaller cuts  $\rrule' \preceq_\der \rrule $. Termination for shallow cut elimination can be then proven by induction on a lexicographic ordering defined on $\preceq_\der$. The Key Lemma will also provide an estimation of  the number of calls of a shallow $\nwbox$ that become shallow after each round, which depends on  the rank of exponential flows and the fact that the latter cannot increase during cut elimination.

\newcommand{\wt}[1]{\mathsf{wt}(#1)}

\begin{definition}[Partial ordering and weight]
	Let $\der\in \nupll$. We define a partial ordering on steady cuts $\preceq_\der$ defined by $\rrule \preceq_\der \rrule'$ if only only if there is a directed path from the active formulas of $\rrule$ to those of $\rrule'$.  We set $\prec_\der$ as the strict version of $\preceq_\der$.
	
	Let $\rrule$ be a steady cut of $\der$. The  \defin{weight of  $\rrule$ (in $\der$)}, written $\wt{\rrule}$,   is the number of nodes  in $\phi_\rrule$ from the $\wn$-active formula of $\rrule$ to its end node.  
\end{definition} 

	The following properties are straightforward from the above definition, the definition of residue, and~\Cref{prop:invariance}.\ref{enum:invariance2}: 
	
	\begin{proposition}\label{prop:properties-preceq-weight}
	Let $\der\in \nupll$, and let $\rrule, \rrule'$ be steady cuts. Then:
	\begin{enumerate}
\item  \label{enum:properties-preceq-weight1} $\rrule \prec_\der \rrule'$  implies $\phi_\rrule= \phi_{\rrule'}$ and $\wt{\rrule}>\wt{\rrule'}$. 
\item \label{enum:properties-preceq-weight2} If  $\der \cutelim \der'$ is a cut elimination step reducing a steady cut $\rrule^* \not \in \{\rrule, \rrule'\}$ such that $\rrule^* \neq \cutstep{\cprule}{\wnwrule}$ then $\rescut$ and $\widehat{\rrule'}$ exist and $\rescut \preceq_\der \widehat{\rrule'}$.
	\end{enumerate}
	\end{proposition}

We can now prove the key lemma for the termination of shallow cut elimination strategies.

\begin{lemma}[Key Lemma]\label{lem:termination-key-lemma}
	Let $\der\in \nupll$ with $\oc$-free conclusion. Then, the shallow cut elimination strategy applied to $\der$ satisfies the following properties:
	\begin{enumerate}
		\item \label{enum:round1}     \textbf{Phase 1} terminates (i.e., $\mphase{{\derround {d-1}}}{{\derphase d}}$). In particular,  ${\base{\derround {d-1}}} \cutelims{\base{\derphase d}}$ for every $1 \leq d \leq \depth{\der}+1$.
		\item   \label{enum:round3}  \textbf{Phase 2}  terminates (i.e., $\ephase{\derphase d}{\derround d}$). In particular, ${\trunk{\derphase d}{\rank{}{\derphase d}}}\cutelims{\base{\derround d}}$   for every $1 \leq d \leq \depth{\der}+1$.
		\item  \label{enum:round4} If $\depth{\derround {d-1}}>0$ then 
		$\depth{\derround d}= \depth{\derround {d-1}}-1$.
	\end{enumerate}
\end{lemma}
\begin{proof}
	\Cref{enum:round1} follows from	the fact that shallow cuts that are not steady  only affect $\base{\derround {d-1}}$, so that 
${\base{\derround {d-1}}}\cutelims{\base{\derphase d}}$  by~\Cref{thm:cut-elimApp}. 	
	
		Let us prove~\Cref{enum:round3}. First, we show that \textbf{Phase 2} terminates. Notice that after \textbf{Phase 1} all  shallow cuts  are steady, so $\derphase d$ contains only steady cuts. Let $\der_n$ be the $n$-th step of cut elimination of this phase (so $\der_0=\derphase d$). 	We prove the statement by induction on the lexicographical ordering $L(\der_n)\dfn(n^s, \wt{\rrule^n_1}, \ldots, \wt{\rrule^n_{n^s}})$ where:
		\begin{itemize}
			\item  $n^s$ is the number of steady cuts of $\der_n$ (except the new ones produced during this phase). 
			\item  	$\rrule^n_1\prec_{\der_n}^l\ldots   \prec^l_{\der_n}\rrule^n_{n^s} $ is a (strict) linear ordering that extends the strict ordering $\prec_{\der_n}$
		\end{itemize}
	Now, suppose that $\der_n \cutelim \der_{n+1}$ reduces a steady cut $\rrule^n_{i_0}$ for some $i_0$. We have two cases:
		\begin{itemize}
			\item 	If $\rrule^n_{i_0}= \cutstep{\cprule}{\wnwrule}$ then  ${(n+1)}^s<n^s$. This means that $L(\der_{n-1})<L(\der_n)$.
			\item Otherwise, all steady cuts have a steady residue. Since the  linear ordering is preserved by~\Cref{prop:properties-preceq-weight}.\ref{enum:properties-preceq-weight2}, we have $\widehat{\rrule^n_i}=\rrule_i^{n+1}$, and so 	${(n+1)}^s=n^s$.  Moreover, since $\wt{\rrule^{n}_i}>\wt{\rrule^n_{i+1}}$  by~\Cref{prop:properties-preceq-weight}.\ref{enum:properties-preceq-weight1}, we also have   $\wt{\rrule^{n+1}_i}>\wt{\rrule^{n+1}_{i+1}}$ for all $i$. 	By inspecting the  commuting and exponential cut elimination rules (see~\Cref{fig:cut-elim-finitary,fig:cut-elim-pll}), we notice that $\der_n \cutelim \der_{n+1}$   decreases the weight of the residue  of $\rrule^n_{i_0}$, i.e., $\wt{\rrule^{n+1}_{i_0}}=\wt{\widehat{\rrule^{n}_{i_0}}}<\wt{\rrule^n_{i_0}}$. Moreover, that reduction step can only increase the weight of the residue of steady cuts $\rrule^n_i \preceq_{\der_n}\rrule^n_{i_0}$. So, for any  $\rrule^n_i$ with $i \neq i_0$:
			\begin{itemize}
\item  if  $\rrule^n_i$ and $\rrule^{n}_{i_0}$ are incomparable w.r.t.  $\preceq_{\der_n}$, we have 
 $\wt{\rrule^{n+1}_{i}}=\wt{\widehat{\rrule^{n}_{i}}}=\wt{\rrule^n_{i}}$
 \item If   $\rrule^n_{i_0}\preceq_{\der_n}\rrule^n_{i}$,  we have again 
 $\wt{\rrule^{n+1}_{i}}=\wt{\widehat{\rrule^{n}_{i}}}=\wt{\rrule^n_{i}}$
 \item otherwise $\rrule^n_i\preceq_{\der_n}\rrule^n_{i_0}$, in which case $\wt{\rrule^{n+1}_i}= \wt{\widehat{\rrule^n_i}}\geq \wt{\rrule^n_i}$.
			\end{itemize}
			 This shows that $L(\der_{n-1})<L(\der_n)$.
		\end{itemize}
		Therefore, \textbf{Phase 2} terminates. We now show that ${\trunk{\derphase d}{\rank{}{\derphase d}}}\cutelims{\base{\derround d}}$,  for every $1 \leq d \leq \depth{\der}+1$. For every $n$, let  $\rrule^n_{i_0}$ be the steady cut reduced by $\der_n \cutelim \der_{n+1}$, and let $\phi_{\rrule^n_{i_0}}$ be the exponential flow of $\egraph{\der_n}$ that crosses $\rrule^n_{i_0}$. If $\rrule^n_{i_0}= \cutstep{\cprule}{\wnwrule}$ then $\rrule^n_{i_0}$ has no residue. Otherwise, by~\Cref{prop:invariance}.\ref{enum:invariance2} we have $\rank{}{\widehat{\phi_{\rrule^n_{i_0}}}}\leq \rank{}{\phi_{\rrule^n_{i_0}}}$. Moreover, for every other $\rrule^n_{i}$ such that $\phi_{\rrule^n_{i}}\neq \phi_{\rrule^n_{i_0}} $, we have $\rank{}{\widehat{\phi_{\rrule^n_{i_0}}}}= \rank{}{\phi_{\rrule^n_{i_0}}}$. Therefore, only the first $\rank{}{\derphase{d}}$ calls of every shallow $\nwbox$ in $\derphase{d}$ can become shallow at the end of the round. Since only shallow rules are affected by \textbf{Phase 2}, this implies ${\trunk{\derphase d}{\rank{}{\derphase d}}}\cutelims{\base{\derround d}}$.
		
	Let us finally show~\Cref{enum:round4}. Suppose  $\depth{\derround {d-1}}>0$. \textbf{Phase 1} does not affect the depth of $\derround{d-1}$, so $\depth{\derround {d-1}}=\depth{\derphase{d}}$. Let us consider \textbf{Phase 2}. Let $\nwpromotion_1, \ldots, \nwpromotion_k$ be the shallow $\nwboxes$ of $\derphase{d}$. Since we assumed $\depth{\derphase{d}}=\depth{\derround {d-1}}>0$, we have $k>0$. By hypothesis there are only steady cuts, and so all such $\nwboxes$ are steady by~\Cref{prop:advanced-properties-exp-flow}.\ref{enum:steady1}.  Since every steady cut is crossed by an exponential flow, and by~\Cref{prop:invariance}.\ref{enum:invariance1} every exponential flow ends at a $\wrule$-node, the weight of every steady cut during this phase will eventually decrease to $0$. This means that every steady $\nwbox$  will eventually be erased by reducing a $\cutstep{\cprule}{\wnwrule}$ cut. Therefore,  after \textbf{Phase 2} the depth decreases by $1$, i.e., $\depth{\derround d}=\depth{\derphase{d}}-1= \depth{\derround {d-1}}-1$.
\end{proof}

\begin{restatable}[Termination]{theorem}{TER}\label{prop:termination}
	Let $\der\in \nupll$ with $\oc$-free conclusion. The shallow cut elimination strategy applied to $\der$ terminates in a \emph{finite} number of steps into a cut-free~\emph{derivation}. 
\end{restatable}
\begin{proof}
	Termination is a consequence of~\Cref{lem:termination-key-lemma}, as by~\Cref{lem:finite-depth} the depth of $\der$ is finite and reduces at every round. In particular, by~\Cref{lem:termination-key-lemma}.\ref{enum:round4} we have
	 $\depth{\derround {\depth{\der}}}=0$ and so $\derround {\depth{\der}+1}$
	  is  cut-free and $\nwbox$-free.  Hence, by~\Cref{rem:finite-derivation},   $\derround {\depth{\der}+1}= \base{\derround {\depth{\der}+1}}$ is finite, i.e.,  $\derround {\depth{\der}+1}$ is a derivation.
\end{proof}

\subsection{Polynomial modulus of continuity}\label{subsec:poly-mod-cont}

In this final subsection, we analyse the complexity of the shallow cut elimination strategy, leveraging the estimations provided by the Key Lemma (\Cref{lem:termination-key-lemma}). From the polynomial modulus of continuity result (\Cref{prop:polynomial-moduli}), we will infer soundness for our non-wellfounded proof systems (\Cref{thm:soundness}).

We start with introducing a notion of (finite) size for coderivations in $\nupll$, called \emph{cosize},  relying on~\Cref{lem:finite-depth}.

\begin{definition}[Cosize]
	Let $\der\in \nupll$, and  $\bord{\der}=\{v_1, \ldots, v_k\}$ be its decomposition prebar (thus $\nwpromotion_i \coloneqq \der_{v_i}$ is a $\nwprule$ for all $ 1 \leq i \leq k$).  We define the \defin{cosize of $\der$}, written $\isize{\der}$,  by induction on    $\depth{\der}$. 
	If $\depth \der=0$ then $\der=\base{\der}$  and we set $\isize \der \dfn \size \der$. Otherwise $\depth \der >0$, and  
	$
	\isize{\der} \dfn \size{\base\der }+ \sum_{i= 1}^k \sum_{\der' \in \supportof{\nwpromotion_i}
		}\isize{\der'} 
		$.
		
		The \defin{cosize at depth $d$}, written $\isize{\der}_d$, is defined for all $d \leq \depth \der$
		as
			$\isize{\der}_0 =  \size{\base{\der}}$,   
				and
				as $\isize{\der}_{d+1} = \max\set{\isize{\der'}_d\mid \der'\in
					\supportof{\nwpromotion_i} \mbox{ for some } 1 \leq i\leq k} 
				$.
		
	\end{definition}

Notice that, by~\Cref{prop:weak-regular-finite-support},   
{$\supportof{\nwpromotion_i}$}
is a \emph{finite} set for any $i\in\intset1k$, and so $\isize{\der}_d \leq \isize{\der}\in \mathbb{N}$. 

	\newcommand{\form}[1]{\mathsf{form}(#1)}
\newcommand{\expfl}[1]{\mathsf{expfl}(#1)}

The following relation between the size of the $n$-truncation of $\der$, i.e.,  $\prun \der n$ (\Cref{defn:trunc}), and the cosize of $\der$ holds.

\begin{proposition}\label{prop:bound-pruning} 	
	Let $\der\in \nupll$. 
	Then 
	$$\size{\prun \der n}\in \mathcal{O}(n^{\depth{\der}+1}\cdot \isize{\der}^{\depth{\der}+1})
	\qquad.
	$$
\end{proposition}
\begin{proof}
	Let $\nwpromotion_i= \der_{v_i}$ with $v_i\in\bord \der$.
	If $\depth{\der}=0$, then $\size{\prun \der n}= \size{\der}= \isize \der$. 
	If $\depth{\der}=d+1$, then $\size{\prun{\der}n} \dfn \size{\base{\der} }+ \sum_{i= 1}^k \sum_{j=0}^n \size{\prun{\nwpromotion_i(j)}n}$ by definition.  
	By the induction hypothesis  $\size{\prun {\nwpromotion_i(j)} n}\in \mathcal{O}(n^{d}\cdot \isize{\nwpromotion_i(j)}^d)$, hence $\size{\prun {\nwpromotion_i(j)} n}\in \mathcal{O}(n^{d}\cdot \isize{\der}^d)$. 
	Since $k\leq \isize{\der}$, then we have $\size{\prun \der n}\in \mathcal{O}(\isize{\der}+ n \cdot \isize{\der}\cdot n^d\cdot \isize{\der}^d)$, 
	and we conclude that 
	$\size{\prun \der n}\in \mathcal{O}(n^{d+1}\cdot \isize{\der}^{d+1})$.
\end{proof}

We now show  that shallow cut elimination requires a number of steps that can be polynomially bounded w.r.t.~the {cosize} of the starting coderivation. First, we need a preliminary lemma.

\begin{lemma}\label{lem:pre-computing-modulus}
	Let $\der \in \nupll$ be a coderivation of a  $\oc$-free sequent.  Then, 
	$
	\isize{\derround d}_0 \in \mathcal{O}\left( \prod_{i=0}^{d} \left(\isize{\derround 0}_i \right)^{ 2^{d+1-i}}\right)
	$
	for all  $0 \leq d \leq \depth{\der}+1$.

	\end{lemma}
\begin{proof}
	 First, we notice  that, since $\der$ is weakly regular and progressing, there is a bound $s^*\geq 0$ on the maximum number of $\wn$-formulas in the conclusion of a $\cprule$ rule of $\der$, i.e.,  $\mathcal{S}(\der)$ (see~\Cref{thm:cut-elimApp}). Moreover, by~\Cref{prop:cut-elim-preserves-finexp-reg-weakreg}, we can assume that $s^* \geq 0$ bounds $\mathcal{S}(\derphase d)$ and $\mathcal{S}(\derround d)$. Hence $\mathcal{S}(\derphase d)$ and $\mathcal{S}(\derround d)$ will be considered as \emph{constants} throughout this proof.
	 
 We prove the statement by induction on $0 \leq d \leq \depth{\der}+1$. The case $d=0$ is trivial. 	If $d>0$ then,   by~\Cref{lem:termination-key-lemma}.\ref{enum:round3},   we have ${\trunk{\derphase {d}}{\rank{}{\derphase {d}}}}\cutelims{\base{\derround {d+1}}}$. Then by~\Cref{thm:cut-elimApp}.\ref{enum:cubic2}:
\begin{equation}\label{eqn:a}
	\isize{\derround{d}}_0 =\size{\base{\derround{d}}}
	\in
	\mathcal{O}( \mathcal{S}(\trunk{\derphase {d}}{\rank{}{\derphase {d}}}) \cdot  \size{\trunk{\derphase {d}}{\rank{}{\derphase {d}}}})
	=
	\mathcal{O}(  \size{\trunk{\derphase {d}}{\rank{}{\derphase {d}}}})
\end{equation}

Moreover, if $\nwpromotion_1, \ldots, \nwpromotion_n$ are the shallow  $\nwprule$s of $\derphase{d}$, since $n, \rank{}{\derphase {d}}\leq \isize{\derphase {d}}_0$ and $\base{\nwpromotion_i(j)}\leq \isize{\derphase{d}}_1$, then
	\begin{equation}\label{eqn:b}
		\begin{array}{rcl}
	\size{\trunk{\derphase {d}}{\rank{}{\derphase {d}}}}&=&\size{\base{\derphase {d}}}+ \sum_{i=0}^{n} \sum_{j=0}^{\rank{}{\derphase{d+1}}} \size{\base{\nwpromotion_i(j)}}
\\
&\in&
\mathcal{O}\left(\isize{\derphase{d}}_0 + \left( \isize{\derphase{d}}_0 \right)^2\cdot  \isize{\derphase{d}}_1\right)
\\
&=&\mathcal{O}\left(\left(\isize{\derphase{d}}_0 \right)^2\cdot  \isize{\derphase{d}}_1\right)
		\end{array}
	\end{equation}

 
On the other hand:
\begin{equation}\label{eqn:c}
	\begin{array}{rcll}
		\isize{\derphase d}_0&=&\size{\base{\derphase d}}\in  \mathcal{O}(\mathcal{S}(\base{\derround {d-1}} \cdot\base{\derround {d-1}}) &\mathrm{Prop.}~\ref{thm:cut-elimApp}.\ref{enum:cubic2}
		\\
		&=&\mathcal{O}(\mathcal{S}(\base{\derround {d-1}} \cdot\isize{\derround {d-1}}_0)\\
		&=&\mathcal{O}(\isize{\derround {d-1}}_0)
	\end{array}
\end{equation}

 Finally, we notice that:
 \begin{equation}\label{eqn:d}
 \isize{\derphase d}_1= \isize{\derround {d-1}}_1= \isize{\derround 0}_{d}
 \end{equation}
  as by~\Cref{prop:termination}.\ref{enum:round3} the rules of $\der^{0}$ with nesting level $d$ are unaffected  in the first $d-1$ rounds of cut elimination, and by~\Cref{prop:termination}.\ref{enum:round4}   each round decreases the depth. 
 Then, by inductive hypothesis we have:
\begin{equation*}
		\begin{array}{rcll}
			\isize{\der^{d}}_0& \in&
			\mathcal{O}\left(\left( \isize{\derphase d}_0 \right)^2\cdot  \isize{\derphase d}_1\right) 
			&\mathrm{Eq.}~\ref{eqn:a}\textrm{-}\ref{eqn:b}
			\\
			&=&
			\mathcal{O}\left( \left(\isize{\derphase d}_0 \right)^2\cdot  \isize{\derround 0}_{d}\right) &\mathrm{Eq.}~\ref{eqn:d}
		\\
	&		=&
			\mathcal{O}\left(\left(\isize{\derround{d-1}}_0\right)^2\cdot \isize{\derround 0}_{d}\right)&\mathrm{Eq.}~\ref{eqn:c}
\\
		&	=&
			\mathcal{O}\left(\left( \prod_{i=0}^{d-1} \left(\isize{\derround 0}_i \right)^{ 2^{d-i}} \right)^2 \cdot \isize{\derround 0}_{d}\right)
			&\textrm{Induction hypothesis}
			\\
		&	= &
			\mathcal{O}\left( \prod_{i=0}^{d-1} \left(\isize{\derround 0}_i \right)^{ 2^{d+1-i}} \cdot \isize{\derround 0}_{d}\right)	
			\\
			&=&
			\mathcal{O}\left( \prod_{i=0}^{d} \left(\isize{\derround 0}_i \right)^{ 2^{d+1-i}}\right)	
			\;.
		\end{array}
	\end{equation*}
	\qedhere
\end{proof}

\begin{restatable}[Polynomial modulus of continuity]{lemma}{MOD}\label{prop:polynomial-moduli}
	Let $\der \in \nupll$ be a coderivation of a  $\oc$-free sequent. Then, for some polynomial $p:\Nset \to \Nset$    depending solely on $\depth{\der}$, $\prun \der {p(\isize{\der})}$ rewrites {by the shallow cut elimination strategy}  to a cut-free $\zero$-free {derivation}.
\end{restatable}	
\begin{proof}
	By~\Cref{lem:pre-computing-modulus} 	we have:
	$$
	\begin{array}{rcl}
		\isize{\der^d}_0\in  &
		\mathcal{O}\left(  \prod_{i=0}^{d} \isize{\derround 0}^{ 2^{d+1-i}}_i\right)	
		& 	\text{\Cref{lem:pre-computing-modulus}}
		\\
		=&	\mathcal{O}\left(  \prod_{i=0}^{\depth{\der}} \isize{\derround 0}^{ 2^{\depth{\der}+1-i}}_i\right)	
		&\text{\Cref{lem:depth}}
		\\=&
		\mathcal{O}\left(  \isize{\derround 0}^{\depth{\der}\cdot 2^{\depth{\der}+1}}\right)	
		\\=&
		\mathcal{O}\left(  \isize{\der}^{\depth{\der}\cdot 2^{\depth{\der}+1}}\right)	
	\end{array}
	$$
	Hence,  since~\Cref{thm:cut-elimApp}.\ref{enum:cubic2} implies  $\rank{}{\derphase d}\leq\size{\base{\derphase d}}\in \mathcal{O}( \size{\base{\derround d}})=\mathcal{O}(\isize{\derround d}_0)$, by~\Cref{lem:termination-key-lemma}.\ref{enum:round1}-\ref{enum:round4} there is some $k>0$ depending solely on $\depth{\der}$ and a constant $c>0$ such that:
	$${\prun{\derround {d-1}}{c\cdot\isize{\der}^{k}}}\cutelims{ {\prun{\derphase d}{c\cdot \isize{\der}^{k}}}\cutelims{  \prun{\derround d}{c\cdot \isize{\der}^{k}} }}$$ 
	for any $0 \leq d \leq\depth{\der}+1$.  
	This means that $\prun \der {c\cdot \isize{\der}^{k}}\cutelims \prun{\derround{\depth{\der}+1}}{c\cdot \isize{\der}^{k}} $. But  
	$\prun{\derround{\depth{\der}+1}}{c\cdot \isize{\der}^{k}} =
	\derround{\depth{\der}+1}$ by~\Cref{prop:termination}. Therefore,  we have that $\prun \der {c\cdot \isize{\der}^{k}}$  rewrites {by the shallow cut elimination strategy}  to a cut-free $\zero$-free {derivation}.\qedhere
\end{proof}

From the polynomial modulus of continuity on cut elimination we obtain our  soundness theorems for $\nupll$ and $\cpll$.
	
\begin{theorem}[Soundness] \label{thm:soundness} Let $f:(\{\false,\true\}^*)^n\to \{\false,\true\}^*$:
	\begin{enumerate}
		\item \label{enum:soundness1}	If $f$ is representable in $\nupll$  then $f \in \fppoly$; 
		\item  \label{enum:soundness2}	If $f$ is representable in $\cpll$  then $f \in \fptime$.
	\end{enumerate}
\end{theorem}
\begin{proof}
	We  only show the case where $f$ is unary for the sake of simplicity. 	Let $\cod{f}\in \nupll$  represent $f$, and let us consider the following coderivation, with $s=b_1, \ldots, b_n \in \{\false,\true\}^*$:
	\[
	\small
	\begin{array}{rcl}
		\der_{f(s)} &\dfn & \vlderivation{
			\vliin{\er \limp}{}{\String}{\vldr{\cod{f}}{ {\String[]}\limp \String}}{\vldr{\cod{s}}{\String[]}}
		}
	\end{array}
	\]
	By~\Cref{prop:polynomial-moduli}   there are $\der_0, \der_1,\ldots, \der_m$ such that: 
	\[
	\prun {\der_{f(s)}} {c \cdot \isize{\der_{f(s)}}^k}=\der_0 \cutelim\der_1 \cutelim \ldots \cutelim \der_m=\cod{f(s)}
	\]
	for some constant $c>0$, and  for some $k>0$ depending solely on $\depth{\der_{f(s)}}= \depth{\cod{f}}$ (since $\depth{\cod{s}}=0$). In particular, $\isize{\der_{f(s)}}\in \mathcal{O}(\isize{\cod{s}})= \mathcal{O}(\size{\cod{s}})=\mathcal{O}(\size{s})$, where $\size{s}$ is the size of the string $s$. So, we have:
	\[
	\prun {\der_{f(s)}} {c \cdot \size{s}^k}=\der_0 \cutelim\der_1 \cutelim \ldots \cutelim \der_m=\cod{f(s)}
	\]
	for some constant $c >0$ and some $k>0$ depending solely on $\depth{\cod{f}}$. Moreover:
	\begin{itemize}
		\item 	By~\Cref{prop:bound-pruning} we have 
		\[
		\def\arraystretch{1.2}
		\arraycolsep=2pt
		\begin{array}{rcl}
			\size{\prun{\der_{f(s)}}{c \cdot \size{s}^k}}&\in& \mathcal{O}(\size{s}^{k\cdot  \depth{\der_{f(s)}}+1}\cdot \isize{\der_{f(s)}}^{\depth{\der_{f(s)}}+1})\\
			&=&
			\mathcal{O}(\size{s}^{k \cdot \depth{\der_{f}}+1}\cdot \size{s}^{\depth{\der_{f}}+1})\\
			&=& 	\mathcal{O}(\size{s}^{k \cdot \depth{\cod{f}}+1}\cdot \size{s}^{\depth{\cod{f}}+1})
			= \mathcal{O}(\size{s}^{h})
		\end{array}
		\]
		for some $h >0$ depending solely on $\depth{\cod{f}}$.
		\item By~\Cref{thm:cut-elimApp}, we have     $m\in \mathcal{O}(|s|^{3h})$ and $\size{\der_i}\in \mathcal{O}(|s|^{h})$. 
	\end{itemize}
	%
	This means that we can construct a polysize family of circuits $\mathcal{C}=(C_n)_{n \geq 0}$ such that, for any $n \geq 0$,  on input $s=b_1, \ldots, b_n \in \{\false,\true\}^*$, $C_n(s)$ evaluates $\der_{f(s)}$ to $\cod{f(s)}$ and returns $f(s)$.  Therefore, $f \in  \fppoly$. Suppose now that $f$ is representable in $\cpll$. Then $\cod{f}$ is regular, and so the function $n \mapsto C_n$ can be constructed uniformly by a  Turing machine. Moreover, it is easy to see that this Turing machine works in polynomial time (actually even in logarithmic space). Therefore, $f \in \fptime$. 
\end{proof}

\section{Completeness}\label{subsec:completeness}

\newcommand{\streamterm}[1]{\mathbf{#1}}
\newcommand{\streamtype}[1]{\omega{#1}}
\newcommand{\hd}[1]{{hd}{(#1)}}
\newcommand{\tail}[1]{{tl}{(#1)}}
\def\betar{\to_\beta}
\newcommand{\sub}{\mathsf{sub}}
\newcommand{\tier}[1]{\mathsf{tier}(#1)}
\newcommand{\detier}{\mathsf{dec}}
\newcommand{\contr}{\mathsf{contr}}
\newcommand{\wk}{\mathsf{wk}}
\newcommand{\letbang}[4]{\mathsf{let\ }{#2}\! ::\!  \oc {#3}={#1} \mathsf{\ in \ }{#4}}
\newcommand{\letn}[5]{\mathsf{let\ }{#2}\! ::\! \ldots  \! ::\!  {#3} \! ::\!  \oc {#4}={#1} \mathsf{\ in \ }{#5}}
\newcommand{\letlist}[3]{\mathsf{let\ }{#2}={#1} \mathsf{\ in \ }{#3}}

\begin{figure*}[t!]\label{fig:type-systems}
\[
		\vlinf{\axr}{}{x: A \vdash x: A}{}
		\quad 	
		\vlinf{\ir \limp }{}{\Gamma \vdash \lambda x. M: \sigma  \limp B}{\Gamma , x:\sigma \vdash M:B}
		\quad
		\vliinf{\er \limp}{}{\Gamma, \Delta \vdash MN : B}{\Gamma \vdash M: \sigma \limp B}{\Delta \vdash N: \sigma }
\]
\medskip
\[
		\vliinf{\ir \otimes}{}{\Gamma, \Delta \vdash M \otimes N : \sigma  \otimes \tau}{\Gamma \vdash M: \sigma}{\Delta \vdash  N: \tau} 
		\quad
		\vliinf{\er \otimes }{}{\Gamma, \Delta   \vdash \lettensor{M}{x}{y}{P}: C}{\Gamma \vdash M : \sigma \otimes \tau}{\Delta, x: \sigma, y: \tau \vdash P: C}
	\]
	\medskip
	\[
		\vlinf{\ir \forall}{}{\Gamma \vdash M: \forall X. A}{\Gamma \vdash M:A} 
		\quad 
		\vlinf{\er \forall}{(\star)}{\Gamma \vdash M: A[B/X]}{\Gamma\vdash M:  \forall X. A}
		\quad
		\vlinf{\ir \identity}{}{\vdash \identity: \unit}{}
		\quad 
		\vliinf{\er\identity}{}{\Gamma, \Delta \vdash	\letid{N}{M} : C}{\Gamma \vdash N: \unit}{\Delta \vdash M: C}
		\]
		\medskip
		\[
		\vlinf{\fprule}{}{\oc \Gamma \vdash   M: \oc  \sigma }{ \Gamma \vdash M : \sigma }
		\quad
		\vlinf{\wnwrule}{}{\Gamma, x: \oc \sigma    \vdash   M :  \tau}{\Gamma \vdash M: \tau}
		\quad
		\vlinf{\wnbrule}{}{\Gamma, x: \oc \sigma    \vdash M[x/y, x/z] :  \tau}{\Gamma, y:  \sigma,  z: \oc\sigma \vdash M: \tau}
	\]
	\medskip
	\[
		\vliiiiinf{\stream}{
			\text{\scriptsize{$\set{\streamterm{M}(i)\mid i\in\Nset}$ is finite}}
		}{\vdash \streamterm{M}: \streamtype \sigma}{\vdash  \streamterm{M}:(0): \sigma}{\vdash  \streamterm{M}:(1): \sigma}{\ldots}{\vdash  \streamterm{M}:(n):  \sigma}{\ldots}
\]
\medskip
\[
		\vlinf{\erasebool}{}{\vdash \erasebool:  \streamtype \sigma \limp \unit}{}
		\quad 
		\vlinf{\popbool}{}{\vdash \popbool: \streamtype \sigma   \limp \sigma  \otimes \streamtype \sigma }{}
\]
	\caption{Typing rules for system $\typestream$ with $(\star)\coloneqq \mbox{$B$  is $(\oc, \omega)$-free}$. }
	\label{fig:typestream}
\end{figure*}

In this section we establish the completeness theorem for $\nupll$ and  $\cpll$ (\Cref{thm:completeness-non-wellfounded}). To this end we introduce  $\typestream$, a type system designed to express computation with access to bits of streams,  and we show that the system can encode   polynomial time Turing machines with  advice. By a similar reasoning,   polynomial time computable functions can be represented in $\pta$, a stream-free subsystem of $\typestream$.  We then translate   the type systems into $\dpll$ and $\pll$, respectively (\Cref{thm:embedding}), and conclude by the simulation theorem relating the inductive and non-wellfounded proof systems  (\Cref{thm:simulation}).

\subsection{The type  systems $\pta$ and $\typestream$}

The type system $\typestream$ is a type-theoretical counterpart of $\dpll$, where the linearity restriction in the second-order rules of~\Cref{fig:sequent-system-pll} is duly reflected by a weaker  polymorphism, and modal formulas ``$\omega \sigma$" express types of streams. We  also introduce $\pta$,  the  stream-free subsystem of $\typestream$ corresponding  to $\pll$.

	\begin{definition}\label{defn:terms}[$\term$]
		We define $\term$ as the set of terms generated by the following grammar:	
		$$
		\arraycolsep=1.4pt
		\begin{array}{rcl}
			M  &\dfn & x \ \vert \  \identity \ \vert \ \letid{x}{M} \ \vert \  	M \otimes M \ \vert \ 	\lettensor M {x_1}{x_2}M \\
			&& \lambda x. M \ \vert \ MM \  \vert \ \streamterm{M}\ \vert \ \erasebool \ \vert \ \popbool 
		\end{array}
		$$where $x$ ranges over a countable set of term variables and $\streamterm{M}= \streamterm{M}(0)\!::\! \streamterm{M}(1)\!::\! \ldots$ is a stream of terms.  Meta-level substitution for terms, written $M[N/x]$, is defined in the standard way. The reduction rules for $\term$ are the following:
		\begin{equation*}
			\small
			\def\arraystretch{1.2}
\begin{array}{ll}
		(\lambda x. M)N \betar M[N/x]&\\
		\lettensor {M \otimes N} {x_1}{x_2}P \betar  P[ M/x_1, N/x_2]& \letid{\identity}{M} \betar M \\
		 		\popbool \, \streamterm M  \betar  \hd{\streamterm{M}} \otimes \tail{\streamterm M}&\erasebool \, \streamterm  M \betar  \identity
	\end{array}
		\end{equation*}
		and apply to any context, where $\hd{\streamterm{M}}$  and $\tail{\streamterm{M}}$ are meta operations  returning, respectively,  head and  tail of $\streamterm M$. With $\betar^*$  we denote  the reflexive   and transitive closure of $\betar$.
	\end{definition}	
	
Type assignment systems for the standard lambda calculus that are based on linear logic  do not satisfy subject reduction, i.e., preservation of typing under normalisation~\cite{failure}. Following~\cite{subject}, there are at least two approaches to recover this property:
\begin{itemize}
\item extend the lambda calculus with explicit constructors and destructors corresponding to the exponential modalities.
\item restrict types to prevent the system typing pathological lambda terms.
\end{itemize} 
Mazza's type systems for parsimonious logic follow the first approach~\cite{Mazza14,Mazza15,MazzaT15}. In this paper we rather adopt the second approach, which will allow us to work with a much simpler system and avoid technicalities. Specifically, we  consider the restricted class of   \emph{essential types} introduced in~\cite{Ronchi-Gaboardi}. Roughly, with essential types the exponential modalities cannot occur to the right of an implication, so that types of the form $A \limp \oc A$ are forbidden.     This restriction  prevents the  system expressing forms of sharing and duplication of data, which cause the failure of subject reduction.

\begin{defn}[$\pta$ and $\typestream$]
	The \emph{essential types} are generated by the following grammar:
	\begin{equation}
	\begin{array}{r@{\;}c@{\;}l@{\qquad}r@{\;}c@{\;}l}
		A &\dfn&  X \ \vert \ \unit \ \vert \  \sigma \multimap A \ \vert \ \forall X . A 
		&
		\sigma& \dfn & A \ \vert \ \sigma \otimes \sigma \ \vert \ \oc \sigma  \ \vert \ \streamtype{\sigma}
	\end{array}
	\end{equation}
	where $X$ ranges over a countable set of type variables. We denote  by $\sigma[\tau/X]$ the meta-level substitution of $\tau$ for the free occurrences of the type variable $X$ in $\sigma$. A \emph{context} is a set of  the form $x_1: \sigma_1, \ldots, x_n: \sigma_n$ for some $n\geq 0$, where the $x_i$'s are pairwise distinct term variables and $\sigma_i$ are types. Contexts range over $\Gamma, \Delta, \Sigma, \ldots$. 
	We denote by $\oc \Gamma$ a context  of the form $x_1: \oc \sigma_1, \ldots, x_n: \oc \sigma_n$.
	The type assignment system for $\term$, called $\typestream$,  derives \emph{judgements} of the form $\Gamma \vdash M: \sigma$ according to the typing rules in~\Cref{fig:typestream}. The restriction of $\typestream$ without the typing rules $\stream$, $\erasebool$ and  $\popbool$ is called $\pta$. We write   $\Gamma \vdash_{\typestream} M : \sigma$ (resp.~ $\Gamma \vdash_{\pta} M : \sigma$) when the judgement  $\Gamma \vdash M : \sigma$ is derivable in $\typestream$ (resp.~$\pta$), omitting the subscript when it is clear from the context. If $\der$ is a typing derivation of $\Gamma \vdash M: \sigma$ then we write $\der:\Gamma \vdash M: \sigma$. 
\end{defn}

Essential types ensure subject reduction 	for $\pta$ and $\typestream$. To show this, we start with a structural property of the type systems that is straightforward consequence of the  linearity restrictions introduced by the essential types.

\begin{proposition} \label{prop:functoriality} If $\der : \Gamma \vdash M: \oc \sigma$ then $\Gamma= \oc \Gamma'$ and $\der$ is obtained from a typing derivation $\der'$ by one application of $\fprule$, followed by a series of applications of $\wnwrule$ and $\wnbrule$.
\end{proposition}
\begin{proof}
	Straightforward, by induction on $\der$. 
\end{proof}

We now introduce substitution properties for both types and  typable terms.

\begin{lemma}\label{lem:substitution-types}
	If $\Gamma \vdash M: \sigma $ then $\Gamma[\vec C/\vec X] \vdash M: \sigma[\vec C/\vec X]$ for every $\vec C= C_1, \ldots, C_n$ and $\vec X= X_1, \ldots, X_n$.
\end{lemma}

\begin{lemma}[Substitution]\label{lem:substitution}
	If  $\der_1:\Gamma, x: \tau \vdash M: \sigma $ and $\der_2: \Delta \vdash N: \tau$ then there is a typing derivation  $S(\der_1, \der_2)$ of $\Gamma, \Delta \vdash M[N/x]: \sigma $.
\end{lemma}
\begin{proof}
	The proof is by induction on the  lexicographic order over  $(h(\der_1), s(\tau))$, where $h(\der_1)$ is the height of $\der_1$ and $s{(\tau)}$ is the number of symbols of the formula $\tau$. The only interesting case is when $\der_1$ is obtained from $\der'_1$ by applying a $\wnbrule$ rule, where  $M= M'[x/y, x/z]$,  $\tau= \oc \tau'$, and  $\der'_1$ is
	\begin{equation*}
        \small
		\vlderivation{
			\vlin{\wnbrule}{}{\Gamma, x: \oc \tau' \vdash M'[x/y, x/z]: \sigma}{\vltr{\der'_1}{\Gamma, y:\tau', z: \oc \tau' \vdash M': \sigma}{\vlhy{\ \  } }{\vlhy{\ \ }  }{\vlhy{\ \ }   }}
		}
	\end{equation*}
	By~\Cref{prop:functoriality} we have that $\Delta= \oc \Sigma$ and
	\begin{equation*}
        \small
		\der_2 \dfn 
		\vlderivation{
			\vliq{\wnbrule, \wnwrule}{}{\oc \Sigma \vdash N: \oc \tau'}
			{\vltr{\der'_2}{\oc \Sigma' \vdash N':\oc  \tau'}{\vlhy{\ \  } }{\vlhy{\ \ }  }{\vlhy{\ \ }   }}
		}
		\qquad 
		\der'_2 \dfn 
		\vlderivation{
			\vlin{\fprule}{}{\oc \Sigma' \vdash N': \oc \tau'}{\vltr{\der''_2}{\Sigma' \vdash N': \tau'}{\vlhy{\ \  } }{\vlhy{\ \ }  }{\vlhy{\ \ }   }}
		}
	\end{equation*}
	Since $h(\der_1')<h(\der_1)$ then $(h(\der_1'), s(\oc \tau))<(h(\der_1), s(\oc \tau))$ and by induction hypothesis we have a typing derivation $S(\der'_1, \der_2)$ of $\Gamma, \oc \Sigma' \vdash M'[N'/z]: \sigma $. Moreover, since $s(\tau)<s(\oc \tau)$ then $(h(S(\der'_1, \der_2)), s(\tau))<(h(\der_1), s(\oc \tau))$, and by applying the  induction hypothesis again  there is a typing derivation $S(S(\der'_1, \der_2), \der'_2)$ of $\Gamma, \Sigma', \oc \Sigma' \vdash M'[N'/z,  N'/ y ]: \sigma $. We conclude by applying a series of $\wnbrule$ rules and $\wnwrule$ rules.
\end{proof}

\begin{restatable}[Subject reduction]{proposition}{SUBJECT}\label{prop:subject-reduction}
	Let $\der:\Gamma\vdash M: \sigma$. If  $M\betar N$ then there is $\der'$ such that $\der':\Gamma \vdash N : \sigma$.
\end{restatable}
\begin{proof}[Proof]
	It suffices to check that the reduction rules given in~\Cref{defn:terms} preserve types. We consider the most interesting reduction rule, i.e.,  $M=(\lambda x.P)Q\betar P[Q/x]=N$. By  inspecting the typing rules in~\Cref{fig:typestream},  $\der$ must have the following structure:
	{
		\vspace{-0.2cm}
		$$
            \small
		\vlderivation{
			\vlin{}{}{\Sigma, \Delta \vdash (\lambda x. P)Q: \sigma}
			{
				\vlin{}{}{\delta\ \vdots}
				{
					\vliin{\er{\limp}}{}{\Sigma', \Delta' \vdash (\lambda x. P')Q':  B'}
					{
						\vltr{\der_1}{\Sigma' \vdash \lambda x. P': \tau \limp B'}{\vlhy{\ }}{\vlhy{\ }}{\vlhy{\ }}            
					}{
						\vltr{\der_2}{\Delta' \vdash Q' : \tau}{\vlhy{\ }}{\vlhy{\ }}{\vlhy{\ }}
					}
				}
			}	
		}
		$$
	}
	where:
	\begin{itemize}
		\item  $\Gamma= \Sigma, \Delta$ and $\sigma= {{\oc \overset{n}{\ldots} \oc}} \forall \vec X. B$, for some $n\geq 0$,  $\vec X$, and $B$.
		\item $\delta$ is a sequence of rules in $\{ \wnwrule, \wnbrule, \fprule, \ir \forall, \er \forall\}$, 
		\item $B= B'[\vec C/\vec Y]$, for some $\vec C$ and $\vec Y$ not free in $\Sigma', \Delta'$
		\item $P'[\vec x/\vec y]=P$ and $Q'[\vec x/\vec y]=Q$, for some $\vec x, \vec y$. 
	\end{itemize}
	By a similar reasoning, $\der_1$ has the following shape:
	$$
        \small
	\vlderivation{
		\vlin{}{}{\Sigma' \vdash \lambda x. P': \tau \limp B'}
		{
			\vlin{}{}{\varepsilon\ \vdots}
			{
				\vlin{\ir \limp}{}{\Sigma''\vdash \lambda x. P'': \tau' \limp B''}
				{
					\vltr{\der'_1}{\Sigma'', x:\tau' \vdash P'':  B''}{\vlhy{\ }}{\vlhy{\ }}{\vlhy{\ }}
				}
			}
		}  
	}
	$$
	where:
	\begin{itemize}
		\item $\varepsilon$ is sequences of typing rules in $\{ \wnwrule, \wnbrule, \ir \forall, \er \forall\}$,  
		\item $B'= B''[\vec D/\vec Z]$ and $\tau= \tau'[\vec D/\vec Z]$, for some $\vec D$ and $\vec Z$ not free in $\Sigma''$
		\item $P''[\vec z/\vec w]=P'$ for some $\vec z, \vec w$. 
	\end{itemize}
	Since $\tau= \tau'[\vec D/\vec Z]$, $B'= B''[\vec D/\vec Z]$ and $\vec Z$ do not occur free in $\Sigma''$,  by~\Cref{lem:substitution-types}  there is a typing derivation  $\der''_1$ of $\Sigma'', x: \tau \vdash P'': B'$. By~\Cref{lem:substitution} there is a typing derivation $S(\der''_1, \der_2)$ of $\Delta', \Sigma'' \vdash P''[Q'/x]: B'$. Finally, by applying the sequences of rules $\delta$ and $\varepsilon$ we obtain:
	{
		\vspace{-0.2cm}
		$$
            \small
		\vlderivation{
			\vlin{}{}{\Sigma, \Delta \vdash P[Q/x]: \sigma}
			{
				\vlin{}{ }{  \vdots} 
				{
					\vlin{}{}{\Delta', \Sigma' \vdash P'[Q'/x]: B'}
					{
						\vlin{}{}{  \vdots}
						{
							\toks0={0.5}
							\vltrf{S(\der''_1, \der_2)}{\Delta', \Sigma'' \vdash P''[Q'/x]: B'}{\vlhy{\ \  \ 	\  \ \  \  \  }}{\vlhy{\ \ \ \  \ \ \  	 }}{\vlhy{\ \ \  \  \ \  \ 	 }}{\the\toks0}
						}
					}
				}
			}
		}
		$$
	}
\end{proof}

\subsection{Completeness results for  $\typestream$ and $\pta$}

In this subsection we show  completeness of  $\typestream$ and $\pta$  for, respectively,  $\fppoly$ and $\fptime$. The proof adapts to our setting the encoding of polynomial time Turing machines from~\cite{MairsonTerui,Ronchi-Gaboardi}.


\subsubsection{Definability and  data types in $\pta$ and $\typestream$}\label{subsec:definability}

Polymorphic type systems based on linear logic typically encode inductive datatypes  by universally quantified types  (see, e.g.,~\cite{Ronchi-Gaboardi}). Examples are natural numbers,  defined by  $\Nat \dfn \forall X. \oc (X \limp X)\limp X\limp X$. Because of linearity restrictions on polymorphism, however,  parsimonious logic  cannot freely apply instantiation when  encoding functions over inductive datatypes. As a consequence, its computational strength relative to standard notions of representability~\cite{Barendregt1981-BARTLC}   would be  fairly poor. To circumvent this technical issue, following previous works on parsimonious logic~\cite{MazzaT15, Mazza15}, we  adopt a  parametric notion of representability, where natural numbers are defined by types of the form $  \oc(A \limp A)\limp A\limp A$, i.e., by  instantiations of $\Nat$.

To this end, we generalise the usual notion of lambda definability~\cite{Barendregt1981-BARTLC} to different kinds of input data:

\begin{definition}[Representability~\cite{Ronchi-Gaboardi}]
	Let $f: \mathbb{I}_1 \times \ldots \times \mathbb{I}_n\to \mathbb{O}$  be a total function and let  the elements $o \in \mathbb{O}$ and $i_j \in \mathbb{I}_j$ for $0 \leq j \leq n$ be encoded by terms $\cod{o}$ and $\cod{i_j}$ such that $\vdash \cod o: \mathbf{O}$ and $\vdash \cod {i_j}: \mathbf{I}_j$. Then, $f$ is \emph{representable} in $\typestream$ (resp.~$\pta$) if there is a term $\cod f\in \term$ such that $\vdash \cod f:  \mathbf{I}_1 \limp \ldots \limp \mathbf{I}_n\limp \mathbf{O}$ in $\typestream$ (resp.~$\pta$) and
	\[
	f\, i_1 \ldots i_n =o \quad \Longleftrightarrow \quad \cod f \, \cod{i_1}\ldots \cod{i_n} \betar^* \cod o
	\]
\end{definition}

We adopt the usual notational convention $M^n\,  N$ ($n \geq 0$) defined inductively as $	M^0 (N) \dfn N$ and $M^{n+1}(N) \dfn M(M^n(N))$. We also set $M\circ N \dfn \lambda z. M(Nz) $, which generalises to the $n$-ary case $M_1 \circ \overset{}{\ldots} \circ  M_{n} \dfn  \lambda z. M_1(M_2(\ldots (M_n z)))$. Finally, the  $n$-ary tensor product (with $n\geq 3$) can be  defined from the binary one by setting 
$	M_1 \otimes \ldots \otimes M_{n}\dfn 
(M_1 \otimes \ldots \otimes M_{n-1}) \otimes M_{n}$, $\sigma _1 \otimes \ldots \otimes \sigma_n \dfn (\sigma _1 \otimes \ldots \otimes \sigma_{n-1})\otimes \sigma_n $,  and $	\lettensor{z}{x_1\otimes \ldots}{x_n}{M} \dfn \lettensor{z}{y}{x_n}{(\lettensor{y}{x_1\otimes \ldots}{x_{n-1}}{M})}$. 
We also  use the shorthand notation $\sigma^n \dfn \sigma \otimes \overset{n}{\ldots}\otimes \sigma$.

\renewcommand{\succ}[1]{\mathsf{s}_{#1}}
\newcommand{\suc}[1]{\mathsf{s}_{#1}}
\newcommand{\pred}{\mathsf{p}}
\newcommand{\floor}[1]{\lfloor {#1} \rfloor}
\newcommand{\hlf}[1]{\floor{#1 /2}}
\newcommand{\cnd}{\mathsf{cond}}
\newcommand{\proj}[2]{\pi^{#1}_{#2}}
\newcommand{\s}[1]{\vert {#1}\vert}
\newcommand{\ifte}[3]{\mathsf{if}\, {#1}\, \mathsf{then\, }{#2}\, \mathsf{else}\, {#3}}

In what follows we encode some relevant  data types  and their basic operations in $\typestream$ (and $\pta$).

\begin{definition}[Booleans]
	Booleans $\false, \true$ and basic  Boolean  operations are  encoded as in~\Cref{fig:functional-complete}.
	The can be typed by $\Bool\dfn \forall X. (X \otimes X)\limp (X \otimes X)$. 
\end{definition}
\begin{figure*}
	\	\adjustbox{max width=.9\textwidth}{$
		\def\arraystretch{1.2}
		\begin{array}{rcrl}
			\cod{\true}&\dfn & \lambda x. \lambda y. x \otimes y   & :\Bool\\
			\cod \false &\dfn & \lambda x. \lambda y. y \otimes x & : \Bool\\
			\erase& \dfn & \lambda b. \lettensor{b (\identity \otimes \identity )}{x_1}{x_2}{\letid {x_2}{x_1}}&: \Bool \limp \unit\\
			\pi^2_1 &\dfn & \lambda x. \lettensor{x}{x_1}{x_2}{\letid {\erase \, x_2}{x_1}}&: \Bool \otimes \Bool \limp \Bool\\ 
			\dupl & \dfn &  \lambda b. \pi^2_1 (b (\cod \true \otimes \cod \true)\otimes (\cod \false \otimes \cod \false))&: \Bool \limp \Bool \otimes \Bool \\
			\cod{\neg}&\dfn& \lambda b. \lambda x. \lambda y. b(y \otimes x) &: \Bool \limp \Bool \\
			\cod{\vee}& \dfn & \lambda b_1. \lambda b_2. \pi_1^2(b_1 \cod \false \, b_2)&: \Bool \otimes \Bool \limp \Bool
		\end{array}
		$}
	\caption{Encoding of basic operations on Booleans.}
	\label{fig:functional-complete}
\end{figure*}

The following is a straightforward consequence of the  encodings in~\Cref{fig:functional-complete}.

\begin{proposition}[Functional completeness]  \label{prop:fun-comple} Every Boolean function $f: \{\false, \true\}^n \to \{\false, \true\}^m$ with $n \geq 0$, $m> 0$ can be represented by a term $\cod f\in \term$ such that $\vdash \cod f: \Bool^n \limp \Bool^m $.
\end{proposition}

Notice that~\Cref{prop:fun-comple} crucially relies on the terms $\dupl$ and $\erase$, which duplicate and erase Booleans in a purely linear fashion. Following~\cite{MairsonTerui}, we can generalise  linear erasure of data to a fairly large class of types.

\begin{definition}[$\Pi_1 $ and $e\Pi_1$ types~\cite{MairsonTerui}] 
	Let $A$ be a type build from $ \unit , \otimes, \limp,  \forall$.  
	We say that $A$ is $e\Pi_1$ if every $\forall$-type  occurring in it is inhabited.
\end{definition}

\begin{proposition}[Linear erasure~\cite{MairsonTerui,CurziR20}] \label{prop:linear-erasure} For any closed type $A$ in $e\Pi_1$ there is a term  $\erasetype A$ in a term  that inhabits $A \limp \unit$.
\end{proposition}

\begin{proposition}[Conditional]\label{prop:conditional}
	For any  $A$ in $e\Pi_1$,  the following rule  is derivable:
	$$\vliinf{\cnd}{}{x: \Bool \vdash \ifte{x}{R}{L}: A}{\vdash R: A}{\vdash L:A}
	$$
	where $ \ifte{x}{R}{L}$ satisfies the following reductions:
	\[
	\def\arraystretch{1.2}
	\arraycolsep=2pt
	\begin{array}{rcl}
		\ifte{\cod{\true}}{R}{L} &\betar^*& R\\
		\ifte{\cod{\false}}{R}{L} &\betar^*& L
	\end{array}
	\]
\end{proposition}
\begin{proof}
	We set $\ifte{x}{R}{L}\dfn \pi^2_1 (x RL)$, where $\pi^2_1$ is as in~\Cref{fig:functional-complete}.
\end{proof}

\begin{definition}[Streams of Booleans] \label{defn:stream} A stream  (of Booleans) $\alpha$ is encoded by a term $\streamterm{M}$ such that $\streamterm{M}(i)\dfn \cod {\alpha(i)}$. We write $\cod{\alpha}$ for the encoding of $\alpha$. 
	Streams can be typed by $\Stream \dfn \streamtype{\Bool}$.  
\end{definition}

\begin{definition}[Natural numbers and Boolean strings]  \label{defn:strings}
	The encoding of Boolean strings and natural numbers is as follows, for any $n \geq 0$ and  $s = b_1\cdots b_n\in \{\false, \true\}^*$:
	\[
	\def\arraystretch{1.2}
	\arraycolsep=2pt
	\begin{array}{rcl}
		\cod{n}&\dfn& \lambda f. \lambda z. f^nz\\
		\cod{s}&\dfn& \lambda f.\lambda  z. f\, \cod{b_n} (f\, \cod{b_{n-1}}(\ldots (f\, \cod {b_1} \, z)\ldots ) )
	\end{array}
	\]
	For any type $A$, natural numbers and Boolean strings can be typed, respectively, by
	\[
	\def\arraystretch{1.2}
	\arraycolsep=2pt
	\begin{array}{rcl}
		\Nat[A]&\dfn& \oc (A \limp A)\limp A\limp A\\
		\String[A]&\dfn&\oc (\Bool \limp A \limp A)\limp A\limp A
	\end{array}
	\]
	With $\Nat[]$ we denote $\Nat[A]$ for some $A$, and similarly for $\String[]$.	
\end{definition}

We need to encode the function that, when applied to a  Boolean string, returns  its length:

\begin{proposition}[Length]\label{prop:length}
	There exists a term $\length$  of type $ \String[A]\limp \Nat[A]$  satisfying the following reduction,  for all  $s = b_1\cdots b_n\in \{\false, \true\}^*$:
	\[
	\length \, \cod{s} \betar^* \cod{n}
	\]
\end{proposition}	
\begin{proof}
	We set $\length  \dfn     \lambda s.  \lambda   f. s  (\lambda x. \lambda y. \letid {\! \!\erase\,  x}{fy} )     $, where $\erase$ is as in~\Cref{fig:functional-complete}.
\end{proof}
	%

	The following proposition shows that encodings of natural numbers and Boolean strings can be used as iterators.
	
	\begin{proposition}[Iteration]
		For any $A$, the following rule is derivable
		$$
		\vliinf{\itrule}{}{\oc \Gamma, \Delta, n: \Nat[A] \vdash \iter  n S  B: A }{\oc \Gamma \vdash S: \oc (A \limp A)}{\Delta \vdash B: A}
		$$
		where $\iter  n S  B$ satisfies the  reduction
		\[
		\begin{array}{rcl}
			\iter {\cod{n}} SB & \betar ^* & S^nB
		\end{array}
		\]
		Similarly,  the following rule is derivable:
		$$
		\vliiinf{\itrulestring}{}{\oc\oc \Gamma, \Delta, n: \String[A] \vdash \iterstring  n  {S_0} {S_1}  B: A }{\oc \Gamma \vdash S_0: \oc (A \limp A)}{\oc  \Gamma \vdash S_1: \oc (A \limp A)}{\Delta \vdash B: A}
		$$
		where $\iterstring  n  {S_0} {S_1}  B$ satisfies the  reduction
		\[
		\begin{array}{rcl}
			\iter {\cod{b_1\cdots b_n}} {S_0} {S_1} B & \betar ^* & S_{b_n}\ldots S_{b_1} B
		\end{array}
		\] 
	\end{proposition}
	\begin{proof}
		It suffices to set, respectively,    $\iter \dfn \lambda n. \lambda  s. \lambda b.  n  sb$ and $\iterstring \dfn \lambda s. \lambda  t.\lambda u. \lambda b.  s  tub$. 
	\end{proof}
	
	Our next goal is to show that any polynomial over natural numbers can be encoded in $\typestream$ (and $\pta$). 	The encoding of  polynomials requires nesting types, so we introduce a notation for denoting iterated nesting  in a succinct way.

	\begin{definition}[Nesting]\label{defn:tiering}
		Let $A$ be a type.	We define $\Nat_A[d]$ and $\String_A[d]$ by induction on  $d \geq 0$:
		\[
		\def\arraystretch{1.2}
		\arraycolsep=2pt
		\begin{array}{rcl}
			\Nat_A[0]&\dfn & A\\
			\Nat_A[{d+1}] &\dfn & \Nat_A[\Nat_A[d]]
		\end{array}
		\qquad   
		\def\arraystretch{1.2}
		\arraycolsep=2pt
		\begin{array}{rcl}
			\String_A[ 0]&\dfn & A\\
			\String_A[{d+1}] &\dfn & \String_A[\String_A[d]]
		\end{array}
		\]
		If $A$ is clear from the context, we simply write $\Nat[d]$ and $\String[d]$.   
	\end{definition}
	
	\begin{proposition}
		For any $d \geq 0$, there  exist a term $	\down d$ of type  $\Nat[{d+1}] \limp \Nat[{d}]$ and a term $	\dows d$ of type $\String[{d+1}] \limp \String[{d}]$ 
		satisfying the following reductions:
		\[
		\begin{array}{rcl}
			\down d \, \cod n &\betar^*& \cod n\\
			\dows d \, \cod n &\betar^*& \cod n
		\end{array}
		\]
	\end{proposition}
	\begin{proof}
		We set  $	\down d\dfn \lambda x.\iter x \success {\cod{0}}$ and $	\dows d\dfn   \lambda x.\iter x  (\lambda b. \lambda s. \lambda c. \lambda z. cb(scz)  ) \cod{\epsilon}$.
	\end{proof}

	\begin{definition}[Successor,  addition, multiplication]
		Successor, addition and multiplication can be represented by the following terms:
		$$
		\def\arraystretch{1.2}
		\arraycolsep=2pt
		\begin{array}{rcll}
			\success & \dfn & \lambda n.  \lambda  f. \lambda z. n  ( f) (fz)&\\
			\add & \dfn & \lambda n. \lambda m.  \iter n ( \success)\,  m   \\
			\mult & \dfn & \lambda  n. \lambda   m.\iter m {(\lambda y. \add  \,  n \, y)} {\cod 0} &
		\end{array} 
		$$
		they are typable as follows:
		$$
		\def\arraystretch{1.2}
		\arraycolsep=2pt
		\begin{array}{rl}
			&	\vdash 	\success : \Nat[i] \limp \Nat[i] \\
			&	\vdash \add  :  \Nat[i+1] \limp \Nat[i] \limp \Nat[i] \\
			&\vdash 	\mult  :    \oc \Nat[i+ 1] \limp  \Nat[i+ 1]\limp \Nat[i]
		\end{array} 
		$$
		
	\end{definition}



\begin{theorem}[Polynomial completeness]\label{thm:polynomial-completeness}
		Let $p(x): \Nset \to \Nset$ be a polynomial with degree $\degree p>0$. Then there is a term $\cod p$ representing $p$ such that, for any  $i\geq 0$:
		\[
		x: \oc ^{\degree p-1} \Nat[\degree p +i] \vdash \cod p : \Nat[i]
		\]
\end{theorem}
\begin{proof}
	For the sake of readability, we will avoid writing the index $i$ in typing  judgements. Thus, with $\Nat[n]$ we mean $\Nat[n+i]$.

	Consider a polynomial $p(x): \Nset \to \Nset$ in Horner normal form, i.e., $p(x)= a_0+ x(a_1+ x(\ldots (a_{n-1}+ xa_n)\ldots))$. We actually show   something stronger:
	\begin{equation}\label{eqn:stronger-invariant}
		x_0:\Nat[1],x_1: \oc \Nat[2],  \ldots, x_{ n-1}: \oc ^{ n-1}\Nat[n] \vdash \cod {\hat{p}}: \Nat[0]
	\end{equation}
	where 	$\hat{p}= a_0+ x_0(a_1+ x_1(\ldots (a_{n-1}+ x_{n-1}a_{n})\ldots))$.  The proof is by induction on $\degree {{p}}=n$.	If $\degree {{p}}=1$ then $\hat{p}= a_0+ x_0a_1$, and we simply set $\cod{\hat{p}}\dfn \add \, \cod {a_0}\, {(\mult\, \cod{a_1}\, x_0)}$. If $\degree {{p}}>1$ then $\hat{p}= a_0+ x_0\hat{q}$ with $\hat{q}\dfn a_1+ x_1(a_2+ x_2(\ldots (a_{n-1}+ x_{n-1}a_{n})\ldots)) $.  By induction hypothesis on $q$ we have
	\[
	x_1:\Nat[1],x_2: \oc \Nat[2],  \ldots, x_{n-1}: \oc ^{n-2}\Nat[n-1] \vdash \cod {\hat{q}}: \Nat[0]
	\]
	By repeatedly applying $\down{k}$ for appropriate $k$ we obtain  a term $M$ such that:
	\[
	x_1:\Nat[2],x_2:\oc  \Nat[3],  \ldots, x_{n-1}: \oc ^{n-2}\Nat[n] \vdash M: \Nat[0]
	\]
	We set $\cod{\hat{p}}\dfn \add \, \cod{a_0}\, (\mult \, M\, x_0)$, which is typable as:
	\[
	x_0: \Nat[1], x_1:\oc \Nat[2],  \ldots, x_{n-1}: \oc ^{n-1}\Nat[n] \vdash \cod{\hat{p}}: \Nat[0]
	\]	
	and we can conclude since $\degree{{q}}= \degree{{p}}-1$.
	
	Now, to prove the theorem it suffices to repeatedly apply $\down{k}$ for appropriate $k$  to the typable term in~\Cref{eqn:stronger-invariant} in order to get a term  $N$ that represents $\hat{p}$ and typable as
	\[
	x_0:\Nat[n],x_1: \oc \Nat[n],  \ldots, x_{n-1}: \oc ^{n-1}\Nat[n] \vdash \cod {\hat{p}}: \Nat[0]
	\]
	By applying a series of $\wnbrule$ we obtain a term $\cod p$ representing the polynomial $p$ and such that $	x: \oc ^{\degree p -1} \Nat[\degree p ] \vdash \cod p : \Nat[0]$. 
\end{proof}
	
\subsubsection{Encoding polytime Turing machines with polynomial advice in $\typestream$}

In this subsection we show how to encode a Turing machine working in polynomial time with access to advice in $\typestream$, following essentially~\cite{MairsonTerui,Ronchi-Gaboardi}.  

W.l.o.g., we will assume that the alphabet  of the machine is composed by the two symbols $\true$ and $\false$\footnote{We can encode the alphabet $\{\false, \true, \sqcup\}$ within the alphabet $\{\false, \true\}$, by setting $\sqcup \dfn\false\false$, $\false\dfn \false \true$ and $\true\dfn \true \true$.},  and that final states are divided into accepting and rejecting.  

A \emph{configuration} $C$ of the machine will be represented by a term  of the following form:
\begin{equation}\label{eqn:configuration}
	\cod{C}\dfn	\lambda  c. (c\cod{b^l _0} \circ \ldots \circ c \cod{b^l_n} )\otimes (c\cod{b^r_0}\circ \ldots \circ c\cod{b^r_m}) \otimes \cod{q} \otimes  \cod{\alpha} 
\end{equation}
where:
\begin{itemize}
	\item  	$b^r_0\in \{\false, \true\}$ is the scanned symbol 
	\item 	$b^r_1 \cdots b^r_m\in \{\false, \true\}$ are the symbols of the tape  to the right of the scanned symbol
	\item   $b^l_n\cdots b^l_0\in \{\false, \true\}$ are the symbols of the tape  to the left of the scanned symbol (notice that we encode this tuple in reverse order)
	\item  $q=b_1\cdots b_k\in \{\false, \true\}$ is the (encoding of the) current state of the machine. 
	%
	\item $\alpha$ represents the advice of the machine as a single Boolean  stream (see~\Cref{prop:fppoly=fptime(RR)}). 
\end{itemize}

Terms as in~\Cref{eqn:configuration} have the following type:
$$
\tm \dfn \forall X. \oc (\Bool \limp X \limp X) \limp ((X \limp X)^2\otimes  \Bool^k\otimes\Stream )
$$
where $\Stream$ is as in~\Cref{defn:stream}.

The \emph{initial configuration} $C_0$ describes a machine with tape filled by blank symbols (here $\false\false$s) the head at the beginning of the tape and in the initial state $q_0$.  To render the construction of the initial configuration in $\typestream$, we define  the  following  term: 
\begin{equation}\label{eqn:init}
	\init \dfn \lambda n . \lambda  c. (\lambda z.z)\otimes  n ( \lambda z'. c \cod{\false}(c \cod{\false} z')) \otimes  \cod{q_0} \otimes \cod{\alpha}
\end{equation}
It takes the encoding of a natural number $n$ in input and returns the  term
$$
\cod{C_0}\dfn	\lambda  c. (\lambda z. z) \otimes (c\cod{\false}\circ \overset{2n}{\ldots} \circ c\cod{\false}) \otimes \cod{q} \otimes  \cod{\alpha} 
$$
representing the first $n$ blank symbols of the tape. Terms as in~\Cref{eqn:init} have the  type below
$$
\Nat[d] \limp \Stream \limp \tm  
$$
for all $d \geq 0$.

Following~\cite{MairsonTerui,Ronchi-Gaboardi}, in order to show that Turing machine transitions are representable we consider two distinct phases:
\begin{itemize}
	\item A \emph{decomposition phase}, where the encoding of the  configuration $C$ is decomposed to extract the symbols $\cod{b^l_0}$,  $\cod{b^r_0}$. 
	\item A \emph{composition phase}, where the components of $C$ are assembled back to get  the  configuration of the machine  after   the transition. 
\end{itemize}

The decomposition of a configuration has type $\idtype$:
$$
\idtype \dfn \forall X. \oc (\Bool \limp A[X])\limp  (A[X]^2\otimes B[X] ^2 \otimes \Bool^k \otimes \Stream)
$$
where $A[X]\dfn X \limp X$ and $B[X]\dfn (\Bool \limp A[X])\otimes \Bool$. 

The decomposition phase is described by the  term $\decomp$ of type $\tm \limp \idtype$ defined as follows:
\begin{multline}\label{eqn:decomp}
	\decomp \dfn  \lambda m. \lambda c. \lettensor{m\, (F[c])}{l}{r  \otimes q\otimes \alpha}{} \\
	(\lettensor{l(\identity \otimes (\lambda x. \letid{\erasetype{\Bool}\, x}{\identity}) \otimes \cod \false)}{s_l}{c_l \otimes b^l_0}{}\\
	(\lettensor{r(\identity \otimes (\lambda x. \letid{\erasetype{\Bool}\, x}{\identity}) \otimes \cod \false)}{s_r}{c_r \otimes b^r_0}{}\\
	s_l \otimes s_r \otimes c_l \otimes \otimes  b^l_0 \otimes c_r \otimes b^r_0\otimes q \otimes \cod \alpha ))
\end{multline}
where  $\erasetype{\Bool}$ is the eraser for $\Bool$ given by~\Cref{prop:linear-erasure}, and $F[x]\dfn \lambda b. \lambda z. \lettensor{z}{g}{h\otimes i}{(h\, i  \circ g )\otimes x \otimes b}$, which is typable as:
\[
x: \Bool \limp A[X ]\vdash  F[x]: B[(A[X]\otimes B[X])/X]
\]

The term $\decomp$ in~\Cref{eqn:decomp} satisfies the reduction in~\Cref{fig:red-decomp-comp}.
\begin{figure*}[t]
	\centering
	$$
	\begin{array}{c}
		\decomp \, (\lambda  c. (c\cod{b^l _0} \circ \ldots \circ c \cod{b^l_n} )\otimes (c\cod{b^r_0}\circ \ldots \circ c\cod{b^r_m}) \otimes \cod{q} \otimes  \cod{\alpha} ) \\[0.3cm]
		\betar^*\\[0.3cm]  \lambda  c. (cb^l _1 \circ \ldots \circ c b^l_n )\otimes (cb^r_1\circ \ldots \circ cb^r_m) \otimes c \otimes b^l_0 \otimes c \otimes b^r_0  \otimes \cod{q} \otimes  \cod{\alpha}\\[4ex]
	\end{array}
	$$
	$$
	\begin{array}{c}
		\comp \ (	\lambda  c. (cb^l _1 \circ \ldots \circ c b^l_n )\otimes (cb^r_1\circ \ldots \circ cb^r_m) \otimes c \otimes b^l_0 \otimes c \otimes b^r_0  \otimes \cod{q}\otimes  \cod{\alpha})\\[0.3cm] \betar^* \\[0.3cm]
		\begin{cases}
			\lambda  c. (cb' \circ cb^l _0 \circ \ldots \circ c b^l_n )\otimes (cb^r_1\circ \ldots \circ cb^r_m)  \otimes \cod{q'}\otimes \cod{\tail{\alpha}} &\text{if } \delta (b^r_0, \hd{\alpha},  q)= (b', q', \text{Right})\\[1ex]
			\lambda  c. (cb^l _1 \circ \ldots \circ c b^l_n )\otimes (c b^l_0 \circ cb' \circ cb^r_1\circ \ldots \circ cb^r_m)  \otimes q'\otimes \cod{\tail{\alpha}}&\text{if } \delta (b^r_0, \hd{\alpha},  q)= (b', q', \text{Left})
		\end{cases}
		%
	\end{array}
	$$
	\caption{Reductions for $\decomp$ and $\comp$.}
	\label{fig:red-decomp-comp}
\end{figure*}

Analogously, the composition phase is described by the term $\comp$ of type $\idtype \limp \tm$ defined as follows:
\begin{multline}\label{eqn:comp}
	\comp \dfn \lambda s. \lambda c. \lettensor{s\, c }{l}{r \otimes c_l \otimes b_l \otimes c_r \otimes b_r\otimes q \otimes \alpha }{}\\ 
	\lettensor{\popbool \, \alpha}{h}{t}{} (\lettensor{\cod{\delta}(b_r \otimes h \otimes q)}{b'}{q'\otimes m}{}\\
	((\ifte m R L ) b'q' (l \otimes r \otimes c_l \otimes b_l \otimes c_r)\otimes t 
	)
\end{multline}
where $\ifte m R L $ is defined as in~\Cref{prop:conditional}, and
\[
\def\arraystretch{1.2}
\arraycolsep=1.4pt
\begin{array}{rcl}
	R &\dfn & \lambda b' .\lambda q'. \lambda s.\mathsf{\, let}\, l \otimes r\otimes c_l \otimes b_l \otimes c_r =s \mathsf{\, in \, }(c_r b' \circ c_l b_l\circ l)\otimes r \otimes q'\\[2ex]
	L &\dfn & \lambda b' .\lambda q'. \lambda s.  \, \mathsf{let\, }l \otimes r\otimes c_l \otimes b_l \otimes c_r=s\mathsf{\, in\, }  l \otimes (c_l b_l \circ c_r b' \circ r) \otimes q'\\
\end{array}
\]
The term $\comp$ in~\Cref{eqn:comp} satisfies the reduction in~\Cref{fig:red-decomp-comp}, where $\delta: \{\false, \true\}^{k+2}\to \{\false, \true\}^{k+3}$ is the transition function of the Turing machine, which takes as an extra input  the first bit of the current advice stack, i.e., the head of the stream $\alpha$.

\begin{remark}\label{rem:instantiation}
	Notice that in $\comp$ the variable $m$ has type $\Bool$ and is applied to the terms $R$ and $L$. This requires to apply to the variable $m$ the rule $\er \forall$, which instantiates the type variable $X$ with the $\oc$-free type $\Bool \limp \Bool^k \limp ((X \limp X)^2 \otimes (\Bool \limp X \limp X)\otimes \Bool \otimes (\Bool \limp X \limp X) )\limp (X \limp X)^2\otimes \Bool^k$.
\end{remark}

By combining the above terms we obtain the encoding of the  Turing machine transition step:
\begin{equation}\label{eqn:transition}
	\tr \dfn \comp \circ \decomp
\end{equation}
with type $\tm \limp \tm$.

We now need a  term that encodes the  \emph{initialisation} of  the  machine with an input Boolean  string. This is given by the term $\In$ 
of type $\String[\tm]\limp \tm \limp \tm$ defined as follows:
\begin{equation}\label{eqn:initialisation}
	\In \dfn \lambda s. \lambda m. s(\lambda b. (Tb)\circ \decomp)\, m
\end{equation}
where $\decomp$ is defined as in~\Cref{eqn:decomp} and
\begin{equation*}
	\def\arraystretch{1.2}
	\arraycolsep=1.4pt
	\begin{array}{rcl}
		T &\dfn&  \lambda b. \lambda s. \lambda c.\,  \lettensor{sc}{l}{r \otimes c_l \otimes b_l \otimes c_r \otimes b_r \otimes q \otimes \alpha}{}\\
		&&	(\letid{\erasetype{\Bool}\, b_r}{(Rbq(l \otimes r \otimes c_l \otimes b_l \otimes c_r))}\otimes \alpha) \\[2ex]
		R& \dfn& \lambda b'. \lambda q'. \lambda s. \, \lettensor{s}{l}{r \otimes c_l \otimes b_l \otimes c_r}{}((c_r b' \circ c_l b_l \circ l)\otimes r \otimes q')
	\end{array}
\end{equation*}
Intuitively, the term $\In$ defines a function that, when supplied with a Boolean string $s$ and a Turing machine $M$, writes $s$ as input  on the tape of $M$. 

Finally, we need a term that \emph{extracts} the output string from the final configuration. This is given by the term $\extr$ of type $\tm \limp \String$, defined as follows:
\begin{equation}\label{eqn:extr}
	\def\arraystretch{1.2}
	\arraycolsep=1.4pt
	\begin{array}{rcl}
		\extr& \dfn& \lambda s. \lambda c. \, \lettensor{sc}{l}{r \otimes q \otimes \alpha}{}({\letid{\erasetype{\Bool^{k+1}}\, (q \otimes (\erasebool 	\, {\alpha}))}{l \circ r})}
	\end{array}
\end{equation}
where $\erasebool$ is the eraser for streams (see~\Cref{fig:typestream}) and $\erasetype{\Bool^{k+1}}$ is the eraser for $\Bool^{k+1}$ given by~\Cref{prop:linear-erasure}.

We can now prove our fundamental theorem:

	\begin{theorem}\label{thm:completeness} 
	Let $f: (\{\false,\true\}^*)^n \to \{\false,\true\}^*$:
	\begin{enumerate}
		\item \label{enum:completeness1} If  $f\in \fppoly$ then $f$ is  representable in $\typestream$;
		\item \label{enum:completeness2} If  $f\in \fptime$ then $f$ is representable in $\pta$.
	\end{enumerate}
\end{theorem}
\begin{proof}
	We only show the case where $f$ is a unary function for the sake of simplicity. 
	Let us prove~\Cref{enum:completeness1}. If $f \in \fppoly$ then,   $f \in \fptime(\RR)$ by~\Cref{prop:fppoly=fptime(RR)}, so  there is a polynomial Turing machine computing $f$ that performs  polynomially many queries to bits of a Boolean stream $\alpha$.  Let $p(x)$ and $q(x)$ be polynomials bounding, respectively, the  time and space of the Turing machine, and let  $\degree p=m$ and $\degree q = l$ be their degrees. By~\Cref{thm:polynomial-completeness} we obtain $\cod{p}$  and $\cod{q}$ typable as:
	\[
	\def\arraystretch{1.2}
	\arraycolsep=2pt
	\begin{array}{l}
		y:\oc^{m-1} \Nat[m+1] \vdash \cod{p}: \Nat[1]\\
		z:\oc^{l-1} \Nat[l+1] \vdash \cod{q}: \Nat[1]
	\end{array}
	\]
	where $\Nat[i]$ is shorthand notation for $\Nat_{\tm}[i]$ (see~\Cref{defn:tiering}). By applying~\Cref{lem:substitution} and~\Cref{prop:length}, we have:
	\[
	\def\arraystretch{1.2}
	\arraycolsep=2pt
	\begin{array}{l}
		s':\oc^{m-1} \String[m+1] \vdash P: \Nat[1]\\
		s'':\oc^{l-1} \String[l+1] \vdash Q: \Nat[1]
	\end{array}
	\]
	where $P\dfn  \cod{p}[\length\, s'/y]:$ and $Q \dfn \cod{q}[\length \, s''/z]$. On the other hand, by applying again~\Cref{lem:substitution} and~\Cref{eqn:init,eqn:initialisation,eqn:transition,eqn:extr}:
	\[
	t: \String[1], p:\Nat[1], q: \Nat[1] \vdash \extr ((p \, \tr )(\In\, t\, (\init \, q\,  \alpha) )): \String
	\]
	By putting everything together we have:
	\[
	s':\oc^{m-1} \String[m+1],	s'':\oc^{l-1} \String[l+1] ,  t: \String[1]\vdash N: \String
	\]
	where  $N\dfn \extr ((P \, \tr )(\In\, t\, (\init \, Q \,  \alpha) ))$. 
	By repeatedly applying $\wnbrule$ and  $\dows{k}$ for appropriate $k$ we obtain a term $M$ representing $f$ such that:
	\[
	s:\oc^{\max (m, l)} \String[\max(m, l)+1] \vdash  M: \String
	\]
	By applying $\dows{1}$ we obtain
	\[
	x:\String[\oc^{\max (m, l)} \String[\max(m, l)+1]] \vdash  M[\dows{1}\, x/s]: \String
	\]
	We set $\cod{f}\dfn M[\dows{1}\, x/s] $, so that $x: \String[] \vdash \cod{f}: \String$.
	
	\Cref{enum:completeness2} follows directly from~\Cref{enum:completeness1}  by stripping away streams from the above encoding.
\end{proof}

\subsection{Translations and  completeness theorem}

We can compare the computational strength of  type systems and  inductive proof systems based on parsimonious linear logic by means of a translation.

	\begin{definition}[Translation]
	We define a translation $(\_)^\dagger$ from  $\typestream$ to $\dpll$ mapping typing derivations of $\typestream$ to derivations of $\dpll$ such that, when restricted to typing derivations of $\pta$, it returns  derivations of $\pll$:
	\begin{itemize}
		\item It maps types of $\typestream$ to formulas of $\dpll$  according to the following inductive definition:
		\[
		\begin{array}{rcl}
			X^\dagger &\dfn &X \\
			\unit^\dagger &\dfn &\unit \\ 
			(\sigma \limp A)^\dagger &\dfn &\sigma^\dagger \limp A^\dagger \\
			(	\forall X.A)^\dagger &\dfn &	\forall X.A^\dagger \\ 
			(	\sigma \otimes \tau)^\dagger &\dfn &	\sigma^\dagger \otimes \tau^\dagger \\   
			(\oc	\sigma )^\dagger &\dfn &\oc 	\sigma^\dagger  \\   
			(\streamtype \sigma )^\dagger &\dfn &\oc{	\sigma^\dagger  }
		\end{array}
		\]
		we notice that $\sigma^\dagger[\tau^\dagger/X]= (\sigma[\tau/X])^\dagger$.
		\item It maps a context  $\Gamma=  x_1: \sigma_1, \ldots, x_n: \sigma_n$ to a sequent $\Gamma^\dagger = \sigma_1^\dagger, \ldots, \sigma_n^\dagger$. 
		\item It maps judgements $\Gamma \vdash M: \tau$ to sequents $\cneg{\Gamma^\dagger},\tau^\dagger$.
		\item It maps a typing rule to gadgets as in~\Cref{fig:translation} and~\Cref{fig:translation1}. 
	\end{itemize}
\end{definition}

\begin{figure*}[!t]
	\centering
	\[
	\small
			\vlinf{\axr}{}{x: A \vdash x: A}{}
			\ \ 
			\mapsto 
			\ \ 
			\vlinf{\axr}{}{\cneg{A^\dagger}, A}{} 
			\qquad  \quad
			\vlinf{\ir \identity}{}{\vdash \identity: \unit}{}
			\ \
			\mapsto 
			\ \
			\vlinf{\unit}{}{\unit}{}
		\]
		\medskip
		\[
				\vliinf{\er\identity}{}{\Gamma, \Delta \vdash	\letid{N}{M} : \sigma}{\Gamma \vdash N: \unit}{\Delta \vdash M: \sigma}
			\ \
			\mapsto 
			\ \
			\vlderivation{
				\vliin{\cutr}{}{\cneg{\Gamma^\dagger}, \cneg{\Delta^\dagger}, \sigma^\dagger}
				{
					\vlhy{\cneg{\Gamma^\dagger}, \unit^\dagger}
				}
				{
					\vlin{\bot}{}{\bot, \cneg{\Delta^\dagger}, \sigma^\dagger}
					{
						\vlhy{\cneg{\Delta^\dagger}, \sigma^\dagger}
					}
				}
			}
		\]\medskip
		\[
			\vlinf{\ir \limp }{}{\Gamma \vdash \lambda x. M: \sigma  \limp B}{\Gamma , x:\sigma \vdash M:B}
			\ \
			\mapsto 
			\ \
			\vlinf{\parr}{}{\cneg{\Gamma^\dagger}, (\sigma\limp B)^\dagger}{\cneg{\Gamma^\dagger}, \cneg{\sigma^\dagger}, B^\dagger}
\]
\[
\small
			\vliinf{\er \limp}{}{\Gamma, \Delta \vdash MN : B}{\Gamma \vdash M: \sigma \limp B}{\Delta \vdash N: \sigma }
			\ \ 
			\mapsto 
			\ \  
			\vlderivation{
				\vliin{\cutr}{}{\cneg{\Gamma^\dagger},\cneg{\Delta^\dagger}, B^\dagger}
				{
					\vlhy{\cneg{\Gamma^\dagger}, (\sigma \limp B)^\dagger}
				}
				{
					\vliin{\otimes}{}{\cneg{\Delta^\dagger}, A^\dagger \otimes\cneg{B^\dagger}, B^\dagger}
					{
						\vlhy{\cneg{\Delta^\dagger}, A^\dagger}
					}
					{
						\vlin{\axr}{}{\cneg{B^\dagger}, B^\dagger}{\vlhy{}}
					}
				}
			}
\]
\medskip
\[
		\small
			\vliinf{\ir \otimes}{}{\Gamma, \Delta \vdash M \otimes N : \sigma  \otimes \tau}{\Gamma \vdash M: \sigma}{\Delta \vdash  N: \tau}
			\ \  
			\mapsto 
			\ \ 
			\vliinf{\otimes}{}{\cneg{\Gamma^\dagger},\cneg{\Delta^\dagger}, (\sigma \otimes \tau)^\dagger}
			{
				\cneg{\Gamma^\dagger}, \sigma^\dagger
			}{
				\cneg{\Delta^\dagger}, \tau^\dagger
			}
\]
\medskip
\[
\small
			\vliinf{\er \otimes }{}{\Gamma, \Delta   \vdash \lettensor{M \otimes N}{x}{y}{P}: C}{\Gamma \vdash M \otimes N : \sigma \otimes \tau}{\Delta, x: \sigma, y: \tau \vdash P: C}
			\ \  
			\mapsto 
			\ \ 
			\vlderivation{
				\vliin{\cutr}{}{\cneg{\Gamma^\dagger}, \cneg{\Delta^\dagger},C^\dagger }
				{
					\vlhy{\cneg{\Gamma^\dagger}, (\sigma \otimes \tau)^\dagger}
				}
				{
					\vlin{\parr}{}{\cneg{\Delta^\dagger}, \cneg{\sigma^\dagger}\parr\cneg{\tau^\dagger}, C^\dagger}
					{
						\vlhy{\cneg{\Delta^\dagger}, \cneg{\sigma^\dagger},\cneg{\tau^\dagger}, C^\dagger }
					}
				}
			}
		\]
		\medskip
		\[
		\small
			\vlinf{\ir \forall}{}{\Gamma \vdash M: \forall X. A}{\Gamma \vdash M:A} 
			\ \  
			\mapsto 
			\ \  
			\vlinf{\forall}{}{\cneg{\Gamma^\dagger}, (\forall X. A)^\dagger}{\cneg{\Gamma^\dagger}, A^\dagger} 
\qquad 
			\vlinf{\er \forall}{}{\Gamma \vdash M: A[B/X]}{\Gamma\vdash M:  \forall X. A}
			\ \  
			\mapsto 
			\ \  
			\vlderivation{
				\vliin{\cutr}{}{\cneg{\Gamma^\dagger},(A[B/X])^\dagger }
				{
					\vlhy{\cneg{\Gamma^\dagger}, (\forall X. A)^\dagger}
				}
				{
					\vlin{\exists}{}{ \exists X. \cneg{ A^\dagger}, A^\dagger[B^\dagger/X]}{
						\vlin{\axr}{}{\cneg{A^\dagger}[{B^\dagger}/X], A^\dagger[B^\dagger/X]}{\vlhy{}}
					}
				}
			}
\]\medskip
\[
\small
			\vlinf{\fprule}{}{\oc \Gamma \vdash   M: \oc  \sigma }{ \Gamma \vdash M : \sigma }
			\ \  
			\mapsto 
			\ \  
			\vlinf{\fprule}{}{\oc \cneg{ \Gamma^\dagger}, (\oc  \sigma )^\dagger}{ \cneg{\Gamma^\dagger}, \sigma^\dagger }
\qquad 
			\vlinf{\wnwrule}{}{\Gamma, x: \oc \sigma    \vdash   M :  \tau}{\Gamma \vdash M: \tau}
			\ \  
			\mapsto 
			\ \  
			\vlinf{\wnwrule}{}{\cneg{\Gamma^\dagger},  {\cneg{(\oc \sigma)^\dagger}} ,   \tau^\dagger}{{\cneg{\Gamma^\dagger}}, \tau^\dagger}
			\]
			\medskip
			\[
			\small
			\vlinf{\wnbrule}{}{\Gamma, x: \oc \sigma    \vdash M[x/y, x/z] :  \tau}{\Gamma, y:  \sigma,  z: \oc\sigma \vdash M: \tau}
			\ \  
			\mapsto 
			\ \  
			\vlinf{\wnbrule}{}{\cneg{\Gamma^\dagger},    {\cneg{(\oc\sigma)^\dagger}} , \tau^\dagger}{\cneg{\Gamma^\dagger},  \cneg{\sigma^\dagger}, \cneg{(\oc\sigma)^\dagger} , \tau^\dagger}
\]
	\caption{Translation from $\pta$ to $\pll$.}
	\label{fig:translation}
\end{figure*}
\begin{figure*}[t]
	\centering
	\small
	\adjustbox{max width= \textwidth}{$
		\begin{array}{c}
			\vliiiiinf{\stream}{}{\vdash \streamterm{M}: \streamtype \sigma}{\vdash \streamterm M(0): \sigma}{\vdash \streamterm M(1): \sigma}{\ldots}{\vdash \streamterm M(n):  \sigma}{\ldots}
			\quad 
			\mapsto 
			\quad 
			\vliiiiinf{\nuprule}{}{ (\streamtype \sigma)^\dagger}{\sigma^\dagger}{ \sigma^\dagger}{\ldots}{\sigma^\dagger}{\ldots}
		\end{array}$}
	
	\vspace{0.2cm}
	\adjustbox{max width= \textwidth}{$
		\begin{array}{c}
			\vlinf{\erasebool}{}{\vdash \erasebool: \streamtype \sigma \limp \unit}{}
			\quad 
			\mapsto 
			\quad 
			\vlderivation{
				\vlin{\parr}{}{(\streamtype \sigma \limp \unit)^\dagger}{
					\vlin{\wnwrule}{}{  \cneg{(\streamtype \sigma)^\dagger}, \unit^\dagger}
					{
						\vlin{\unit}{}{\unit^\dagger}{\vlhy{}}
					}
				}
			}
			\qquad \qquad \qquad 
			\vlinf{\popbool}{}{\vdash \popbool: \streamtype \sigma   \limp \sigma  \otimes \streamtype \sigma }{}
			\quad 
			\mapsto 
			\quad 
			\vlderivation{
				\vlin{\parr}{}{(\streamtype \sigma   \limp \sigma  \otimes \streamtype \sigma)^\dagger}{
					\vlin{\wnbrule}{}{ \cneg{(\streamtype\sigma)^\dagger},  (\sigma  \otimes \streamtype \sigma)^\dagger}
					{
						\vliin{\otimes}{}{\cneg{\sigma^\dagger},\cneg{(\streamtype\sigma)^\dagger}, \sigma^\dagger \otimes  \oc \sigma^\dagger }
						{
							\vlin{\axr}{}{\cneg{\sigma^\dagger}, \sigma^\dagger}{\vlhy{}}
						}
						{
							\vlin{\fprule}{}{\cneg{(\streamtype\sigma)^\dagger}, \oc \sigma^\dagger}{\vlin{\axr}{}{ \cneg{\sigma^\dagger},  \sigma^\dagger}{\vlhy{}}}
						}
					}
				}
			}
		\end{array}$}
	\caption{Translation from $\typestream$ to $\dpll$.}
	\label{fig:translation1}
\end{figure*}

The  two lemmas below represent stronger versions of \Cref{lem:substitution} and~\Cref{prop:subject-reduction}, respectively.

\begin{lemma}\label{lem:stronger-substitution}
	For any $\der_1: \Gamma \vdash M: \sigma $ and $\der_2: \Delta \vdash N: \tau$ there is $S(\der_1, \der_2)$ such that:
	\[
        \small
		\vlderivation{
			\vliin{\cutr}{}{\cneg{\Gamma^\dagger}, \cneg{\Delta^\dagger}, \sigma^\dagger}
			{
				\vlhy{ \left(
					\vlderivation{
						\vldr{\der_1}{\Delta \vdash N: \tau}
					}
					\right)^\dagger
				}
			}
			{
				\vlhy{  \left(
					\vlderivation{
						\vldr{\der_2}{\Gamma, x: \tau  \vdash M: \sigma }
					}
					\right)^\dagger
				}
			}
		}
		\quad  \cutelims \quad 
		\left(
		\vlderivation{
			\vltr{S(\der_1, \der_2) }{\Gamma,\Delta  \vdash M[N/x]: \sigma}{\vlhy{\ \ \ \ \ \  }}{\vlhy{\ \ \  \ \  \ \ \  }}{\vlhy{\ \ \ \ \  \ }}
		}
		\right)^\dagger
\]
\end{lemma}
\begin{proof}
	It suffices to check that the  derivation $S(\der_1, \der_2)$ can be stepwise computed by the cut elimination rules.
\end{proof}

\begin{lemma}\label{lem:stronger-subject-reduction}
	Let $\der_1:\Gamma \vdash M_1: \sigma$. 	If	$M_1 \betar M_2$ then there is a typing derivation  $\der_2:\Gamma \vdash M_2: \sigma$ such that $\der_1^\dagger \cutelims \der_2^\dagger$.
\end{lemma}
\begin{proof}
	It suffices to check the statement for the reduction rules in~\Cref{defn:terms}, by inspecting the cut elimination rules of $\dpll$. We consider the two most relevant  cases. If $M_1= 	\popbool \, \streamterm M $ and  $M_2= \hd{\streamterm{M}} \otimes \tail{\streamterm M}$, then w.l.o.g.~$\der_1$ has the following shape:
	\[
    \small
		\vlderivation{
			\vliin{\er\limp}{}{
				\vdash \popbool\, \streamterm{M}: \streamtype \sigma \otimes\streamtype{\sigma}
			}{
				\vlin{\popbool}{}{\vdash \popbool: \streamtype{\sigma} \limp \sigma \otimes\streamtype{\sigma}}{\vlhy{}}
			}{
				\vliiiiin{\stream}{}{\vdash \mathbf{M}: \streamtype \sigma}{\vldr{}{\vdash \streamterm M(0): \sigma}}{\vldr{}{\vdash \streamterm M(1): \sigma}}{\vlhy{\ldots}}{\vldr{}{\vdash \streamterm M(n): \sigma}}{\vlhy{\ldots}}
			}
		}
	\]
	We set $\der_2$ as the following typing derivation:
\[
\small
		\vlderivation{
			\vliin{\otimes}{}{
				\vdash \streamterm{M}(0) \otimes tl(\streamterm{M}):\sigma \otimes\streamtype{\sigma}
			}{
				\vlin{}{}{\vdash \streamterm M(0): \sigma}{\vlhy{\vdots}}
			}{
				\vliiiiin{\stream}{}{\vdash tl(\streamterm{M}): \streamtype \sigma}{\vldr{}{\vdash \streamterm M(1): \sigma}}{\vldr{}{\vdash \streamterm M(2): \sigma}}{\vlhy{\ldots}}{\vldr{}{\vdash \streamterm M(n+1): \sigma}}{\vlhy{\ldots}}
			}
		}
	\]
	It is easy to check that $\der_1^\dagger \cutelims \der_2^\dagger$. 
	
	Let $M_1= (\lambda x. P)N$ and $M_2= P[N/x]$. By  inspecting the typing rules in~\Cref{fig:typestream}  $\der$ must have the following structure:
	{
		\vspace{-0.2cm}
		$$
            \small
		\vlderivation{
			\vlin{}{}{\Sigma, \Delta \vdash (\lambda x. P)Q: \sigma}
			{
				\vlin{}{}{\delta\ \vdots}
				{
					\vliin{\er{\limp}}{}{\Sigma', \Delta' \vdash (\lambda x. P')Q':  B'}
					{
						\vltr{\der_1}{\Sigma' \vdash \lambda x. P': \tau \limp B'}{\vlhy{\ }}{\vlhy{\ }}{\vlhy{\ }}            
					}{
						\vltr{\der_2}{\Delta' \vdash Q' : \tau}{\vlhy{\ }}{\vlhy{\ }}{\vlhy{\ }}
					}
				}
			}	
		}
		$$
	}
	where:
	\begin{itemize}
		\item  $\Gamma= \Sigma, \Delta$ and $\sigma= {{\oc \overset{n}{\ldots} \oc}} \forall \vec X. B$, for some $n\geq 0$,  $\vec X$, and $B$.
		\item $\delta$ is a sequence of rules in $\{ \wnwrule, \wnbrule, \fprule, \ir \forall, \er \forall\}$, 
		\item $B= B'[\vec C/\vec Y]$, for some $\vec C$ and $\vec Y$ not free in $\Sigma', \Delta'$
		\item $P'[\vec x/\vec y]=P$ and $Q'[\vec x/\vec y]=Q$, for some $\vec x, \vec y$. 
	\end{itemize}
	By a similar reasoning, $\der_1$ has the following shape:
	$$
        \small
	\vlderivation{
		\vlin{}{}{\Sigma' \vdash \lambda x. P': \tau \limp B'}
		{
			\vlin{}{}{\varepsilon\ \vdots}
			{
				\vlin{\ir \limp}{}{\Sigma''\vdash \lambda x. P'': \tau' \limp B''}
				{
					\vltr{\der'_1}{\Sigma'', x:\tau' \vdash P'':  B''}{\vlhy{\ }}{\vlhy{\ \ \ }}{\vlhy{\ }}
				}
			}
		}  
	}
	$$
	where:
	\begin{itemize}
		\item $\varepsilon$ is sequences of typing rules in $\{ \wnwrule, \wnbrule, \ir \forall, \er \forall\}$,  
		\item $B'= B''[\vec D/\vec Z]$ and $\tau= \tau'[\vec D/\vec Z]$, for some $\vec D$ and $\vec Z$ not free in $\Sigma''$
		\item $P''[\vec z/\vec w]=P$ for some $\vec z, \vec w$. 
	\end{itemize}
	Since $\tau= \tau'[\vec D/\vec Z]$, $B'= B''[\vec D/\vec Z]$ and $\vec Z$ do not occur free in $\Sigma''$,  by~\Cref{lem:substitution-types}  there is a typing derivation  $\der''_1$ of $\Sigma'', x: \tau \vdash P'': B'$. By~\Cref{lem:substitution} there is a typing derivation $S(\der''_1, \der_2)$ of $\Delta', \Sigma'' \vdash P''[Q'/x]: B'$. Finally, by applying the sequences of rules $\delta$ and $\varepsilon$ we obtain:
	{
		\vspace{-0.2cm}
		$$
		\begin{array}{rcl}
            \small
			\widehat{\der} &\dfn& 	\vlderivation{
				\vlin{}{}{\Sigma, \Delta \vdash P[Q/x]: \sigma}
				{
					\vlin{}{ }{  \vdots} 
					{
						\vlin{}{}{\Delta', \Sigma' \vdash P'[Q'/x]: B'}
						{
							\vlin{}{}{  \vdots}
							{
								\toks0={0.5}
								\vltrf{S(\der''_1, \der_2)}{\Delta', \Sigma'' \vdash P''[Q'/x]: B'}{\vlhy{\ \  \ 	\  \ \  \  \  }}{\vlhy{\ \ \ \  \ \ \ \ \  \ \ 	 }}{\vlhy{\ \ \  \  \ \  \ 	 }}{\the\toks0}
							}
						}
					}
				}
			}
		\end{array}
		$$
	}
	
	Let us now show that $\der^\dagger\cutelims \widehat{\der}^\dagger$. First, notice that $\der^\dagger$ is as follows:
	$$
        \small
	\vlderivation{
		\vlin{}{}{\cneg{\Sigma^\dagger}, \cneg{\Delta^\dagger},\sigma^\dagger}
		{
			\vlin{}{}{\delta^\dagger \, \vdots}
			{
				\vliin{\cutr}{}{\cneg{(\Sigma')^\dagger}, \cneg{(\Delta')^\dagger},(B')^\dagger }
				{
					\vlin{}{}{\cneg{(\Sigma')^\dagger}, \cneg{\tau^\dagger}\parr (B')^\dagger}{
						\vlin{}{}{\varepsilon^\dagger \ \vdots}{
							\vlin{\parr}{}{\cneg{(\Sigma'')^\dagger}, \cneg{(\tau')^\dagger}\parr (B'')^\dagger}{
								\vltr{(\der'_1)^\dagger}{\cneg{(\Sigma'')^\dagger}, \cneg{(\tau')^\dagger}, (B'')^\dagger}{\vlhy{\ \  }}{\vlhy{\ \ \ \ \ \   }}{\vlhy{\ \  }}
							}
						}
					}
				}
				{
					\vliin{\otimes}{}{\cneg{(\Delta')^\dagger},  \tau^\dagger \otimes \cneg{(B')^\dagger}, (B')^\dagger}
					{
						\vltr{\der_2^\dagger}{\cneg{(\Delta')^\dagger}, {\tau^\dagger}, }{\vlhy{\ \  }}{\vlhy{\ \ \ }}{\vlhy{\ \  }}
					}
					{
						\vlin{\axr}{}{\cneg{(B')^\dagger}, (B')^\dagger}{\vlhy{}}
					}
				}
			}
		}
	}
	$$
	Moreover, 	since $\tau= \tau'[\vec D/\vec Z]$, $B'= B''[\vec D/\vec Z]$ and $\vec Z$ do not occur free in $\Sigma''$, the above derivation reduces by cut elimination to the following:
	$$
        \small
	\vlderivation{
		\vlin{}{}{\cneg{\Delta^\dagger}, \cneg{\Sigma^\dagger},\sigma ^\dagger}
		{
			\vlin{}{}{\vdots}
			{
				\vlin{}{}{\cneg{(\Delta')^\dagger}, \cneg{(\Sigma')^\dagger},(B')^\dagger} 
				{
					\vlin{}{}{\vdots}
					{
						\vliin{\cutr}{}{\cneg{(\Delta')^\dagger}, \cneg{(\Sigma'')^\dagger},(B')^\dagger}
						{
							\vltr{\der_2^\dagger}{\cneg{(\Delta')^\dagger}, {\tau^\dagger}, }{\vlhy{\ \  }}{\vlhy{\ \  \ }}{\vlhy{\ \  }}
						}
						{
							\vltr{(\der''_1)^\dagger}{\cneg{(\Sigma'')^\dagger}, \cneg{\tau^\dagger}, (B')^\dagger}{\vlhy{\ \ \ \  }}{\vlhy{\ \ \ \   }}{\vlhy{\ \  }}
						}
					}
				}
			}
		}
	}
	$$
	for some $\varepsilon^\dagger$ and $\delta^\dagger$. By~\Cref{lem:stronger-substitution}  the above derivation reduces by cut elimination to the following:
	$$
        \small
	\vlderivation{	\vlin{}{}{\cneg{\Delta^\dagger}, \cneg{\Sigma^\dagger},\sigma ^\dagger}
		{
			\vlin{}{}{\vdots}
			{
				\vlin{}{}{\cneg{(\Delta')^\dagger}, \cneg{(\Sigma')^\dagger},(B')^\dagger} 
				{
					\vlin{}{}{\vdots}
					{
						\vlhy{
							\left(
							\vlderivation{
								\vltr{S(\der''_1, \der_2) }{\Delta', \Sigma'' \vdash P''[Q'/x]: B'}{\vlhy{\ \ \ \ \ \  }}{\vlhy{\ \ \  \ \  \  \ \ \ \ }}{\vlhy{\ \ \ \ \  \ }}
							}
							\right)^\dagger}
	}}}}}
	$$
	which is $\widehat{\der}^\dagger$.
\end{proof}

\begin{theorem}\label{thm:embedding}  
	Let $f: (\{\false,\true\}^*)^n\to \{\false,\true\}^*$:
	\begin{enumerate}
		\item \label{enum:embedding1} If  $f$ is representable in $\pta$  then it is in  $\pll$;
			\item \label{enum:embedding2} If  $f$ is representable in $\typestream$  then so it is in  $\dpll$.
	\end{enumerate}
\end{theorem}
\begin{proof}
	We  only consider the case where $f$ is unary for the sake of simplicity. Let $\cod{f}$ be a typable term  of $\typestream$ representing $f$,  so that $\cod{f}\,\cod{s}\betar^* \cod{f(s)}$  for any  $s\in \{\false,\true\}^*$. Consider the following derivation:
	$$
	\der=	\vlderivation{
		\vliin{\er \limp}{}{\vdash \cod{f}\, \cod{s}:\String[]}{\vltr{\der_f}{\vdash \cod{f}:\String[]\limp \String[]}{\vlhy{\ \ }}{\vlhy{\ \ \ }}{\vlhy{\  \ }}
		}{\vltr{\der_s}{\vdash \cod{s}: \String[]}{\vlhy{\ \ }}{\vlhy{\ \ \ }}{\vlhy{\  \ }}
		}
	}
	$$
	By repeatedly applying~\Cref{lem:stronger-subject-reduction} there is $\der_{f(s)}$ such that
\[
		\der^\dagger	= \vlderivation{
			\vliin{\cutr}{}{\String[]}{\vltr{\der_f^\dagger}{\String[]\limp \String[]}{\vlhy{\ \ }}{\vlhy{\ \ \ }}{\vlhy{\  \ }}
			}
			{
				\vliin{\otimes}{}{\String[]\otimes \cneg{\String[]}, \String[]}	{\vltr{\der_s^\dagger }{ \String[]}{\vlhy{\ \ }}{\vlhy{\ \  }}{\vlhy{\  \ }}
				}{\vlin{\axr}{}{\cneg{\String[]},\String[]}{\vlhy{}}}
			}
		}
		\quad \cutelims \quad 
		\vlderivation{
			\vltr{\der_{f(s)}^\dagger}{\String[]}{\vlhy{\ \ }}{\vlhy{\ \ \ \ \ \ \  }}{\vlhy{\  \ }}
		}
\]
	in $\dpll$, where we  can safely assume that  $\der_s^\dagger\cutelims \cod{s}$ and $\der_{f(s)}^\dagger\cutelims \cod{f(s)}$  in $\dpll$.  This means that $\der_f^\dagger$ represents $f$ in $\dpll$. If moreover $\cod{f}$ is typable term of $\pta$ then $\der_f^\dagger$ represents $f$ in $\pll$.
\end{proof}

\begin{theorem}[Completeness]\label{thm:completeness-non-wellfounded}
	Let $f: (\{\false,\true\}^*)^n\to \{\false,\true\}^*$:
	\begin{enumerate}
		\item \label{enum:comp1} If  $f\in \fppoly$  then $f$ is representable in   $\nupll$;
		\item \label{enum:comp2} If  $f\in \fptime$ then $f$  is representable  in  $\cpll$.
	\end{enumerate}
\end{theorem}
\begin{proof}
	For $i \in \set{1,2}$, Item (i) follows from~\Cref{thm:completeness}.(i),  \Cref{thm:embedding}.(i) and~\Cref{thm:simulation}.(i).
\end{proof}

We conclude  the section by  discussing  some computational aspects of  the finiteness condition  on the typing rule $\nuprule$, and the restriction on second-order    instantiation to $(\oc, \omega)$-free types  in $\typestream$.

\begin{remark}\label{rem:typ-1-completeness}
	If the side condition on the typing rule $\stream$ (i.e.,  that 	$\set{\streamterm{M}(i)\mid i\in\Nset}$ is finite) were dropped, then  $\typestream$ would  represent any function on natural numbers. Indeed, given a function $f: \Nset \to \Nset$, we can define the term $\streamterm F \dfn \cod{f(0)}:: \cod{f(1)}::\ldots$ with type $\streamtype{\oc \Nat}$,  encoding all the values of the function $f$. We set $A\dfn  \Nat[\unit] \otimes \streamtype  \Nat[\unit]$ and define:
	\[
	\def\arraystretch{1.2}
	\arraycolsep=1.5pt
	\begin{array}{rcll}
		\mathsf{step} &\dfn & \lambda x. \lettensor{x}{y_1}{y_2}{
			\letid{y_1 \,  (\lambda z.z) \, \identity}{ \popbool \, y_2}}\\
		\cod{f}&\dfn &  \lambda n. \mathsf{let \, }x \otimes y = \!(n \,   \mathsf{step}\,  ( \popbool \ \streamterm F)) \mathsf{\, in\, } \letid{\! \!(\erasebool \,  y)}{x} 
	\end{array}
	\]
	where $	\mathsf{step}$ has type $A \limp A$ and $	\cod{f}$ has type $\Nat[A]  \limp \Nat[\unit]$.	It is easy to check that $\cod{f}\, \cod{n}\betar^* \cod{f(n)}$, for any $n \in \Nset$.
	
	This observation can be easily adapted to the proof systems  $\dpll$ (w.r.t.~the finiteness condition on $\nuprule$) and $\nupll$  (w.r.t.~ the weak regularity condition). 
\end{remark}

\begin{remark}\label{rem:exponential-blow-up}
	If the $(\oc, \omega)$-freeness condition on $\er \forall$ were dropped then  $\typestream$ could  represent exponential functions. Indeed, we can define the following functions:
	\[
	\def\arraystretch{1.2}
	\arraycolsep=2pt
	\begin{array}{rcll}
		\mathsf{plustwo}&\dfn & \lambda n. \lambda   f. \lambda z. n f(f(fz))&: \Nat[] \limp \Nat[]\\
		\mathsf{double}&\dfn& \lambda n. n \, ( \mathsf{plustwo}) \, \cod{0}&: \Nat[\Nat[]]\limp \Nat[]\\
		\mathsf{exp}&\dfn & \lambda n. n \ ( \mathsf{double})\ \cod{1}&: \Nat[\Nat[]] \limp \Nat[]
	\end{array}
	\]
	It is easy to check that, for any $n \in \Nset$:
	\[
	\mathsf{plustwo} \ \cod{n}\betar^* \cod{n+2} \qquad 	\mathsf{double}  \ \cod{n} \betar^*\cod{2n}\qquad 	 			\mathsf{exp}  \ \cod{n} \betar^* \cod{2^n} 
	\]
		This observation can be easily adapted to the proof systems  $\dpll$  and $\nupll$  (w.r.t. the $\oc$-freeness condition of instantiations in the rule $\exists$). 
\end{remark}

\section{Conclusion and future work}

This paper builds on a series of recent works aimed at developing {implicit computational complexity} in  the setting of cyclic and non-wellfounded proof theory~\cite{CurziDas,Curzi023}. 	We proved that the non-wellfounded proof systems $\nupll$ and $\cpll$ 
capture the complexity classes $\fppoly$ and $\fptime$ respectively.
We then establish a series of characterisations for various finitary proof~systems.

We envisage extending the contribution of this paper, among others, to the following research directions.

\paragraph*{Polynomial time over the reals} \cite{HainryMP20} introduces a  characterisation of Ko's class of polynomial time computable functions over real numbers~\cite{Ko} based on parsimonious logic.
By employing the  co-absorption rule $\ocbrule$ to represent the \emph{pop} operation on streams, this complexity class could be modelled {within $\nwpll$} via cut elimination as in~\cite{CSL}.
 
\paragraph*{Probabilistic complexity}
De-randomisation methods showing the inclusion of the complexity class $\mathbf{BPP}$ (bounded-error probabilistic polynomial time) in $\fppoly$	suggest that this class can be characterised within  $\nupll$.  Challenges are expected, since $\mathbf{BPP}$ is defined by explicit (error) bounds, as observed in~\cite{LagoT15} (so, not entirely in the style of ICC), but we conjecture that error bounds can be traded for appropriate global proof-theoretic conditions on $\nupll$ that restrict computationally the access to streams.

\paragraph*{Logarithmic Space}
In~\cite{MazzaT15,Mazza15} the authors characterize the complexity classes $\logspace$ (logarithmic space problems) and  its non-uniform counterpart $\lpoly$ (problems decided by polynomial size branching programs) by stripping away  second-order quantifiers from their proof systems capturing  $\ptime$ and $\ppoly$.
We expect that a similar result can be obtained for our non-wellfounded proof  systems.

\paragraph*{Non-uniform Proofs-as-Processes}
Processes such as a scheduler sorting tasks among a (finite) set of servers according to a predetermined order (e.g., a token ring of servers) may easily be modelled by \nwboxes,
making $\nupll$ appealing for the study of the proofs-as-processes correspondence and its applications~\cite{ABRAMSKY19945,caires2010session,wadler2014propositions,DARDHA2017253,montesi2021linear}.

\section*{Acknowledgement}
We would like to thank Anupam Das, Abhishek De, Farzad Jafar-Rahmani, Alexis Saurin, Tito (L\^{e} Thành Dung Nguy\^{e}n) and Damiano Mazza for their useful comments and suggestions. This work was supported by the Wallenberg Academy Fellowship Prolongation project “Tam-
ing Jörmungandr: The Logical Foundations of Circularity” (project reference 251080003), and by the
VR starting grant “Proofs with Cycles in Computation” (project reference 251088801).

\bibliographystyle{alphaurl}
\bibliography{biblo}

\end{document}

%% file: macros.tex
\def\cutstep#1#2{#1\mbox{-vs-}#2}

\newcommand{\weight}[1]{\mathsf{W}(#1)}
\newcommand{\heightcut}[1]{\mathsf{H}(#1)}

\newcommand{\rank}[2]{\mathsf{rk}_{#1}({#2})}

\newcommand{\egraph}[1]{\mathcal{G}(#1)}

\newcommand{\isize}[1]{||{#1}||}

\newcommand{\prun}[2]{\llfloor{#1} \rrfloor_{#2}}
\newcommand{\base}[1]{\mathsf{base}(#1)}
\newcommand{\trunk}[2]{\lfloor{#1} \rfloor_{#2}}

\def\Sys{\mathcal S}

\newcommand{\finitenwpromotion}{\mathfrak{F}}

\newcommand{\derfinstream}[3]{{\cprule_{\langle #1, #2,\ldots , #3\rangle}}}
\newcommand{\derfinstreamshort}[3]{{\cprule_{\langle #1,\ldots , #3\rangle}}}

\newcommand{\dfn}{\coloneqq}

\def\defin#1{\textbf{#1}}

\def\set#1{\{#1\}}

\def\intset#1#2{\set{#1,\ldots, #2}}

\definecolor{mygreen}{rgb}{0, 0.6, 0}

\definecolor{airforceblue}{rgb}{0.36, 0.54, 0.66}

\def\cneg#1{{#1}^\lbot}

\def\atoms{\mathcal A}
\newcommand{\Bool}{\mathbf{B}}
\newcommand{\String}{\mathbf{S}}
\newcommand{\Nat}{\mathbf{N}}
\newcommand{\Stream}{\mathbf{Stream}}

\def\axr{\mathsf{ax}}
\def\cutr{\mathsf{cut}}

\def\wrule{\mathsf{w}}

\def\brule{\mathsf{b}}

\def\wnbrule{\wn\brule}
\def\ocbrule{\oc\brule}
\def\wnwrule{\wn \wrule}

\def\ocwrule{\oc \wrule}

\def\prule{\mathsf{\oc p}}

\def\fprule{\mathsf{f\prule}}
\def\cprule{\mathsf{c\prule}}
\def\rrule{\mathsf{r}}

\def\unit{\mathbf{1}}

\newcommand{\zero}{\mathsf{hyp}}

\def\mathacronym#1{\textsf{#1}}

\newcommand{\nuprule}{\mathsf{ib}\prule} 
\newcommand{\nwprule}{\mathacronym{nwb}\xspace} 
\newcommand{\nwpromotion}{\mathfrak{S}}
\newcommand{\nwpstream}[3]{\seq{ {#1},{#2}, \ldots,\allowbreak {#3}, \ldots } }

\newcommand{\derstream}[1]{\cprule_{\seq{ #1, \ldots}}}%

\def\supportof#1{\mathsf{Calls}(#1)}

\newcommand{\der}{\mathcal{D}}
\newcommand{\dD}{\der}

\newcommand{\zeroder}{\der_{\lbot}}
\newcommand{\wnder}{\der_{\wn}}
\def\trueder{\cod\true}
\def\falseder{\cod\false}

\newcommand{\vlhyE}[1]{\global\setbox\vlhybox=\hbox{$#1$}%
	\vlhyaux{\box\vlhybox}}%
\newcommand{\vldr}[3][8]{\vltr{#2}{#3}{\vlhyE{}}{\vlhyE{\hspace{#1pt} }}{\vlhyE{}}}

\newcommand{\pll}{\mathsf{PLL}_{2\ell}}
\newcommand{\nwpll}{\mathsf{PLL}_{2\ell}^\infty}
\newcommand{\nupll}{\mathsf{wr}\mathsf{PLL}_{2\ell}^\infty}
\newcommand{\cpll}{\mathsf{r}\mathsf{PLL}_{2\ell}^\infty}
\newcommand{\dpll}{\mathsf{nuPLL}_{2\ell}}
\newcommand{\pta}{\mathsf{PTA}_{2\ell}}
\newcommand{\typestream}{\mathsf{nuPTA}_{2\ell}}

\newcommand{\mell}{\mathsf{MELL}_2}

\newcommand{\pllprop}{\mathsf{PLL}}
\newcommand{\nwpllprop}{\mathsf{PLL}^\infty}

\def\FElity{finite expandability\xspace}
\def\FE{finitely expandable\xspace}
\def\octhread{$\oc$-thread\xspace}

\def\prog{progressing\xspace}

\newcommand{\Nodes}{\mathcal{V}}

\newcommand{\Boolset}{\mathbb{B}}
\newcommand{\Nset}{\mathbb{N}}

\def\cproj#1{\cprojp{\left(#1\right)}}
\def\fproj#1{\fprojp{\left(#1\right)}}
\newcommand{\cprojp}[1]{#1^\bullet}
\newcommand{\fprojp}[1]{#1^\circ}

\newcommand{\size}[1]{|#1|}

\def\depth#1{\mathbf{d}(#1)}
\def\height#1{\mathsf{h}({#1})}

\def\cutelim{\to_{\cutr}}
\def\cutelims{\cutelim^*}

\def\comp#1{\mathcal{K}(#1)}

\newcommand{\fptime}{\mathbf{FP}}
\newcommand{\ptime}{\mathbf{P}}
\newcommand{\fppoly}{\fptime/\mathsf{poly}}
\newcommand{\ppoly}{\ptime/\mathsf{poly}}

\newcommand{\lpoly}{\mathbf{L}/\mathsf{poly}}
\newcommand{\nl}{\mathbf{NL}}
\newcommand{\logspace}{\mathbf{L}}
\newcommand{\conl}{\mathbf{coNL}}
\newcommand{\f}{\mathbf{F}}
\newcommand{\elementary}{\mathbf{ELEMENTARY}}
\newcommand{\felementary}{\f\elementary}

\newcommand{\cod}[1]{\underline{#1}}

\newcommand{\true}{\mathbf{1}}
\newcommand{\false}{\mathbf{0}}

\newcommand{\term}{\Lambda_{\mathsf{stream}}}
\newcommand{\identity}{\mathsf{I}}
\newcommand{\letid}[2]{\mathsf{let\ }\identity=\, {#1}  \mathsf{ \ in \ }{#2}}
\newcommand{\lettensor}[4]{\mathsf{let\ } {#2}\otimes {#3} ={#1} \mathsf{ \ in \ }{#4}}
	
\newcommand{\popbool}{\mathsf{pop}}
\newcommand{\erasebool}{\mathsf{disc}}

\newcommand{\ir}[1]{{#1}_i}
\newcommand{\er}[1]{{#1}_e}
\newcommand{\stream}{\mathsf{stream}}

\newcommand{\erase}{\mathbf{W_\Bool}}
\newcommand{\erasetype}[1]{\mathbf{W}_{#1}}
\newcommand{\dupl}{\mathbf{C_\Bool}}
\newcommand{\success}{\mathsf{succ}}
\newcommand{\mult}{\mathsf{mult}}
\newcommand{\add}{\mathsf{add}}
\newcommand{\itrule}{\mathsf{iter}_\Nat}
\newcommand{\itrulestring}{\mathsf{iter}_\String}
\newcommand{\iter}[3]{\mathsf{iter}_\Nat\, {#1}\, {#2}\, {#3}}
\newcommand{\iterstring}[4]{\mathsf{iter}_\String\, {#1}\, {#2}\, {#3}}
\newcommand{\length}{\mathsf{length}}

\newcommand{\RR}{\mathbb R}

\def\cutstep#1#2{#1\mbox{-vs-}#2}

\newcommand{\nwbox}{\mathacronym{nwb}\xspace} 

\newcommand{\opll}{\mathsf{o}\nwpll}

\def\seq#1{(#1)}
\newcommand{\bord}[1]{\mathsf{border}(#1)}

\newcommand{\nwboxes}{\mathacronym{nwb}s\xspace}

\def\sysX{\mathsf{X}}

\defedgetype{T}{draw, very thick, gray, opacity=.5}{}
\defedgetype{dT}{draw, very thick, gray, opacity=.5, densely dotted}{}

\newcommand{\tm}{\mathbf{TM}}
\newcommand{\init}{\mathsf{init}}
\newcommand{\idtype}{\mathbf{ID}}
\newcommand{\decomp}{\mathsf{dec}}
\renewcommand{\comp}{\mathsf{comp}}
\newcommand{\tr}{\mathsf{Tr}}
\newcommand{\In}{\mathsf{In}}
\newcommand{\extr}{\mathsf{Ext}}
\newcommand{\down}[1]{\mathsf{down}^{#1}_\Nat}
\newcommand{\dows}[1]{\mathsf{down}^{#1}_\String}
\newcommand{\degree}[1]{\delta(#1)}